%% file: main.tex
\documentclass{article}

\usepackage{fullpage}
\usepackage{mymacros}

\usepackage{authblk}
\author[1]{Sijing Tu\footnote{\url{sijing@kth.se}}}
\author[1]{Stefan Neumann\footnote{\url{neum@kth.se}}}
\affil[1]{KTH Royal Institute of Technology, Stockholm, Sweden}

\date{}

\begin{document}

\title{A Viral Marketing-Based Model For Opinion Dynamics in Online Social Networks}

\maketitle 

\begin{abstract}
\noindent
Online social networks provide a medium for citizens to form opinions on different societal issues,
and a forum for public discussion.
They also expose users to viral content, such as breaking news articles.  
In this paper, we study the interplay between these two aspects: 
\emph{opinion formation} and \emph{information cascades} in online social networks.
We present a new model that allows us to quantify how users change their opinion
as they are exposed to viral content.
Our model is a combination of the popular Friedkin--Johnsen model for opinion dynamics and 
the independent cascade model for information propagation. 
We present algorithms for simulating our model, and
we provide approximation algorithms for optimizing certain network indices, 
such as the sum of user	opinions or the disagreement--controversy index; 
our approach can be used to obtain insights into how much viral content 
can increase these indices in online	social networks. 
Finally, we evaluate our model on real-world datasets. We show experimentally
that marketing campaigns and polarizing contents have vastly different effects on
the network: while the former have only limited effect on the polarization in
the network, the latter can increase the polarization up to 59\% even when only
0.5\% of the users start sharing a polarizing content. We believe that this finding sheds
some light into the growing segregation in today's online media.
\end{abstract}

\section{Introduction}

Online social networks are a ubiquitous part of modern societies. 
In addition to connecting users with their friends, many people also use them as
content aggregators, by following media outlets or reading articles shared by their peers.  
Clearly, engaging in social networks may impact one's opinions with respect to societal issues: 
users might adjust their opinions during a discussion based on arguments by their peers; 
or they might adapt their opinions based on new facts revealed in a news article they read.

Due to this strong connection between opinion formation and information spread,
online social networks have become the target of viral disinformation campaigns.
Popular examples include groups like QAnon who spread conspiracy theories and
fake news about topics, such as vaccination, or state actors who try to
influence election results in opposing countries.  While it is well-researched
how viral content spreads through social networks, such models do not consider
how user opinions are impacted by the viral content.  Therefore, understanding
how new information influences user opinions and being able to quantify the
impact of such disinformation campaigns is highly desirable.

A prominent model to quantify
opinion dynamics in social networks 
is the \FJfull~(FJ) model~\cite{friedkin1990social}. 
The FJ model stipulates that each user has an \emph{expressed opinion} 
that the user reveals publicly and is network-dependent,  
and an \emph{innate opinion} that is fixed and network-independent.  
However, it does not take into account how users
change their opinions based on new information (e.g., viral content)
that is disseminated in the network.

Furthermore, researchers have studied problems related to optimizing certain
opinion-based network indices, for instance, maximizing the average
opinion~\cite{gionis2013opinion} or
polarization~\cite{chen2020network,gaitonde2020adversarial}, 
or minimizing polarization and disagreement~\cite{chen2018quantifying,matakos2017measuring,musco2018minimizing}. 
In this line of work, optimization occurs by nudging the expressed or innate opinions 
of a set of seed users towards a certain direction. 
However, existing works do not specify how such nudging takes place, 
nor do they consider the interplay of opinion nudging within a more realistic setting of information~cascades.

Therefore, the current research 
\emph{either} allows us to quantify user opinions and optimize opinion-based
network indices
without taking into account viral content 
\emph{or} it allows to assess the spread of viral content without reasoning about user opinions.
This limitation leads us to the following questions:
\begin{enumerate}
\item \emph{Can we quantify how viral content influences user opinions in online social networks?}
\item \emph{Can we study the interplay between information cascades and opinion dynamics?}
\item \emph{Can we optimize opinion-based network indices by taking into account the spread of viral content?}
\end{enumerate}
\para{Our contributions.}
We answer the above questions affirmatively by proposing a new model
that combines the \FJfull model~\cite{friedkin1990social} and 
the influence-maximization framework of Kempe et al.~\cite{kempe2015maximizing}.  
To the best of our knowledge, 
our model is the first that allows to quantify the impact of viral content on user opinions.

Contrasted with the vanilla FJ model in which the innate user opinions are ``fixed'' 
but the expressed opinions are changing over time based on user interactions, 
our model considers the viral content that is shared in the network, 
and it assumes that for users who are exposed to this  content,
\emph{there is a probability that their innate opinion changes}.
This could be the case, for example, when a user reads an article that makes them reconsider
their stance on a certain topic.  

When a subset of users change their innate opinions, 
their expressed opinions will also be modified, 
which in turn will have an impact on 
the whole network via the FJ opinion dynamics.
Thus, \emph{the change of the innate opinions of few users may have an
impact on the whole network}: 
even when a user's innate opinion does not change by the viral content 
(because they ignore it or the content never reaches them), 
they still might change their expressed opinion due to ``peer-pressure.''

Our model connects these two phenomena: 
it allows us to understand how viral content can impact individual users,
while it also enables us to study how individual behavior ripples through the network and
affects the overall discussion.

We consider two different types of content: \emph{non-polarizing} and \emph{polarizing}. 
For non-polarizing content, such as marketing campaigns, the innate opinions of the users can only increase.
For polarizing content, we take into account the \emph{backfire~effect} \cite{nyhan2010corrections}: 
interaction with opposing content may lead to a decrease in a user's innate opinion.
This could be the case, e.g., in political campaigns when a party runs an ad
that makes their supporters react positively but their opponents react negatively.

From an algorithmic view point, we present methods for simulating our model.
Additionally, we consider the problem of optimizing certain opinion-based
network indices.
We present a greedy $(1-1/e-\varepsilon)$-approximation algorithm for maximizing 
the sum of user opinions for non-polarizing content. 
We also present algorithms for maximizing the controversy and 
the disagreement--controversy indices~\cite{musco2018minimizing} for non-polarizing content; 
our algorithms have data-dependent approximation ratios. 
Finally, we provide heuristics
for maximizing other indices, such as polarization and disagreement, 
for non-polarizing and polarizing~content.

To obtain our optimization algorithms, we build upon the reverse-reachable sets
framework~\cite{borgs2014maximizing,tang2014influence,tang2015influence}.  
One challenge is that, in our setting, the arising optimization problems are based
on quadratic forms and, therefore,  we have to extend the reverse reachable set
framework to this more general setting.  

We evaluate our methods on real-world data. Our
experiments reveal a striking difference between non-polarizing and polarizing content. 
On the one hand, non-polarizing content can significantly increase
the sum of user opinions, %
but it has limited impact on the polarization %
and sometimes even \emph{decreases} it.  
The situation for polarizing content is the opposite: 
it barely increases the sum of user opinions but it can
increase the polarization significantly. 
We see that even when only 0.5\% of the users start sharing a polarizing content, 
the network polarization increases by more than 20\% on average and can rise up to 59\%.
We believe that this finding provides an explanation for the growing polarization that
can be witnessed in modern day's online media.  

We present the proofs of our claims in the appendix.

\section{Related work}
\label{sec:related}

Our aim %
is to quantify how viral content impacts user opinions in social networks.
Naturally, our approach builds on existing models for opinion dynamics and information~cascades.

Opinion dynamics have been studied in different research areas, including
psychology, social sciences, and
economics~\cite{castellano2009statistical,jackson2008social}.  Here we build on
the popular Friedkin--Johnsen (FJ) model~\cite{friedkin1990social}, which is an
extension of a classic model by DeGroot~\cite{degroot1974reaching}.  Many
extensions of the FJ model have been proposed. For example, Amelkin et
al.~\cite{amelkin17polar} assume that the innate user opinions change over
time based on the expressed opinions. We refer to the discussion
in~\cite{amelkin17polar} for other related models. However, these works do not
take into account the changes of innate opinions based on exposition to viral
contents.

Recent work used these models for understanding properties 
of opinion dynamics and formulating natural optimization problems.
Bindel et al.~\cite{bindel2015bad} analyze the ``price of anarchy'' in the FJ model
by considering as cost the internal conflict of the individuals in the network
and comparing the cost at equilibrium and the social optimum. 
Gionis et al.~\cite{gionis2013opinion} maximize the sum of opinions
of the network users. 
Other works study the problem of measuring and reducing polarization of opinions, 
or other disagreement indices, in the FJ model~\cite{chitra2019understanding,matakos2017measuring, musco2018minimizing,zhu2021minimizing}, 
while adversarial settings have also been considered, aiming to quantify the power of an adversary 
seeking to induce discord in the
model~\cite{chen2018quantifying,chen2020network,gaitonde2020adversarial}.

To model information cascades, we build on the popular in\-de\-pend\-ent-cascade
model of Kempe et al.~\cite{kempe2015maximizing}.  Many extensions and variants
of this model have been proposed. For example, Sathanur et
al.~\cite{sathanur18exploring} incorporate intrinsic user activations based on
external sources. 
Another popular extension is the topic-aware cascade model
by Barbieri et al.~\cite{barbieri2013topic}, and has various applications
including social advertising~\cite{aslay2015viral,aslay2017revenue}.
While such models allow to model information spread, they do
not allow to quantify how these change the user opinions.

The backfire effect, 
the tendency of individuals to hold firmly on their beliefs when faced with
factual corrections, has been observed in political sciences~\cite{nyhan2010corrections}, 
but has not been studied extensively in computational social sciences.
Exceptions are the works of Chen et al.~\cite{chen19opinion}, 
who incorporate backfire in an opinion-dynamics model for biased assimilation~\cite{dandekar2013biased}, 
and Hirakura et al.~\cite{hirakura2020model}, 
who propose a model of polarization that incorporates empathy and repulsion.

Our optimization algorithms rely on reverse reachable sets, introduced by Borgs
et al.~\cite{borgs2014maximizing}, and improved by subsequent
techniques~\cite{tang2014influence,tang2015influence}.  We extend these ideas to
our setting, to obtain algorithms for objectives that include quadratic terms.
We note that the activity-maximization task defined by Wang et al.~\cite{wang2017activity} 
is a special case of our setting. 
We apply the ``sandwich'' framework~\cite{lu2015competition} to obtain 
data-dependent approximation guarantees for some of our objectives. 
For the efficient computation of our objective functions, we use the methods by
Xu et al.~\cite{xu2021fast} based on Laplacian solvers.

To our knowledge, this is the first work that studies how user opinions change
due to viral information in online social networks.
For fixed user opinions, Monti et al.~\cite{monti2021learning} studied how
cascades spread through the network, based on the user opinions and the topics
of the contents.

\section{Preliminaries}
\label{sec:preliminaries}

Let $\graph = (V, E, w)$ be an un\-directed weighted graph, with $n=\abs{V}$
nodes and edge weights $w \colon E \rightarrow \Real_{>0}$.  
We let $N(u)$ denote the set of neighbors of $u\in V$.
The Laplacian of $\graph$ is $\laplacian = \m+D - \m+W$,
where $\m+D$ is the $n\times n$ diagonal
matrix with $\m+D_{u,u} = \sum_{v\in N(u)}w(u,v)$ for all $u\in V$
and $\m+W$ is the $n\times n$ matrix with
$\m+W_{u,v} = w_{u,v}$ for all $u,v\in V$. 

\spara{\FJfull (FJ) model.}
In the FJ model, we are given a weighted un\-directed graph $G=(V,E,w)$ with $n$ nodes. 
Each node~$u$ %
corresponds to a user of a social network. 
Each user $u$ has an \emph{expressed} opinion $z_u\in[0,1]$, which depends on the network, 
and a fixed \emph{innate} opinion $s_u\in[0,1]$. %
We write $\begop\in[0,1]^n$ and $\finop\in[0,1]^n$ to denote the vectors of innate and expressed opinions.

The expressed opinions are updated in rounds.
More concretely, let $\begop$ be the vector of innate opinions, and
$\finop^{(t)}$ be the vector of expressed opinions at time~$t$. 
The update rule is given by 
\begin{equation}
\label{eq:opinion-round}
\finop^{(t+1)} = (\m+D + \ID)^{-1}(\m+W\finop^{(t)} + \begop).
\end{equation}
Taking the limit $t \rightarrow \infty$, 
the expressed opinions converge to
\begin{equation}
\label{eq:opinion-equilibrium}
  \finop^* = (\ID + \laplacian)^{-1}\begop. 
\end{equation}

\begin{table}[t!]
\begin{center}
\begin{small}
\caption{Matrices of the different indices.}
\label{table:index}
 \begin{tabular}{ lcc } 
 \toprule
 \textbf{Index} & \textbf{Notation} & \textbf{Matrix} \\ 
 \midrule
 Polarization & \MasIdx{\PolIdx{}} & $(\ID + \laplacian)^{-1} (\ID - \frac{\ind \ind^\intercal}{n}) (\ID + \laplacian)^{-1}$ \\ 
 Disagreement & \MasIdx{\DisIdx{}} & $(\laplacian + \ID)^{-1} \laplacian (\laplacian + \ID)^{-1}$ \\ 
 Internal conflict & \MasIdx{\IntIdx{}} & $(\laplacian + \ID)^{-1} \laplacian^2 (\laplacian + \ID)^{-1}$ \\ 
 Controversy & \MasIdx{\ConIdx{}} & $(\laplacian + \ID)^{-2}$ \\ 
 Disagreement--controversy & \MasIdx{\DisConIdx{}} & $(\laplacian + \ID)^{-1}$ \\ 
 \bottomrule
\end{tabular}
\end{small}
\end{center}   
\end{table}

We study the following popular network indices in our
model, where the matrices of the quadratic forms are as defined in
Table~\ref{table:index}:
\begin{itemize}
\item
\emph{sum of user opinions}, which is given by $\SumIdx{\begop} = \ind^{\intercal} \begop$,  
and it is well-known that $\SumIdx{\begop} = \ind^{\intercal} \finop$;
\item 
\emph{polarization~\cite{musco2018minimizing}}  
$\PolIdx{G, \begop} = \sum_{u\in V} (\efinop{u}^* - \bar{\finop})^2 =
\begop^\intercal \MasIdx{\PolIdx{}} \begop$, where
$\bar{\finop} = \frac{1}{n} \sum_{u\in V} \efinop{u}^*$ is the average user opinion;
\item
\emph{disagreement~\cite{musco2018minimizing}}
$\DisIdx{G, \begop} = \sum_{(u,v)\in E} w_{u,v} (\efinop{u}^*-\efinop{v}^*)^2 
= \begop^\intercal \MasIdx{\DisIdx{}} \begop$;
\item
\emph{internal conflict~\cite{chen2018quantifying}}  
$\IntIdx{G, \begop} = \sum_{u\in V} (\ebegop{u} - \efinop{u}^*)^2 = 
\begop^\intercal \MasIdx{\IntIdx{}} \begop$;
\item
\emph{controversy~\cite{chen2018quantifying, matakos2017measuring}} 
$\ConIdx{G, \begop} = \sum_{u\in V} (\efinop{u}^*)^2 = \begop^\intercal \MasIdx{\ConIdx{}} \begop$; 
and
\item
\emph{disagreement-controversy~\cite{xu2021fast, musco2018minimizing}}\! 
$\DisConIdx{G, \begop}\!=\!\begop^\intercal \MasIdx{\DisConIdx{}} \begop\!=\!\ConIdx{G, \begop}\!+\!\DisIdx{G, \begop}$.
\end{itemize}

\section{Modelling the Influence of Viral Content on User Opinions}
\label{sec:model}

We formally introduce our model
in Sec.~\ref{sec:spread} and show how it can be
simulated in Sec.~\ref{sec:equivalence}.

\subsection{The \spread}
\label{sec:spread}

Following the independent cascade model~\cite{kempe2015maximizing},
we assume that %
a value $p_{u,v}\in[0,1]$ encodes the
probability that user~$v$ reacts to content received from user~$u$;
we allow that $p_{u,v}\neq p_{v,u}$.
Furthermore, we introduce parameters $\varepsilon,\delta > 0$, as explained below.

As per the FJ model, each user $u$ has an expressed opinion $z_u$ and an innate opinion $s_u$.
Additionally, each user has a \emph{state}, 
which is either \stinactive, \stignore, \stacknowledge or \stspread. 
We order the states by 
``\stinactive $<$ \stignore $<$ \stacknowledge $<$ \stspread'' 
and we follow the convention that when a user changes their state, 
they can only pick one that is higher with respect to this ordering.
An illustration of the model with respect to state transitioning and 
actions performed for a single node $v$ is provided in Figure~\ref{figure:sam-illustration}.

Our model proceeds in rounds.  
Initially, in round~$0$, there are $k$~users whose state is \stspread and all other users are \stinactive; 
in later rounds, it is possible that users change their state. 
We will refer to the users whose initial state is \stspread as \emph{seed nodes}.   
Each round $t>0$ has two \emph{phases}:\footnote{We note that for our model and our analysis it is not necessary to consider two phases, we only make this assumption for the sake of better exposition.
We could as well assume that both phases are interleaved and happen	simultaneously.}
In the first phase, the users update their expressed opinions. 
In the second phase, the viral content is spread through the network and 
users may change their state and their innate opinion. 
We describe both phases below.

\begin{figure}[t]
\begin{center}
\includegraphics[width=0.7\columnwidth]{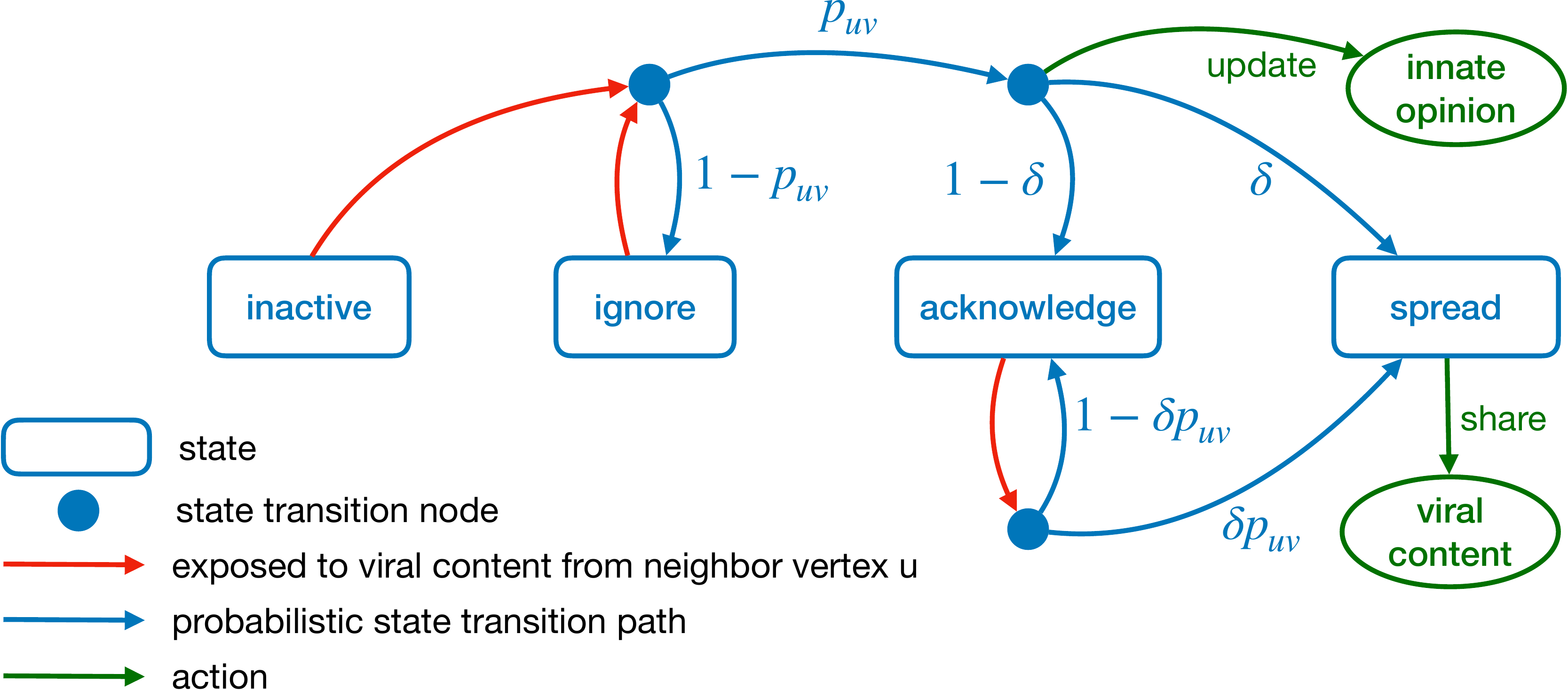}
\caption{\label{figure:sam-illustration}An illustration of the \spread with respect to 
state transitioning and actions performed for a single node $v$.
In the initial round, $k$ seed nodes are in state \stspread, 
while the rest of nodes are in state \stinactive.
}
\end{center}
\end{figure}

\spara{Phase~I: Updating user opinions.} 
The users update their expressed opinion as in the FJ model,
i.e., according Eq.~\eqref{eq:opinion-round}. 

\spara{Phase~II: Information spreading.} 
Consider round $t>0$.
Let $U$ denote the set of users who have changed their state to \stspread in round~$t-1$.
If $U=\emptyset$, we consider Phase~II finished. 
Otherwise ($U\neq\emptyset$), each user $u\in U$ shares the viral content with all of its neighbors.
When a neighbor $v$ of $u$ is exposed to the viral content, 
it switches to a new state and possibly adjusts its innate opinion. 
If $v$ is in state \stinactive or \stignore, then this is done as follows:
\begin{itemize}
	\item 
    With probability $\delta p_{uv}$, user $v$ switches to state \emph{spread};
		$v$~adjusts its innate opinion (described below) and shares the content
		in the next round with its neighbors.
	\item 
    With probability $(1-\delta) p_{uv}$, user $v$ switches to state \emph{acknowledge};
		$v$~adjusts its innate opinion (described below) but does \emph{not}
		share the content with its neighbors in the next round.
	\item 
		With probability $1-p_{uv}$, user $v$ switches to state \emph{ignore};
		$v$~performs no action (i.e., $v$ does not try to share the content and
		$v$ also does not adjust its innate opinion).
\end{itemize}
If $v$ is in state \stacknowledge then it switches to state \stspread with
probability $\delta p_{uv}$ and remains in state \stacknowledge with probability
$1-\delta p_{uv}$. 
Finally, if $v$ is in state \stspread then $v$ always stays
in this state. In both of these cases, $v$ does not adjust its opinion again.

The above process ensures that the state ordering defined before
is obeyed during state switching and that each user adjusts its innate opinion
at most once.  Finally, note that our model is a generalization of the
independent cascade model if $\delta=1$.

\spara{Adjusting innate opinions.}
Now we describe how users change their innate opinions when they are exposed to viral content.

Consider user~$u$ at state \stinactive or \stignore whose new state becomes \stacknowledge or \stspread.  
Then the innate opinion $\begop_u$ changes to a new value $\s_u$. 
We consider two scenarios:
\begin{itemize}
	\item \emph{Marketing campaign:}
		The user's opinion becomes more positive after seeing the content, i.e.,
	   	$\s_u = \min\{\begop_u + \varepsilon, 1\}$ for the parameter
		$\varepsilon > 0$ from above. Here, we use the $\min$-operation to ensure
		that the new opinion $\s_u$ is in the interval $[0,1]$.
	\item \emph{Polarizing campaign with backfire:}
		In a polarizing campaign we assume that while some users embrace the
		content, others will find it repelling.
    More concretely, we assume that there is a threshold $\tau\in[0,1]$ such that:
		(1)~If $\begop_u \geq \tau$ then $u$ embraces the content and adjusts
		its opinion to $\s_u = \min\{\begop_u + \varepsilon, 1\}$;
		(2)~If $\begop_u < \tau$, then $u$ finds the content repelling and
		adjusts $\s_u = \max\{0, \begop_u - \varepsilon\}$. 
\end{itemize}

Note that in our model with polarizing campaigns, users can still share a
content they dislike.  While this might seem non-intuitive at first, we believe
that it is a realistic behavior: users who oppose a certain content often share
it together with a counter-argument.  We remark that our model can be
modified to avoid this.

Finally, observe that $\s$ is a random vector that depends on the outcome of
the information spread. However, once we fix the
randomness of the information spread, $\s$ becomes
deterministic. This will be a useful property in the following.

\spara{Possible model extensions.} We note that our model is quite general and
can be extended in various ways. 
First, we modelled information cascades via the independent cascade
model~\cite{kempe2015maximizing}. However, our model and our result from
Lemma~\ref{lem:equivalence} also hold if we used the linear threshold
model~\cite{kempe2015maximizing}, topic-aware versions of the
independent cascade and linear threshold models~\cite{barbieri2013topic}, as well
as intrinsic user activations~\cite{sathanur18exploring}.
In particular, using the linear threshold model could lead to insights on
contents that spread via complex contagion~\cite{centola2007complex,guilbeault2018complex}.
Second, above we considered the two relatively simple settings for adjusting the
innate user opinions $\s_u$. However, we note that Lemma~\ref{lem:equivalence}
below generalizes to the setting when $\s_u$ is any user-defined function of
$\begop_u$.

\subsection{Equivalence with the two-stage model}
\label{sec:equivalence}

While the \spread is easy to explain and motivate by real-world scenarios, 
it is not clear how to simulate it efficiently. 
If we implemented the model as described above, we would have to update the expressed opinions 
in each round, which can be costly. 
To avoid this, we now introduce a new model that can be
simulated more efficiently, and we show that it produces an identical distribution over
the innate and expressed opinions.

\spara{The two-stage model.}
Our simplified model also proceeds in rounds, but it performs the information
spreading and the updating of the user opinions in two sequential \emph{stages}. More
concretely, in each round of the \emph{first stage}, we perform the information
spreading process that is described in Phase~II above 
(and we do \emph{not} perform the updating of the expressed opinions as per Phase~I). 
In this process, some of the users' innate opinions and their states might change. 
When after a round no new users have changed their state to \stspread, we start
the second stage. 
In each round of the \emph{second stage}, we perform the same update of the
expressed opinions as described in Phase~I above (and we do \emph{not} run Phase~II).

\spara{Simulating the two-stage model.}
Next, we discuss why the two-stage model is well-suited for efficient
simulations.
First, observe that the first stage stops when no node changed their state to
\stspread in the previous round, i.e., when $U=\emptyset$.  Therefore, the first
stage can have at most $O(n)$~rounds (since each of
the $n$ users can take at most four different states and since we assumed that
users only increase their state with respect to the ordering of the states).
Additionally, in each round of the first stage, we can update the states of the
nodes by iterating over all nodes $v$ that are neighbors of a node
$u\in U$ and then updating the state of $v$ with the probability described in
Phase~II. Since each user can become a spreader only once, the time for
executing \emph{all} rounds of the first stage is $O(m)$, where $m$ is the
number of edges in the graph.

Second, recall that the adjusted innate opinions $\s_u$ only depend on the
randomness from the information spreading process. %
Therefore, after the first stage %
finished, the innate opinions $\s_u$ are 
fixed. Thus, %
we can assume that the vector $\s_u$ is known and the
expressed equilibrium opinions are given by $\z^* = (I+L)^{-1} \s$.
The time complexity for the second stage is the time required to solve for $\z^*$.

\spara{Equivalence of the opinion distributions.}
It remains to show that both models induce the same
distribution over the innate and expressed opinions. 
To show this equivalence, we assume that both models are run with the same input graphs, 
the same seed nodes, and the same (non-adjusted) innate opinions $\begop$. 
Now let us denote the adjusted innate opinions generated by the \spread by $\s_u$ and
those by the two-stage model by $\ss_u$. 
Recall that both $\s_u$ and $\ss_u$ are random vectors that depend 
only on the outcome of the information-spreading process. 
The following lemma asserts 
the equivalence of the two models.
The proof is presented in Appendix~\ref{sec:proof-equivalence}.
\begin{lemma}
\label{lem:equivalence}
	For all $a\in[0,1]^n$, $\Pr[\s = a] = \Pr[\ss = a]$.
	Furthermore, let $\z^*=(I+L)^{-1} \s$ and $\zz^*=(I+L)^{-1} \ss$ be the
	equilibrium opinions. 
	Then $\Pr[\z^* = b] = \Pr[\zz^* = b]$ for all $b\in[0,1]^n$.
\end{lemma}

\section{Algorithms}
\label{sec:algorithms}

We now present algorithms for maximizing the indices defined in Sec.~\ref{sec:preliminaries}.
We start with algorithms for approximating the indices
(Sec.~\ref{sec:estimating}).  Then we present our algorithms for maximizing the
sum of user opinions (Sec.~\ref{sec:sum-index}) and for maximizing the
controversy and the disagreement--controversy indices (Sec.~\ref{sec:dis-con}).
We present our proofs in Appendix~\ref{sec:omitted}.

\subsection{Estimating indices}
\label{sec:estimating}

Let $\MasIdx{\mathcal{M}}$ be one of the matrices from Table~\ref{table:index},
which induces the quadratic form for each of the indices that we wish to study. 
Recall that $\begop$ is the non-adjusted vector of innate opinions
and $\s$ is the random vector of adjusted innate opinions.
In the following, our goal is to compute $\Exp[\s^\intercal \MasIdx{\mathcal{M}} \s]$.

Let $\diffvec=\s-\begop$ be the random vector that denotes how the users
changed their opinions. Then observe that
\begin{align*}
	\Exp[\s^\intercal \MasIdx{\mathcal{M}} \s]
	= 
	\begop^\intercal \MasIdx{\mathcal{M}} \begop
	+\Exp[2\begop^\intercal \MasIdx{\mathcal{M}} \diffvec + \diffvec^\intercal \MasIdx{\mathcal{M}} \diffvec].
\end{align*}
Since the first term in the sum is deterministic, we drop it and 
focus on $\Exp[\obj]$, where
$\obj := 2\begop^\intercal \MasIdx{\mathcal{M}} \diffvec + \diffvec^\intercal \MasIdx{\mathcal{M}} \diffvec$.
We show that computing $\Exp[\obj]$ is \SPhard since
our model generalizes the independent cascade model.
\begin{lemma}
\label{lemma:hard-computing}
	Given seed nodes $S$, computing $\Exp[\obj]$ is \SPhard.
\end{lemma}

\spara{Monte Carlo Simulation.}
Since Lemma~\ref{lemma:hard-computing} shows that computing $\Exp[\obj]$ exactly
is hard, we resort to approximations. One option is to use Monte Carlo
simulations of our model.  More concretely, we can simulate our model as
described in Sec.~\ref{sec:equivalence} to obtain multiple samples of
$\s$. Now a Chernoff bound implies that we can compute an approximation of
$\Exp[\s]$ with high probability. Then we can compute an approximation of
$\Exp[\obj]$ in near-linear time using the algorithms by Xu et
al.~\cite{xu2021fast}, which are based on Laplacian solvers.  This approach is
efficient when the number of seed node sets for which we wish to compute
$\Exp[\obj]$ is \emph{small}.

\spara{Reverse reachable sets.}
However, in our optimization algorithms we will need to evaluate $\Exp[\obj]$
for a \emph{large} number of different seed node sets and thus using the Monte
Carlo approach is too inefficient.  Therefore, we use 
\emph{reverse influence sampling}~\cite{borgs2014maximizing,tang2014influence,tang2015influence},
which allow us to reduce the number of simulations of our model.

Our notion of reverse reachable sets is as follows.
Suppose that we want to simulate our model on a graph $G=(V,E)$.
A \emph{possible world}
is a copy of $G$ that has labels on the edges and we generate the labels as
follows. For each edge
$(u,v)\in E$, we pretend that $u$ has
state \stspread and $v$ has state \stinactive. Now we sample the state of $v$ as
described in Phase~II above and we \emph{label} $(u,v)$ with the new
state of $v$. For example, if $v$ changes its state to \stacknowledge then the
label of $(u,v)$ is \stacknowledge. This process is repeated for all
edges $(u,v)\in E$ and we always assume that $u$ has state \stspread
and $v$ has state \stinactive, regardless of the outcomes of previous samples.

Now consider a possible world $g$. We say that there exists a \emph{live
path} from $u$ to $v$ if there exists a path in $g$ in which all edges have
label \stspread except the edge incident upon $v$ which may have label \stacknowledge
or \stspread.  Notice that live paths encode when users change their opinions in
our model: user $v$ adjusts its opinion if and only if there exists a live path
from a seed node to~$v$.

Next, let $g$ be a randomly generated possible world and let $u$
be a random node in $G$.  A \emph{random \rr-set $R$ for $u$ in $g$} is a
set of nodes in $g$ such that there exists a live path to~$u$. 

\spara{Estimating indices.}
Now we turn to estimating $\Exp[\obj]$.  Existing information propagation
methods can be used for
estimating $2\Exp[\begop^\intercal \MasIdx{\mathcal{M}} \diffvec]$,
because $\diffvec$ is the only random quantity in this expression. However,
we also need to approximate
$\Exp[\diffvec^\intercal \MasIdx{\mathcal{M}} \diffvec]
= \sum_{u,v} \MasIdx{\mathcal{M}}_{u,v} \Exp[\diffvec_u\diffvec_v]$, 
which involves products of random variables and which existing methods cannot
do.  Our main observation is that in each possible world it holds that
$\diffvec_u\diffvec_v\neq0$ if and only if there exist live paths from the seed
nodes to $u$ \emph{and} $v$.  Wang et al.~\cite{wang2017activity} followed a
similar approach but only considered pairs $(u,v)$ for which there
exist edges in the graph; here, we have to perform this operation for all
pairs $(u,v)\in V^2$.%

In the following, we set $\Delta\ebegop{u} \in [-\varepsilon,\varepsilon]$ to
denote how much user $u$ adjusts its opinion once it reaches state
\stacknowledge or \stspread.
Note that $\abs{\Delta\ebegop{u}}$ can be smaller than $\varepsilon$ because
of the interval concatenation that we described in Phase~II above.
Next, let $S$ be a set of seed nodes and let $\ind_u(S)$ be an indicator with
$\ind_u(S) = 1$ if $u$ adjusts innate opinion and otherwise $\ind_u(S) = 0$.
Let $\ind(S)$ be a vector of $\ind_u(S)$ consisting of each $u \in V$.
Note that $\ind(S)$ is a random vector and that $\Delta\begop$ is deterministic.
Observe that $\diffvec = \Delta\begop \odot \ind(S)$, where is $\odot$ the Hadamard product.

To simplify our notation, we set
$\LinGain{u} = (2\begop^\intercal \MasIdx{\mathcal{M}})_u \Delta \ebegop{u}$ and
let $\PolGain{u}{v}=(\Delta \ebegop{u})^\intercal \MasIdx{\mathcal{M}}_{u,v} \Delta
\ebegop{v}$.
Then we obtain:
\begin{align*}
\obj = \sum_{u, v\in V}\frac{1}{n}\LinGain{u}\ind_u(S) +
\PolGain{u}{v}\ind_u(S)\ind_v(S) =: F(S).
\end{align*}

Given these definitions, we let $R_u$ and $R_v$ be random \rr-sets for $u$ and
$v$, respectively, and we set
$$\estimator = \IND[(\randomrrp_u \cap S) \neq \emptyset] w_u + \IND[(\randomrrp_u \cap S) \neq \emptyset, (\randomrrp_v \cap S) \neq \emptyset]n \, m_{u, v},$$
and for a set $\c+R$ of random \rr-sets we define
\begin{align}
\label{eq:FR}
	F_{\c+R}(S) = \frac{\sum_{(\randomrrp_u, \randomrrp_v) \in \c+R} \estimator}{\abs{\c+R}}.
\end{align}

We show that $F_{\c+R}(S)$ is an unbiased estimator for $\Exp[F(S)]$.
\begin{lemma}
\label{lemma:fr-1}
Let $\c+R$ be a set of samples of pair of random \rr-sets.
Then $\Exp[F(S)] = \set+E_{u, g \sim \c+G} [n \, F_{\c+R}(S)]$.
\end{lemma}

Since the previous lemma only holds in expectation, we now consider
approximations that hold with high probability.
Let $\ell>0$ be an error parameter,
	$\theta=\abs{\c+R}$ be the number of \rr-sets and suppose
we know $\optvalue = \max_{\abs{S} \leq k} \Exp[F(S)]$ (we show later
how to obtain bounds on $\optvalue$ using statistical tests).
Our goal will be to pick $\theta$ large enough such that 
\begin{equation}
   \label{equation:sample-error-bound}
 \Pr[\abs{n\, \fracrrp{S} - \Exp[F(S)]} \geq \frac{\epsilon}{2}\, \optvalue] \leq \frac{1}{n^{\ell} \binom{n}{k}},
\end{equation}
since then a union bound implies that for any seed set $S$ of size $k$,
$\Exp[F(S)]$ is a good estimator for $\fracrrp{S}$ w.h.p.
We show that if we pick $\theta$ large enough then
Equation~\eqref{equation:sample-error-bound} is satisfied.
\begin{lemma}
\label{lem:approximation}
   Let $\chi=\max_{u,v\in V} | w_u + n \, m_{u,v}|$ and
   $\lambda = \frac{8n \chi}{\epsilon^2}(\frac{\epsilon}{3} +1) (\ell \ln n + \ln 2 + \ln \binom{n}{k})$.
   If $\theta \geq \frac{\lambda}{\optvalue}$ then
   Equation~\eqref{equation:sample-error-bound} holds. 
\end{lemma}

\subsection{Maximizing network indices}
\label{sec:sum-index}
Now we consider the \emph{sum of expressed opinions problem}, where we are given
an undirected weighted graph $G=(V,E)$ with edge probabilities $p_{u,v}$ and a
positive integer $k$.  The goal is to find a set of seed nodes of cardinality at
most $k$ that maximizes the sum of expressed opinions
$\Exp[\SumIdx{\s}] = \Exp[\ind^{\intercal} \z^*]$.
Our main result is as follows.
\begin{theorem}
\label{thm:greedy-marketing}
	There exists a greedy approximation algorithm that computes a
	$(1-1/e-\varepsilon)$-approximation with high probability.
\end{theorem}

Indeed, in Appendix~\ref{sec:comparison} we show that our model is strictly more powerful than
the FJ~model in which we can increase $k$~innate user opinions.

\begin{algorithm2e}[t!]
\caption{$\rrGreedy$}
\label{algo:rrgreedy}
	{\small
	\Input{$\rrpsample$, $k$}
	\Output{{$\approxsoln$}}
	}
	\BlankLine
	{\emph{$\approxsoln \leftarrow \emptyset$}} \\
	\While{$\abs{\approxsoln} \leq k $}{
	{   
	   \emph{$x \gets \arg\max_{x}\fracrrp{\approxsoln \cup \{x\}} - \fracrrp{\approxsoln}$;} \\
		\emph{$\approxsoln \gets \approxsoln \cup \{x\}$;}}
	}
	\KwRet{{$\approxsoln$}}
\end{algorithm2e}

\begin{algorithm2e}[t!]
\caption{$\sampling$}
\label{algo:sampling}
	{\small
	\Input{$\multiGraph$, $\lambda$, $\beta$, $\epsilon_2$, $k$, $\Delta \begop$, $\chi$, $\LB_0$}
	\Output{$\sampleR$}
	}
	{\small
	$\c+R \gets \emptyset$, $\LB \gets \LB_0$\;
	}
	\BlankLine
	\For{$i =1, \ldots, \log_2 n - 1$} {
		$y \gets n / 2^{i}$, $\theta_i \gets \frac{\beta}{y}$\;
		\lWhile{$\lvert \sampleR \rvert \le \theta_i$}{
			$\sampleR \gets \sampleR \cup \mathrm{GenerateRR\text{-}Set}$}
		$\t+X_i \gets \rrGreedy(\sampleR, k)$\;
		\lIf{$n \, \fracrrp{\t+X_i} \ge (1 + \epsilon_2) \, y \chi,$}{
			$\LB \gets \frac{n \, \fracrrp{\t+X_i}}{1 + \epsilon_2}$, break}
	}
	$\theta \gets \lambda /\LB$\;
	\lWhile{$\lvert \sampleR \rvert \le \theta$}{
		$\sampleR \gets \sampleR \cup \mathrm{GenerateRR\text{-}Set}$}
	{\bf Return} $\sampleR$\;
\end{algorithm2e} 

To obtain the theorem, we maximize the sum of the adjusted parts of the innate
opinions $\Exp[\sum_u \diffvec_u]$, since it is well-known that
$\sum_u \z^*_u=\sum_u \s_u$ and thus we can maximize
$\Exp[\sum_u \diffvec_u]$. Equivalently, we can maximize
$F(S) := \sum_{u \in V} \ind_u(S) \Delta \ebegop{u}$, as we show next.
\begin{lemma}
\label{lemma:max-sum-opinions} 
    $\arg\max_S \sum_{u \in V} \efinop{u}^*(S)  = \arg\max_S F(S)$
\end{lemma}

The main benefit of Lemma~\ref{lemma:max-sum-opinions} is that to maximize
$F(S)$, we do not have to compute the sparse matrix inverse from
Equation~\eqref{eq:opinion-equilibrium} which is very costly. 
Note that if $\Delta \ebegop{u}=\varepsilon$ for all $u \in V$, maximizing 
$F(S)$ reduces to the influence maximization problem~\cite{kempe2015maximizing}.
However, if $\Delta \ebegop{u}<\varepsilon$ for some $u$, the solutions might
differ. The approximation result from the theorem stems from the following
lemma.
\begin{lemma}
  \label{lemma:approx-max-sum-opinions}
  The function $\Exp[F(\cdot)]$ is monotone and submodular. Thus the greedy algorithm 
  achieves an approximation ratio of $1 - \frac{1}{e}$. 
\end{lemma}

\spara{Maximizing network indices. }
To estimate $F(S)$, we define $F_{\c+R}(S)$
similar to Equation~\eqref{eq:FR}. 
The difference is that we drop the quadratic terms $\IND[(\randomrrp_u \cap S) \neq \emptyset, (\randomrrp_v \cap S) \neq \emptyset]n \, m_{u, v}$, 
and we set $w_u = \Delta \ebegop{u}$. 

Our algorithm works as follows.  We sample a set $\c+R$ of \rr-sets and greedily
pick the nodes that maximize $F_{\c+R}(S)$.  The algorithm keeps on adding
\rr-sets to $\c+R$ until a statistical test asserts that we have found a lower
bound on \optvalue.  More concretely, we keep on sampling if the value estimated
by $n F_{\c+R}(S)$ is \emph{not} a lower bound on \optvalue (see (1)
in Lemma~\ref{lem:y-size-main}) and when we stop sampling then we obtain a good
enough lower bound $\LB$ (see (2) in Lemma~\ref{lem:y-size-main}).
Then we can apply Lemma~\ref{lem:approximation} with $\theta \geq \lambda / \LB$
to obtain our approximation guarantees.  We present the pseudocode including the
sampling in Algorithm~\ref{algo:sampling} and the greedy subroutine in
Algorithm~\ref{algo:rrgreedy}. We run our algorithms with parameters
$\LB_0 = \max_u \abs{\Delta \ebegop{u}}$ and
$\beta = n(\frac{4}{3}\epsilon_2 + 2)(l\ln n + \ln \log_2 2n + \ln \binom{n}{k})/\epsilon_2^2$. 

\begin{lemma}
\label{lem:y-size-main}
  Let $\t+X$ be the output of Algorithm~\ref{algo:rrgreedy} and
  suppose that $|\c+R|= \theta \geq \frac{\beta}{y}$.
  Then with probability at least $1 - \frac{n^{-\ell}}{\log_2 (n)}$:
  (1)~If $\opt < y \chi$, then $n \,  \fracrrp{\t+X} < (1 + \epsilon_2)y \chi$.
  (2)~If $\optvalue \geq y \chi$, then $n \, \fracrrp{\t+X}\leq (1 + \epsilon_2)\optvalue$.
\end{lemma}

The above approach also extends to other indices if we use $F_{\c+R}(S)$ as per
Equation~\eqref{eq:FR} and set $\LB_0 = \max_{u, v}\abs{w_u + m_{u, v}}$.

\subsection{The sandwich method}
\label{sec:dis-con}
Now we present an algorithm for finding at most $k$ seed nodes that maximize the
\disconidx~$\DisConIdx{G, \begop}$ and the \conidx~$\ConIdx{G, \begop}$. Since
these optimization problems are not submodular, we cannot use the greedy
algorithm from above. However, the indices' matrices $\MasIdx{\DisConIdx{}}$ and
$\MasIdx{\ConIdx{}}$ only contain non-negative entries and
this allows us to define submodular upper and lower bounds on the objective
functions. Thus, we apply the sandwich
method~\cite{lu2015competition} 
to obtain data-dependent approximation guarantees.

We obtain our upper and lower bounds as follows. Let 
$\MasIdx{\mathcal{M}} \in \{ \MasIdx{\DisConIdx{}}, \MasIdx{\ConIdx{}} \}$. 
Now let $\MasIdx{\mathcal{M}}^U$ be the
diagonal matrix in which $\MasIdx{\mathcal{M}}^U_{ii}$ is the sum of all entries in the
$i$'th row of $\MasIdx{\mathcal{M}}$. Let
$\mu_0(S) = \Exp[2 \begop^\intercal \MasIdx{\mathcal{M}} \Delta \s + \Delta \s^\intercal \MasIdx{\mathcal{M}} \Delta \s]$, 
$\mu_L(S) = \Exp[2 \begop^\intercal \MasIdx{\mathcal{M}} \Delta \s]$, 
$\mu_U(S) = \Exp[2 \begop^\intercal \MasIdx{\mathcal{M}} \Delta \s + \Delta \s^\intercal \MasIdx{\mathcal{M}}^U \Delta \s]$. 
Since the entries of all of these matrices are
non-negative, we obtain our desired relationship
$\mu_L(S) \leq \mu_0(S) \leq \mu_U(S)$.

As both $\mu_L(S)$ and $\mu_U(S)$ are monotone and submodular, a greedy
algorithm can approximate them within factor $1 - \frac{1}{e} - \epsilon$.  In
our \emph{sandwich algorithm}, we greedily select nodes $S_L$, $S_U$
and $S_0$ that maximize $\mu_L(S)$, $\mu_U(S)$ and $\mu_0(S)$, respectively.  Then
we evaluate each of the sets on $\mu_0(S)$ and return the one with the highest
objective value, i.e., we return $\arg \max_{S\in\{S_0, S_L, S_U\}} \mu_0(S)$.
We obtain the following approximation guarantees.
\begin{theorem}[Lu et al.~\cite{lu2015competition}]
	\label{theorem:sand-lu}
  Let $S^* = \arg\max_{\abs{S} \leq k} \mu_0(S)$. Then
  $\mu_0(S) \geq \max \left\{\frac{\mu_0(S_U)}{\mu_U(S_U)}, \frac{\mu_L(S^*)}{\mu_0(S^*)}\right\} (1 - \frac{1}{e} - \epsilon)\, \mu_0(S^*)$.
\end{theorem}

\section{Experiments}
\label{sec:experiments}

We present the experimental evaluation of our model and our methods.  
Our experiments are conducted on an Intel Xeon\,E5\,2630\,v4 at 2.20\,GHz with
128GB memory. 
Our code is written in Julia and is available on
github.\footnote{\url{https://github.com/SijingTu/WebConf-22-Viral-Marketing-Opinion-Dynamics}}

\spara{Datasets.}
We use publicly available real-world
datasets~\cite{kunegis2013konect,rossi2015network,snapnets} of social networks.
For each network we extracted the largest connected component.
Dataset statistics are presented in Table~\ref{tab:graphstats}.

\spara{Parameters.}
For each network, we set the innate opinion $\ebegop{u}$ 
of each user $u$ uniformly at random in $[0, 1]$~\cite{xu2021fast}.
We set the parameters $p_{u,v}$ as in the weighted cascade model~\cite{kempe2015maximizing, tang2015influence, tang2014influence}, 
i.e., $p_{u,v} = \frac{1}{d(v)}$, where $d(v)$ is the in-degree of $v$. 
We set $w_{u,v} = 1$ for the FJ model.   
For the polarizing campaigns with backfire, we set $\tau = 0.5$.
For all of our algorithms and heuristics, we choose the parameters 
$\epsilon=0.1$, $\ell=1$ and $\epsilon_2=0.6$.

\begin{table}[t]
  \caption{Statistics of the datasets, where $n$ is the number of nodes and
	  $m$ is the number of edges.}
  \label{tab:graphstats}
  \centering
  \begin{small}
  \begin{tabular}{lrr}  
    \toprule
    \multirow{1}*{Dataset} & \multirow{1}*{$n$} & \multirow{1}{*}{$m$}\\ 
    \midrule 
    $\Convote$ & 219 & 521 \\
    \Netscience & 379 & 914 \\
    $\WikiTalkSmall$ & 404 & 734 \\
    \WikiVote & 889 & 2914 \\
    \Reed & 962 & 18812 \\
    \EmailUniv & 1133 & 5451 \\
    \Hamster & 2000 & 16097 \\
    \USFCA & 2672 & 65244 \\
    \bottomrule
    \end{tabular}~
    \end{small}
    \begin{small}
    \begin{tabular}{lrr}  
    \toprule
    \multirow{1}*{Dataset} & \multirow{1}*{$n$} & \multirow{1}{*}{$m$}\\ 
    \midrule 
    \NipsEgo & 2888 & 2981 \\
    \PagesGovernment & 7057 & 89429 \\
    \HepPh & 11204 & 117619 \\
    \Anybeat & 12645 & 49132 \\
    \CondMat & 21363 & 91286 \\
    \Gplus & 23613 & 39182 \\
    \Brightkite & 56739 & 212945 \\
    \WikiTalk & 92117 & 360767 \\
    \bottomrule
  \end{tabular}
  \end{small}
\end{table}

\spara{Algorithms.} 
We implemented our approximation algorithms from Sec.~\ref{sec:algorithms}
and we denote them \MaxSum for the sum index, and \MaxDisCon for the
disagreement--controversy index.  Additionally, we use heuristic versions of the
greedy Algorithm~\ref{algo:rrgreedy}, together with the statistical test scheme
from Algorithm~\ref{algo:sampling}.  This gives us the following algorithms:
\MaxPol for maximizing the polarization index,
\MaxInt for maximizing the internal conflict, and
\MaxDis for maximizing disagreement.

We will see below that those algorithms that include quadratic terms are very
costly to run.  Therefore, we also introduce scalable heuristics.  We make two
changes in the heuristics: (1)~To obtain the seed nodes, we only consider the
indices' linear components, but we evaluate the final set of seed nodes on the
whole function (including the quadratic part).  
(2)~Following the sampling scheme from Algorithm~\ref{algo:sampling} typically
leads to large sampling sizes and sometimes caused our algorithms to run out of
memory.  Thus, we sample at most $200n$ \rr-sets for smaller datasets, and $5n$
\rr-sets when $n > 50,000$.  This gives good estimates in practice (typically
with less than $1\%$ error).  We denote the heuristics by \MaxLinDisCon,
\MaxLinPol, \MaxLinInt, and \MaxLinDis. 

We compare our optimization algorithms against three baselines:
\MaxInfluence chooses the seed nodes that maximize the influence;
\HighDegree picks the seed nodes with highest degrees;
\Random selects seed nodes uniformly randomly.
Since \Random is the only randomized baseline, we report average values over
$10$~runs. 
As these methods provide us with a fixed seed set, we use the Monte Carlo
simulation from Sec.~\ref{sec:estimating} to evaluate their results.

Additionally, we compare against a greedy heuristic~\FJGreedy by Chen and
Racz~\cite{chen2020network} that maximizes the indices from
Table~\ref{table:index} under the vanilla FJ model. \FJGreedy is allowed to
change $k$~innate user opinions arbitrarily much but, unlike in our model, there
is no information spread;  we provide \FJGreedy with the same parameter $k$ as
all other algorithms. Unlike for the other methods, we do \emph{not} take the
seed set returned by \FJGreedy and compute its score in our model, but we report
the relative increase of \FJGreedy in the vanilla FJ model; this will allow us
to evaluate whether the information spreading makes our model more powerful. We
will also include a value~\FJUpp by Gaitonde, Kleinberg and
Tardos~\cite[Thm.~3.4]{gaitonde2020adversarial} which gives an analytic upper
bound on what is achievable in the setting of \FJGreedy; we note that this
upper bound might be loose (i.e., it is possible that it is too large). Note
that if our algorithms achieve values larger than \FJUpp, our model is strictly
more powerful than what is achievable in the vanilla FJ model without the
information propagation step.

\spara{Evaluation.} We report the relative increases of the
indices from Sec.~\ref{sec:preliminaries}. That is, for
$\MasIdx{\mathcal{M}}$ being a matrix from Table~\ref{table:index},
$\begop$ being the non-adjusted innate opinions, and $\s$ being the adjusted innate
opinions, we report
$(\s^\intercal \MasIdx{\mathcal{M}} \s - \begop^\intercal \MasIdx{\mathcal{M}}
		\begop)/(\begop^\intercal \MasIdx{\mathcal{M}} \begop)$.

\spara{How does viral marketing change the indices?}
First, let us consider how our baselines influence the user opinions under the
\spread.  In Figure~\ref{fig:baseline}, we report how the polarization index
changes when we pick 2\% of the nodes as seeds.
We repeat our
experiments $5$~times and present the mean and the variance.  
In Figure~\ref{fig:baseline}(a) we see that marketing campaigns have
little effect on the polarization index in the network and
increase it by less than 0.1\%. However, the situation is very different when we
consider polarizing campaigns with backfire (Figure~\ref{fig:baseline}(b)):
the polarization increases up to 60\% and typically increases 
\emph{at least} 20\% if the most influential users share the polarizing campaign.
Using random seed nodes has little impact on the polarization.
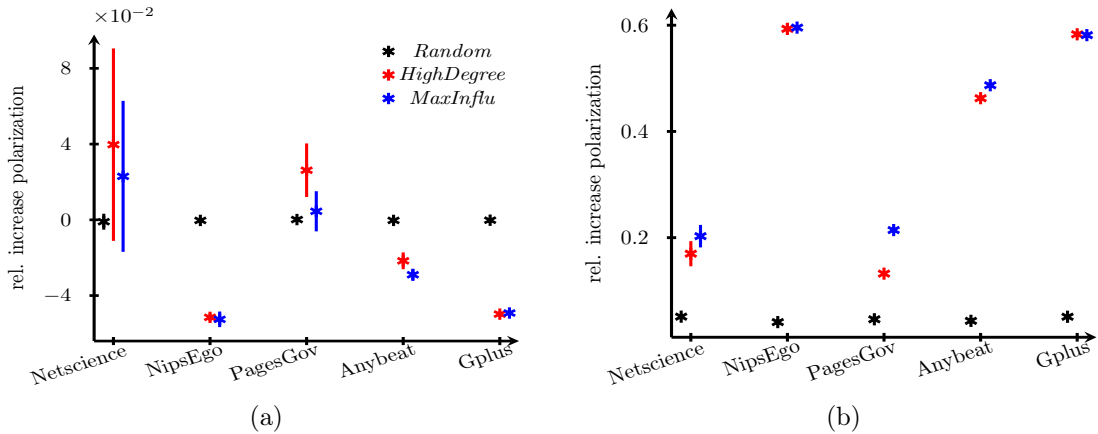
\begin{figure}[t]
	\centering
    \begin{tabular}{cc}
        \resizebox{0.48\textwidth}{0.34\textwidth}{%
			\input{tikz/basline/pol-1.tex}
		}&
		\resizebox{0.48\textwidth}{0.34\textwidth}{%
			\input{tikz/basline/pol-2.tex}
		}\\
		(a)&(b)
	\end{tabular}
	\caption{The relative change of the \polidx on different datasets with
		$k=\lceil 2\%\cdot n\rceil$ seed nodes.  The plots show (a)~marketing
		campaigns and (b)~polarizing campaigns.}
\label{fig:baseline}
\end{figure}

\spara{Scalability and accuracy of the heuristics.}
In Appendix~\ref{sec:add-experiments} we show that the
heuristics scale linearly in the size of the graph and are up to three orders of
magnitude faster than the greedy algorithms, while being of similar quality.
Thus, from here onwards we focus on the heuristic methods that only take into
account the linear terms and scale to larger datasets.

\spara{Experiments for marketing campaigns.}
Next, we evaluate our methods for marketing campaigns with
$k = \lceil 0.5\%\cdot n\rceil$ seed nodes. In Table~\ref{tab:marketing} we
report the results for all previously mentioned methods, excluding \HighDegree
which behaves very similarly to \MaxInfluence.
We will consider the \sumidx and the \polidx and we will evaluate how these
indices change based on solutions of algorithms with different objectives.
While this might look counter-intuitive at first, this approach reveals
interesting connections between the different methods we consider and the
indices we optimize.  For \FJGreedy, we use two corresponding versions that
maximize the \sumidx and the \polidx, respectively; for two large datasets,
\FJGreedy and \FJUpp ran out of time.

\input{tables/marketing-short.tex}

Let us consider the \sumidx. The methods \MaxSum,
\MaxLinDisCon and \MaxInfluence typically achieve the highest values and all
of them are of similar quality. Not surprisingly, this suggests that for
marketing campaigns maximizing the user opinions is essentially the same as
maximizing influence. For nine datasets, the \sumidx increases by less than 5\%
but for some it increases by up to 18.75\%.  Quite interestingly, only on two
datasets \MaxLinPol increases the \sumidx by more than 1\%, which suggests that
the solutions of \MaxLinPol and the other methods are quite dissimilar.
Additionally, we observe that the solution by \FJGreedy barely increases the
\sumidx.

However, the situation is quite different for the \polidx. Here, \MaxLinPol
clearly achieves the biggest increases followed by \MaxLinDis and \FJGreedy.
Interestingly, on several datasets the seed nodes produced by \MaxSum,
\MaxLinDisCon and \MaxInfluence even \emph{decrease} the polarization; we
explain this by the fact that if many users increase their opinions with respect
to a topic, then the overall acceptance of this topic increases and the topic
becomes less polarizing. 
Additionally, we observe that on all datasets,
\MaxLinPol achieves slightly higher values than \FJGreedy, even though \FJGreedy
can change the $k$ innate opinions arbitrarily much, while our marketing
campaign can increase each innate user opinion by at most $\varepsilon$.

For both indices, \Random and \MaxLinInt have little to no effect.

\spara{Experiments with polarizing campaigns.}
Next, we consider polarizing campaigns with backfire and
$k = \lceil 0.5\%\cdot n\rceil$ seed nodes. We report our results in
Table~\ref{tab:backfire}.
\input{tables/backfire-short.tex}
Let us start with the \sumidx. Unlike in marketing campaigns, now \MaxSum is
clearly the best method overall and outperforms \MaxInfluence. However, for all
methods the increase is very small, indicating that it is difficult to increase
the sum of the user opinions with polarizing campaigns.

Now consider the \polidx where the increases compared to the marketing
campaigns are startling. On 10 out of 14 datasets, \MaxLinPol increases the
polarization by \emph{at least} 10\% and the biggest increases reach up to 59\%.
This is in stark contrast to marketing campaigns where on all but two datasets
the polarization increased by \emph{at most} 10\%. Even
\MaxInfluence always increases the polarization by more than 5\% and up to 59\%.
Next, we observe that \MaxLinPol achieves much larger increases in polarization
than \FJGreedy, typically being at least factor~2 larger and up to factor~75
(for \NipsEgo).  Finally, we observe that on three datasets, \MaxLinPol
outperforms that analytic lower bound \FJUpp and for 10 out of 12 datasets it is
within factor~3. These findings suggest that for polarizing campaigns, the
information spread is very powerful compared to only changing the innate
opinions of a given set of users.

\spara{Further experiments} are provided in Appendix~\ref{sec:add-experiments}.

\section{Conclusions}
\label{sec:conclusions}
We presented a novel model that allows to quantify how viral information
effects user opinions in online social networks.  We presented algorithms to
simulate the model and to optimize different network indices. This allowed us to
understand how much impact adversaries can have on the social network.
Our experiments showed that marketing campaigns and polarizing contents behave very
differently. While for marketing campaigns it is possible to significantly
increase the user opinions, this seems very difficult for polarizing contents.
However, the picture is vastly different for the polarization in the network:
it barely increases for marketing campaigns but for polarizing
contents the increase can be very high, even when the number of seed
nodes is small. We believe that this gives an insight into the growing polarization
observed in today's social media.

There are several interesting directions for future work. Obtaining
approximation algorithm for polarizing contents is intriguing. Another important
question is to study how the parameters of our model should be set to capture
real-world behaviors as accurately as possible; beyond pure parameter
estimation, this might involve replacing the independent cascade model or the
FJ-model with other models from the literature.

\section*{Acknowledgements}
We are deeply grateful to Aristides Gionis for his mentorship and many
discussions during this project.
This research is supported by the Academy of Finland projects AIDA~(317085) and
	MLDB~(325117), the ERC Advanced Grant REBOUND~(834862), the EC~H2020 RIA
	project SoBigData++~(871042), and the Wallenberg~AI, Autonomous Systems and
	Software Program~(WASP) funded by the Knut and Alice Wallenberg Foundation.
Our computations were enabled by resources provided by the Swedish National
Infrastructure for Computing (SNIC) at UPPMAX partially funded by the
Swedish Research Council through grant agreement no. 2018-05973.

\bibliographystyle{plain}
\bibliography{bibfile}

\clearpage
\appendix

\input{online-only}

\section{Comparison of the spread--acknowledge model and the FJ model for the sum index}
\label{sec:comparison}
We discuss how the \sumidx is affected when we plant~$k$ seed nodes in the
\spread with marketing campaigns, compared to changing the opinions of $k$ users
in the FJ model. We argue that in the \spread, the \sumidx will increase at
least as much as in the FJ~model, regardless of which seed nodes are picked.

First, suppose that in the FJ model we are allowed to increase $k$~innate user
opinions by $\varepsilon$. Let $\begop$ denote the user opinions before the
increase and let $\s$ denote the user opinions after the increase. Then we have
that for the innate user opinions it holds that
$\ind^{\intercal} \s = \ind^{\intercal} \begop + k\varepsilon$.
Thus, no matter for which $k$~users we change the innate opinions, it will
always increase the \sumidx by exactly $k\varepsilon$.

Second, consider the \spread and consider any set of $k$~seed nodes. Since we
consider marketing campaigns, for each of the $k$~seed nodes, the innate opinion
will be increased by~$\varepsilon$. Additionally, the information cascade may
reach some non-seed users and increase their opinions. Hence, the \sumidx will
increase by at least $k\varepsilon$.

\section{Additional Experiments}
\label{sec:add-experiments}
In this section we present additional results from our experiments. All
parameters are set as described in the main text, unless stated otherwise.

\spara{Further evaluation of baselines.}
In Figure~\ref{fig:baseline} we showed that viral marketing campaigns can have a
strong effect on the \polidx. Here we report our experimental results for other
network indices. In Figure~\ref{fig:baseline-discon} we present the results for
the \disconidx, in Figure~\ref{fig:baseline-dis} for the \disidx and in
Figure~\ref{fig:baseline-int} for the \intidx.

\begin{figure}[H]
	\centering
    \begin{tabular}{cc}
        \resizebox{0.48\textwidth}{0.34\textwidth}{%
			\input{tikz/basline/con-1.tex}
		}&
		\resizebox{0.48\textwidth}{0.34\textwidth}{%
			\input{tikz/basline/con-2.tex}
		}\\
		(a)&(b)
	\end{tabular}
	\caption{The relative change of the \disconidx on different datasets with
		$k=\lceil 2\%\cdot n\rceil$ seed nodes.  The plots show (a)~marketing
		campaigns and (b)~polarizing campaigns.}
\label{fig:baseline-discon}
\end{figure}
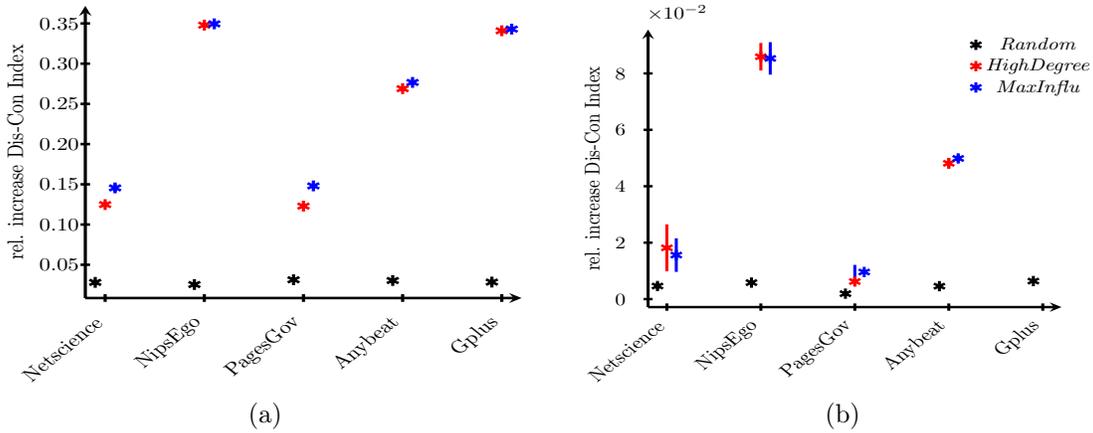

\begin{figure}[H]
	\centering
    \begin{tabular}{cc}
        \resizebox{0.48\textwidth}{0.34\textwidth}{%
			\input{tikz/basline/dis-1.tex}%
		}&
		\resizebox{0.48\textwidth}{0.34\textwidth}{%
			\input{tikz/basline/dis-2.tex}%
		}\\
		(a)&(b)
	\end{tabular}
	\caption{The relative change of the \disidx on different datasets with
		$k=\lceil 2\%\cdot n\rceil$ seed nodes.  The plots show (a)~marketing
		campaigns and (b)~polarizing campaigns.}
\label{fig:baseline-dis}
\end{figure}
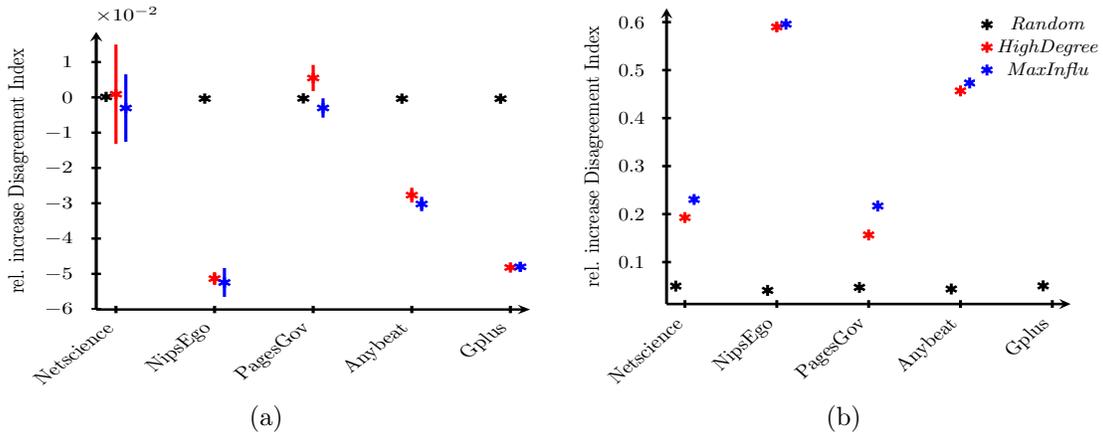

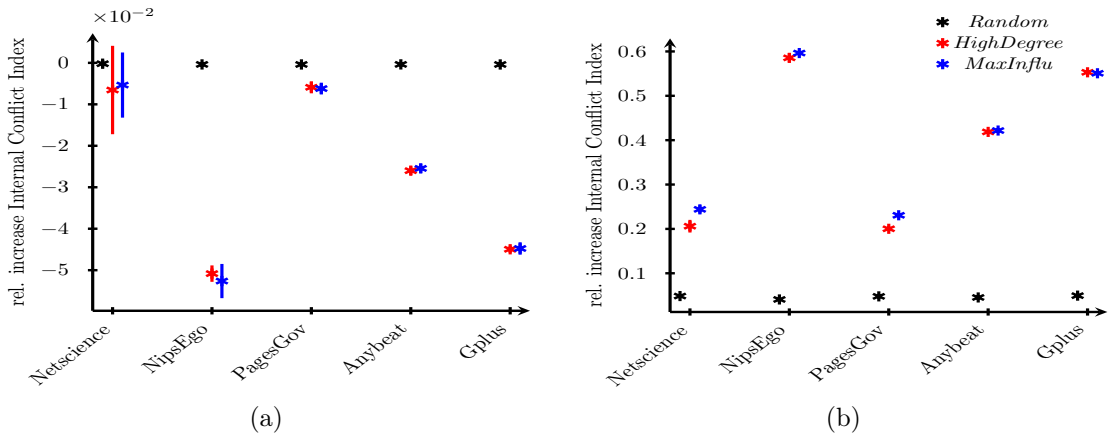
\begin{figure}[H]
	\centering
    \begin{tabular}{cc}
        \resizebox{0.48\textwidth}{0.34\textwidth}{%
			\input{tikz/basline/int-1.tex}%
		}&
		\resizebox{0.48\textwidth}{0.34\textwidth}{%
			\input{tikz/basline/int-2.tex}%
		}\\
		(a)&(b)
	\end{tabular}
	\caption{The relative changes of the \intidx on different datasets with
		$k=\lceil 2\%\cdot n\rceil$ seed nodes.  The plots show (a)~marketing
		campaigns and (b)~polarizing campaigns.}
\label{fig:baseline-int}
\end{figure}

\spara{Evaluation of the heuristics.}
Next, we turn our attention to our algorithms and we study their scalability and their
accuracy.  
Figure~\ref{fig:time}(a) compares the greedy algorithms and the heuristics.  It
shows that the greedy algorithms are up to three orders of magnitude slower than the
heuristics; this makes running the greedy algorithms prohibitively costly on
larger datasets.  
In Figure~\ref{fig:time}(b) we study the solution quality of \MaxLinDisCon. 
We consider the \disconidx (other indices behave similarly)
and compare \MaxLinDisCon with \MaxDisCon; we also
include a lower bound~\LowDisCon, as discussed in Sec.~\ref{sec:dis-con}. We observe that
the heuristic \MaxLinDisCon performs slightly worse than \MaxDisCon.
Nonetheless, the results of \MaxLinDisCon are
of high quality and almost as good as the much slower \MaxDisCon. 

\begin{figure}[t!]
	\centering
    \begin{tabular}{cc}
        \resizebox{0.48\textwidth}{0.34\textwidth}{%
			\input{tikz/greedy/time.tex}
		}&
		\resizebox{0.48\textwidth}{0.34\textwidth}{%
			\input{tikz/greedy/d.tex}
		}\\
		(a)&(b)
	\end{tabular}
\caption{Analysis of the greedy algorithms and heuristics. Plot~(a) shows the
	running times (in seconds) w.r.t.\ the graph size and (b)~the relative
	increase of the \disconidx with $k=5$ and marketing campaigns.}
\label{fig:time}
\end{figure}
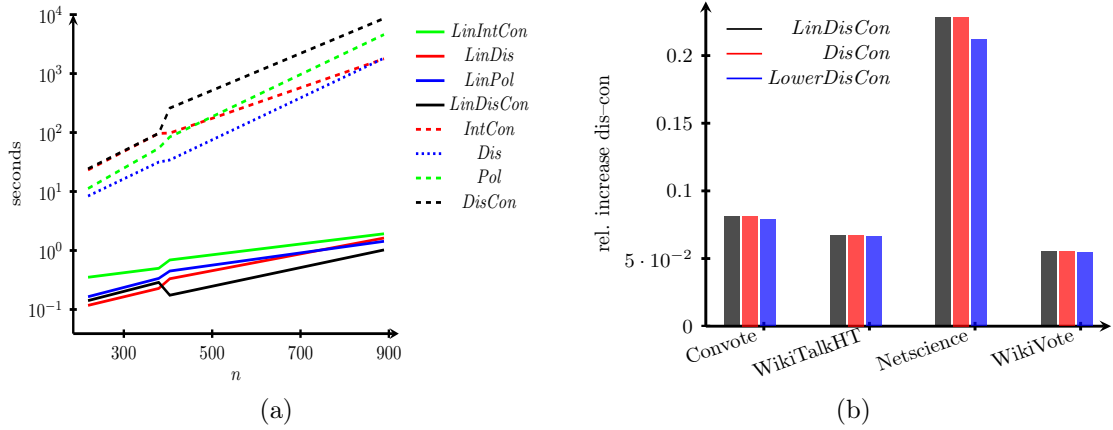

Figure~\ref{fig:time-2} shows that the heuristics that only consider the linear
parts scale linearly in the size of the graph. 
In Figure~\ref{fig:compare-2} and Figure~\ref{fig:compare-3}, 
we add the comparisons of relative increases on other network indices. 

\begin{figure}[H]
	\centering
    \begin{tabular}{cc}
        \resizebox{0.48\textwidth}{0.34\textwidth}{%
			\input{tikz/ttt.tex}
		}
	\end{tabular}
\caption{
The running times in seconds of the heuristics 
w.r.t.\ the graph size.}
\label{fig:time-2}
\end{figure}
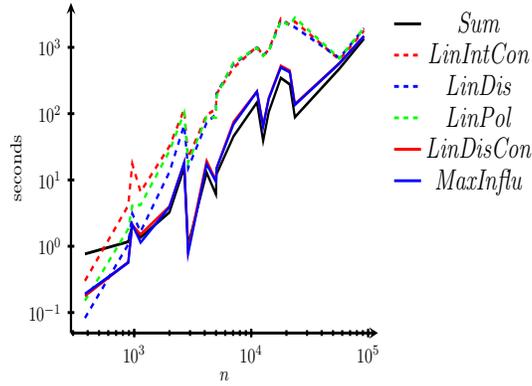

\begin{figure}[H]
	\centering
    \begin{tabular}{cc}
		\resizebox{0.48\textwidth}{0.34\textwidth}{%
			\input{tikz/greedy/a.tex}
		}
	\end{tabular}
\caption{Analysis of the greedy algorithms and heuristics on different datasets.
	The plot shows the relative increase of the \intidx with $k=5$ and marketing
	campaigns.}
\label{fig:compare-2}
\end{figure}
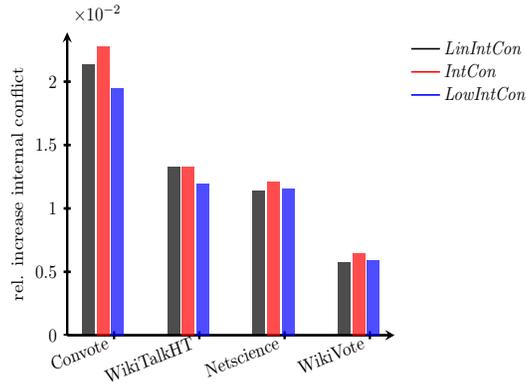

\begin{figure}[H]
	\centering
    \begin{tabular}{cc}
        \resizebox{0.48\textwidth}{0.34\textwidth}{%
			\input{tikz/greedy/b.tex}
		}&
		\resizebox{0.48\textwidth}{0.34\textwidth}{%
			\input{tikz/greedy/c.tex}
		}\\
		(a)&(b)
	\end{tabular}
\caption{Analysis of the greedy algorithms and heuristics on different datasets. Plot~(a) presents the
	relative increase of the \disidx and Plot~(b) presents the relative increase
	of the \polidx. In both cases we used $k=5$ and marketing campaigns.}
\label{fig:compare-3}
\end{figure}
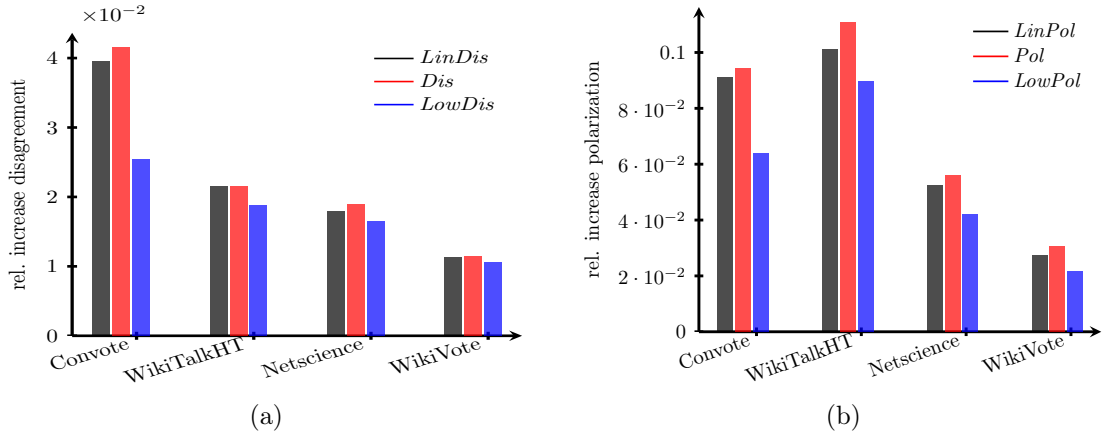

\input{tables/marketing-2.tex}
\input{tables/backfire-2.tex}

\spara{Further evaluation of our algorithms.}
In Tables~\ref{tab:marketing} and~\ref{tab:backfire}, we presented the algorithm
results for the \sumidx and the \polidx. Below we will add other two network
indices, namely, \disidx in Table~\ref{tab:marketing-2} and \intidx in
Table~\ref{tab:backfire-2}. We again observe that both indices only change
slightly for marketing campaigns but for polarizing campaigns with backfire the
indices can increase drastically even for baselines such as \MaxInfluence.
Again, in the setting with polarizing campaigns, our algorithms outperform
\FJGreedy. Furthermore, typically our algorithms get within a small factor of
\FJUpp and sometimes even provide larger gains in indices we study. This further
strengthens the conclusion that the information spread provides a significant
gain over the vanilla FJ model.

\spara{Additional results regarding adjusting $\epsilon$.}
Next, we evaluate how the parameter~$\varepsilon$, i.e., how much the innate
opinions are changed when they are adjusted, impacts the results of our
experiments.  In Figure~\ref{fig:epsilon}, we present the relative increase of
the \polidx and see that it increases linearly in~$\varepsilon$. Here, we picked
$k=0.5\%n$~seed vertices.

\begin{figure}[H]
	\centering
    \begin{tabular}{cc}
        \resizebox{0.48\textwidth}{0.34\textwidth}{%
			\input{tikz/epsilon/alpha1.tex}
		}&
		\resizebox{0.48\textwidth}{0.34\textwidth}{%
			\input{tikz/epsilon/alpha2.tex}
		}\\
		(a)&(b)
	\end{tabular}
\caption{Analysis of varying $\epsilon$ on dataset \HepPh. 
	Plot~(a) presents the relative increase of polarization index on marketing
	campaigns and Plot~(b) presents the relative increase of polarization index
	on polarized campaigns.  In both cases we used $k=0.5\%n$.}
\label{fig:epsilon}
\end{figure}
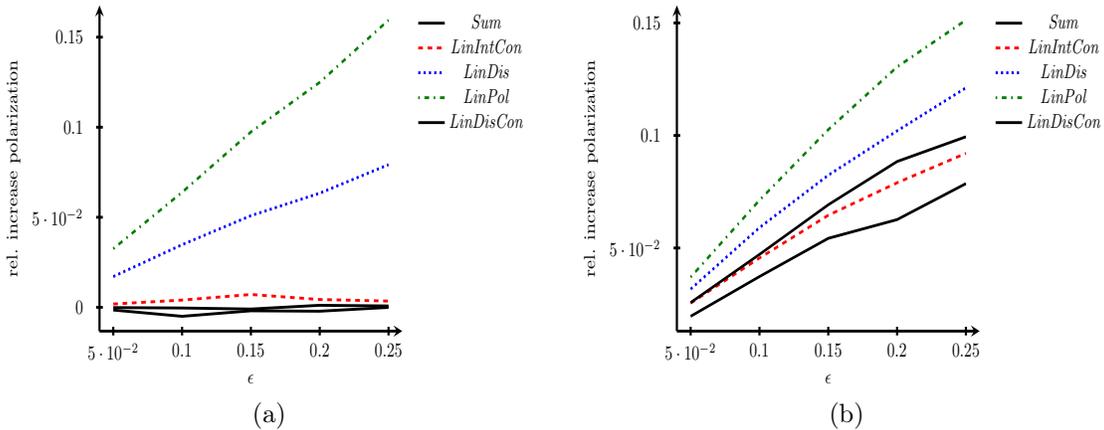

\spara{Initializing the innate opinions with the exponential distribution.}
While in the main text we used the uniform distribution to initialize the
innate opinions, now we initialize the innate opinions using the exponential
distribution. More concretely, we proceed as in Xu et al.~\cite{xu2021fast}:
we sample $n$~numbers $a_1,\dots,a_n\geq 1$ from the ditribution with density
$e^{1-x}$ and we scale it into the interval $[0,1]$ by setting
$\begop_u = \frac{a_u}{\max_u a_u}$.
The results are listed in Table~\ref{tab:exp-marketing-1}, Table~\ref{tab:exp-marketing-2},
Table~\ref{tab:exp-backfire-1}, Table~\ref{tab:exp-backfire-2}.

For marketing campaigns, we observe that the \sumidx can be increased more than
in the setting with uniform opinions. Unlike in the uniform opinions setting,
the \polidx can be increased quite significantly even for marketing campaigns;
however, similar to the uniform opinions setting, the increase of the \polidx is
typically higher for polarizing campaigns than for marketing campaigns.

\input{tables/exp-marketing-1.tex}

\input{tables/exp-marketing-2.tex}
\input{tables/exp-backfire-1.tex}
\input{tables/exp-backfire-2.tex}

\spara{Performance of the sandwich method.}
Next, we present our experiments for the sandwich algorithm from
Sec.~\ref{sec:dis-con} for maximizing the \disconidx with marketing campaigns.
We vary the number of seed nodes $k=1,3,5$. We focus on small datasets since
\MaxDis has large running times, as it also takes into account the quadratic
terms. We denote the lower and upper bounds from Sec.~\ref{sec:dis-con} by
\LowDisCon and \UppDisCon, respectively.

Figures~\ref{fig:sandwich}(a) and~\ref{fig:sandwich}(b) present the results for
\Convote and for \Netscience, respectively.  We observe that \UppDisCon and
\LowDisCon provide similar results to \SandDisCon. This indicates that the
data-dependent approximation ratios that we derived in
Theorem~\ref{theorem:sand-lu} are fairly tight in practice.

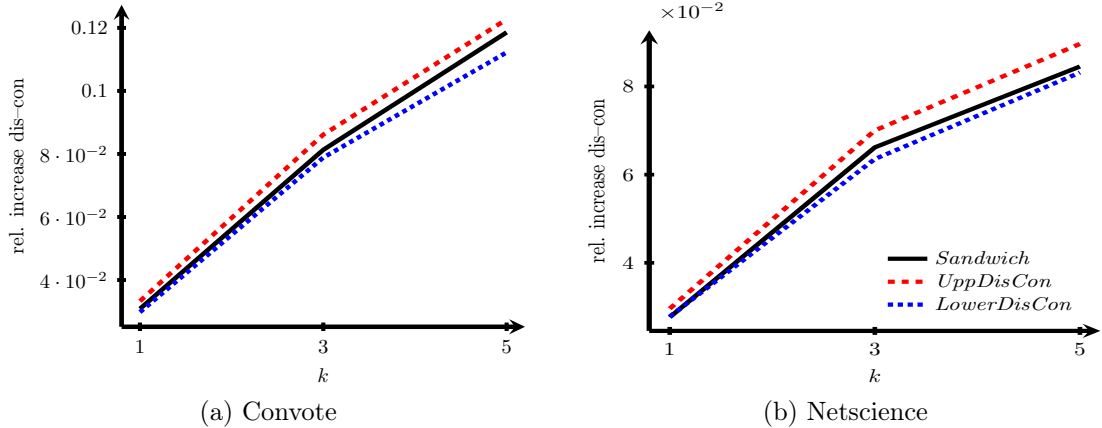
\begin{figure}[t!]
	\centering
    \begin{tabular}{cc}
        \resizebox{0.48\textwidth}{0.34\textwidth}{%
	\input{tikz/greedy/f.tex}
	}&
	\resizebox{0.48\textwidth}{0.34\textwidth}{%
	\input{tikz/greedy/g.tex}
	}\\
	(a) \Convote & (b) \Netscience
\end{tabular}
\caption{Results of the sandwich method
	for maximizing the \disconidx with marketing campaigns.
	We present the relative increases of the \disconidx for $k=1,3,5$.}
\label{fig:sandwich}
\end{figure}

\end{document}

%% file: tikz/basline/pol-1.tex
\pgfplotsset{compat=1.3}
\begin{tikzpicture}

\begin{groupplot}[group style={group size=1 by 1}]
\nextgroupplot[
tick pos = left,
axis lines=left, xtick=\empty, 
x grid style={white!69.0196078431373!black},
yticklabel style = {xshift=-1ex},
xticklabel style = {yshift=-1ex},
xmin=-0.2, xmax=4.2,
xtick style={color=black},
xtick={0,1,2,3,4},
xticklabel style={rotate=20.0,anchor=east},
xticklabels={\Netscience,\NipsEgo,\PagesGovernment,\Anybeat,\Gplus},
y grid style={white!69.0196078431373!black},
ymin=-0.0640178278669176, ymax=0.0978777631034836,
ytick style={color=black},
ytick={-0.08,-0.04,0,0.04,0.08,0.1},
ylabel={rel. increase polarization},
legend style={draw=none},
legend style={at={(1,1), anchor=north east}},
line width = 0.50 mm, 
tick style={line width=0.50mm},
x tick label style={xshift={1em}}, 
ytick scale label code/.code={$\times 10^{-2}$},
]
\path [draw=black]
(axis cs:-0.1,-0.005206868784582)
--(axis cs:-0.1,0.003229237517426);

\path [draw=black]
(axis cs:0.9,-0.001595159324288)
--(axis cs:0.9,0.000985387511902);

\path [draw=black]
(axis cs:1.9,-0.00074522972139)
--(axis cs:1.9,0.001029968531194);

\path [draw=black]
(axis cs:2.9,-0.000712947207044)
--(axis cs:2.9,0.000134865324256);

\path [draw=black]
(axis cs:3.9,-0.000497460652014)
--(axis cs:3.9,0.000112270587958);

\path [draw=red]
(axis cs:0,-0.011117022930643)
--(axis cs:0,0.090518872604829);

\path [draw=red]
(axis cs:1,-0.053535891683399)
--(axis cs:1,-0.049470479267119);

\path [draw=red]
(axis cs:2,0.012042909130786)
--(axis cs:2,0.040353483868796);

\path [draw=red]
(axis cs:3,-0.026082833383687)
--(axis cs:3,-0.017207499867765);

\path [draw=red]
(axis cs:4,-0.050316175413966)
--(axis cs:4,-0.049273967730664);

\path [draw=blue]
(axis cs:0.1,-0.016897879443348)
--(axis cs:0.1,0.062835905550274);

\path [draw=blue]
(axis cs:1.1,-0.056658937368263)
--(axis cs:1.1,-0.048427766999203);

\path [draw=blue]
(axis cs:2.1,-0.006088975465686)
--(axis cs:2.1,0.015117390237112);

\path [draw=blue]
(axis cs:3.1,-0.032178543513789)
--(axis cs:3.1,-0.025837301736653);

\path [draw=blue]
(axis cs:4.1,-0.050302475579959)
--(axis cs:4.1,-0.048114372237689);

\addplot [black, mark=asterisk, mark size=3, mark options={solid}, only marks]
table {%
-0.1 -0.000988815633578
0.9  -0.000304885906193
1.9  0.000142369404902
2.9  -0.000289040941394
3.9  -0.000192595032028
};
\addplot [red, mark=asterisk, mark size=3, mark options={solid}, only marks]
table {%
0 0.039700924837093
1 -0.051503185475259
2 0.026198196499791
3 -0.021645166625726
4 -0.049795071572315
};
\addplot [blue, mark=asterisk, mark size=3, mark options={solid}, only marks]
table {%
0.1 0.022969013053463
1.1 -0.052543352183733
2.1 0.0045142073857133
3.1 -0.029007922625221
4.1 -0.049208423908824
};
\addlegendentry{\Random}
\addlegendentry{\HighDegree}
\addlegendentry{\MaxInfluence}
\end{groupplot}

\end{tikzpicture}

%% file: tikz/basline/pol-2.tex
\pgfplotsset{compat=1.3}
\begin{tikzpicture}

\begin{groupplot}[group style={group size=2 by 1}]
\nextgroupplot[
tick pos=left,
axis lines=left, xtick=\empty,
x grid style={white!69.0196078431373!black},
xmin=-0.2, xmax=4.2,
yticklabel style = {xshift=-1 ex},
xticklabel style = {yshift=-1ex},
xtick style={color=black},
xtick={0,1,2,3,4},
xticklabel style={rotate=20.0,anchor=east},
xticklabels={\Netscience,\NipsEgo,\PagesGovernment,\Anybeat,\Gplus},
y grid style={white!69.0196078431373!black},
ymin=0.0124396969741193, ymax=0.631579454382637,
ytick style={color=black},
ytick={0,0.2,0.4,0.6,0.7},
yticklabels={
  \(\displaystyle {0.0}\),
  \(\displaystyle {0.2}\),
  \(\displaystyle {0.4}\),
  \(\displaystyle {0.6}\),
  \(\displaystyle {0.7}\)
},
ylabel={rel. increase polarization},
legend style={font = \LARGE, draw=none},
legend pos= north west,
line width = 0.50 mm, 
tick style={line width=0.50mm},
x tick label style={xshift={1em}}
]
\path [draw=black]
(axis cs:-0.1,0.048816550714439)
--(axis cs:-0.1,0.053668736544641);

\path [draw=black]
(axis cs:0.9,0.040582413219961)
--(axis cs:0.9,0.041374233836029);

\path [draw=black]
(axis cs:1.9,0.045335004892114)
--(axis cs:1.9,0.047214129305732);

\path [draw=black]
(axis cs:2.9,0.043070097530129)
--(axis cs:2.9,0.043626259624727);

\path [draw=black]
(axis cs:3.9,0.050917644325787)
--(axis cs:3.9,0.051164238997957);

\path [draw=red]
(axis cs:0,0.146109966026742)
--(axis cs:0,0.193468775010156);

\path [draw=red]
(axis cs:1,0.582898762385253)
--(axis cs:1,0.603436738136795);

\path [draw=red]
(axis cs:2,0.128321248490936)
--(axis cs:2,0.136172949735356);

\path [draw=red]
(axis cs:3,0.457402918230087)
--(axis cs:3,0.467776657835823);

\path [draw=red]
(axis cs:4,0.580716308985371)
--(axis cs:4,0.585104167360463);

\path [draw=blue]
(axis cs:0.1,0.181625862679128)
--(axis cs:0.1,0.223995181589108);

\path [draw=blue]
(axis cs:1.1,0.590741909531475)
--(axis cs:1.1,0.600304527238259);

\path [draw=blue]
(axis cs:2.1,0.209928071229362)
--(axis cs:2.1,0.218645209943172);

\path [draw=blue]
(axis cs:3.1,0.482318063178492)
--(axis cs:3.1,0.491725035778242);

\path [draw=blue]
(axis cs:4.1,0.576856117713429)
--(axis cs:4.1,0.585822544929617);

\addplot [black, mark=asterisk, mark size=3, mark options={solid}, only marks]
table {%
-0.1 0.05124264362954
0.9  0.040978323527995
1.9  0.046274567098923
2.9  0.043348178577428
3.9  0.051040941661872
};
\addplot [red, mark=asterisk, mark size=3, mark options={solid}, only marks]
table {%
0 0.169789370518449
1 0.593167750261024
2 0.132247099113146
3 0.462589788032955
4 0.582910238172917
};
\addplot [blue, mark=asterisk, mark size=3, mark options={solid}, only marks]
table {%
0.1 0.202810522134118
1.1 0.595523218384867
2.1 0.214286640586267
3.1 0.487021549478367
4.1 0.581339331321523
};
\end{groupplot}

\end{tikzpicture}

%% file: tables/marketing-short.tex
\begin{table*}[ht!]
\caption{Results for marketing campaigns with
	$k = \lceil 0.5\%\cdot n\rceil$ seeds.
	We report the relative increase of each index in percent.}
\label{tab:marketing}
\centering
  \vspace{-4mm}
\begin{adjustbox}{max width=\textwidth}
\begin{tabular}{c cccccccc ccccccccc}
\toprule
\textbf{Dataset}  & \multicolumn{8}{c}{\textbf{\sumidx}} & \multicolumn{9}{c}{\textbf{\polidx}} \\
\cmidrule(lr){2-9} 
\cmidrule(lr){10-18} 
  &  \MaxSum  &  \MaxLinDisCon  &  \MaxLinPol  &  \MaxLinDis  &  \MaxLinInt  &   \MaxInfluence  &  \Random  &  \FJGreedy  &  \MaxSum  &  \MaxLinDisCon  &  \MaxLinPol  &  \MaxLinDis  &  \MaxLinInt  &  \MaxInfluence  &  \Random  & \FJGreedy  & \FJUpp \\
\midrule
\Netscience   & \textbf{2.79} & 2.75 & 0.74 & 1.01 & 0.21 & 2.78 & 0.27 & 0.11  & 3.15 & 3.18 & \textbf{7.54} & 5.89 & -0.35 & 3.17 & -0.06 & 2.36 &   10.54   \\
\WikiVote   & \textbf{4.14} & 4.12 & 0.53 & 0.64 & 0.48 & 4.11 & 0.3 & 0.11  & -0.64 & -0.61 & \textbf{3.83} & 3.2 & 0.81 & -0.58 & -0.06 & 2.92 &   12.29   \\
\Reed   & 3.2 & \textbf{3.22} & 0.28 & 0.3 & 0.3 & 3.2 & 0.27 & 0.1  & 0.14 & 0.09 & \textbf{10.13} & 8.04 & 0.15 & 0.11 & -0.13 & 9.56 &   68.48   \\
\EmailUniv   & 3.31 & \textbf{3.36} & 0.37 & 0.41 & 1.18 & 3.35 & 0.29 & 0.11  & 0.51 & 0.42 & \textbf{4.64} & 4.13 & 0.35 & 0.44 & -0.13 & 3.8 &   18.12   \\
\Hamster   & \textbf{4.09} & 4.07 & 0.76 & 0.81 & 0.45 & 4.06 & 0.27 & 0.1  & 0.39 & 0.46 & \textbf{7.59} & 5.82 & 0.43 & 0.72 & 0.01 & 5.18 &   25.70   \\
\USFCA   & 3.09 & 3.09 & 0.23 & 0.22 & 0.28 & \textbf{3.1} & 0.3 & 0.11  & -0.68 & -0.67 & \textbf{11.7} & 11.01 & 0.47 & -0.68 & -0.07 & 11.57 &   82.48   \\
\NipsEgo   & \textbf{18.75} & \textbf{18.75} & 0.47 & 0.1 & 0.1 & \textbf{18.75} & 0.15 & 0.1  & -5.3 & -5.29 & \textbf{1.71} & 0.51 & 0.12 & -5.29 & -0.09 & 0.79 &   5.34   \\
\PagesGovernment   & 3.47 & 3.47 & 0.53 & 0.5 & 0.44 & \textbf{3.48} & 0.28 & 0.1  & 0.78 & 0.79 & \textbf{7.31} & 5.49 & 0.85 & 0.51 & -0.06 & 6.96 &   36.87   \\
\HepPh   & 2.52 & \textbf{2.53} & 0.61 & 0.67 & 0.42 & 2.52 & 0.26 & 0.1  & -0.3 & 0.01 & \textbf{6.83} & 4.5 & 0.43 & -0.09 & -0.05 & 3.26 &   16.03   \\
\Anybeat   & \textbf{11.96} & \textbf{11.96} & 0.93 & 1.02 & 0.53 & \textbf{11.96} & 0.25 & 0.1  & -1.21 & -1.19 & \textbf{3.18} & 2.58 & 0.52 & -1.24 & -0.08 & 1.7 &   7.80   \\
\CondMat   & \textbf{2.79} & \textbf{2.79} & 0.59 & 0.65 & 0.48 & \textbf{2.79} & 0.25 & 0.1  & 0.35 & 0.58 & \textbf{6.9} & 4.61 & 0.44 & 0.31 & -0.07 & 3.42 &   15.69   \\
\Gplus   & \textbf{18.06} & \textbf{18.06} & 3.85 & 0.37 & 0.44 & \textbf{18.06} & 0.26 & 0.1  & -4.98 & -4.98 & \textbf{6.2} & 1.03 & 0.29 & -4.98 & -0.07 & 0.92 &   6.41   \\
\Brightkite   & 6.16 & 6.15 & 0.72 & 0.89 & 0.53 & \textbf{6.17} & 0.27 & -  & -0.17 & -0.06 & \textbf{4.27} & 2.53 & 0.47 & -0.24 & -0.07 & - &   -   \\
\WikiTalk   & 9.27 & 9.27 & 1.73 & 1.59 & 0.71 & \textbf{9.28} & 0.29 & -  & -0.82 & -0.71 & \textbf{3.37} & 2.63 & 0.62 & -0.79 & -0.09 & - &   -   \\
\bottomrule
\end{tabular}
\end{adjustbox}
\end{table*}

%% file: tables/backfire-short.tex
\begin{table*}[ht!]
\caption{Results for polarizing campaigns with
	$k = \lceil 0.5\%\cdot n\rceil$ seeds.
	We report the relative increase of each index in percent.}
\label{tab:backfire}
\centering
  \vspace{-4mm}
\begin{adjustbox}{max width=\textwidth}
\begin{tabular}{c ccccccccc ccccccccc}
\toprule
\textbf{Dataset}  & \multicolumn{8}{c}{\textbf{\sumidx}} & \multicolumn{9}{c}{\textbf{\polidx}} \\
\cmidrule(lr){2-9} 
\cmidrule(lr){10-18} 
  &  \MaxSum  &  \MaxLinDisCon  &  \MaxLinPol  &  \MaxLinDis  &  \MaxLinInt  &   \MaxInfluence  &  \Random  &  \FJGreedy  &  \MaxSum  &  \MaxLinDisCon  &  \MaxLinPol  &  \MaxLinDis  &  \MaxLinInt  &  \MaxInfluence  &  \Random  & \FJGreedy  & \FJUpp \\
\midrule
\Netscience   & \textbf{0.48} & 0.46 & 0.04 & 0.07 & 0.34 & 0.34 & -0.01 & 0.11  & 2.72 & 5.03 & \textbf{7.62} & 5.54 & 5.65 & 5.66 & 0.62 & 2.36 &   10.54   \\
\WikiVote   & \textbf{0.33} & 0.25 & -0.28 & -0.28 & -0.27 & -0.33 & -0.02 & 0.11  & 3.14 & 5.83 & 9.46 & \textbf{9.5} & 9.09 & 9.14 & 0.64 & 2.92 &   12.29   \\
\Reed   & \textbf{0.27} & \textbf{0.27} & 0.02 & 0.18 & 0.16 & 0.16 & 0.01 & 0.1  & 6.6 & 7.68 & \textbf{15.82} & 11.35 & 8.87 & 8.79 & 0.74 & 9.56 &   68.48   \\
\EmailUniv   & \textbf{0.33} & 0.31 & -0.23 & -0.21 & -0.18 & -0.16 & -0.02 & 0.11  & 3.39 & 5.23 & \textbf{7.72} & 7.36 & 6.87 & 6.81 & 0.79 & 3.8 &   18.12   \\
\Hamster   & \textbf{0.3} & 0.27 & -0.07 & -0.09 & -0.11 & -0.1 & -0.01 & 0.1  & 4.3 & 5.2 & \textbf{11.1} & 8.98 & 8.55 & 8.56 & 0.66 & 5.18 &   25.70   \\
\USFCA   & \textbf{0.22} & \textbf{0.22} & 0.05 & -0.04 & -0.02 & -0.02 & -0.01 & 0.11  & 6.85 & 10.05 & \textbf{13.18} & 7.39 & 5.68 & 5.74 & 0.74 & 11.57 &   82.48   \\
\NipsEgo   & \textbf{0.46} & 0.22 & 0.2 & 0.21 & 0.21 & 0.21 & 0.0 & 0.1  & 38.05 & 59.55 & \textbf{59.57} & 59.56 & 59.55 & \textbf{59.57} & 0.46 & 0.79 &   5.34   \\
\PagesGovernment   & \textbf{0.29} & 0.28 & -0.03 & 0.01 & 0.0 & -0.0 & -0.0 & 0.1  & 4.78 & 5.89 & \textbf{12.73} & 10.1 & 6.19 & 7.48 & 0.66 & 6.96 &   36.87   \\
\HepPh   & \textbf{0.33} & 0.32 & -0.01 & -0.08 & -0.1 & -0.1 & -0.01 & 0.1  & 3.69 & 5.17 & \textbf{7.62} & 5.91 & 4.61 & 5.16 & 0.66 & 3.26 &   16.03   \\
\Anybeat   & \textbf{0.42} & 0.3 & 0.11 & 0.12 & 0.14 & 0.12 & 0.0 & 0.1  & 29.3 & 38.14 & \textbf{39.84} & 39.75 & 39.12 & 39.55 & 0.48 & 1.7 &   7.80   \\
\CondMat   & \textbf{0.36} & 0.32 & 0.01 & 0.01 & 0.02 & 0.02 & 0.0 & 0.1  & 4.26 & 5.58 & \textbf{8.28} & 6.68 & 5.32 & 5.84 & 0.65 & 3.42 &   15.69   \\
\Gplus   & \textbf{0.49} & 0.15 & 0.1 & 0.1 & 0.1 & 0.1 & 0.0 & 0.1  & 29.75 & 57.48 & \textbf{57.94} & \textbf{57.94} & 57.92 & 57.93 & 0.66 & 0.92 &   6.41   \\
\Brightkite   & \textbf{0.38} & 0.24 & -0.02 & -0.0 & 0.0 & 0.01 & 0.0 & -  & 5.66 & 13.35 & \textbf{15.86} & 15.65 & 15.17 & 15.58 & 0.7 & - &   -   \\
\WikiTalk   & \textbf{0.49} & 0.29 & 0.02 & 0.01 & 0.02 & 0.02 & 0.0 & -  & 13.46 & 25.84 & \textbf{28.79} & 28.71 & 28.19 & 28.57 & 0.73 & - &   -   \\
\bottomrule
\end{tabular}
\end{adjustbox}
\end{table*}

%% file: online-only.tex
\section{Overview of the Appendix}
The appendix is organized as follows:
\begin{itemize}
	\item Sec.~\ref{sec:omitted} contains omitted proofs from the main text.
	\item Sec.~\ref{sec:comparison} presents a comparison of the
		spread--acknowledge model and the FJ model for the sum index.
	\item Sec.~\ref{sec:add-experiments} presents further experimental
		evaluation of our model and our algorithms.
\end{itemize}

\section{Omitted Proofs}
\label{sec:omitted}

\subsection{Proof of Lemma~\ref{lem:equivalence}}
\label{sec:proof-equivalence}
Before we present the formal proof, we first present a proof sketch.
\begin{proof}[Proof Sketch]
	We prove the first claim by induction over the number of rounds.
	In round~$0$, we assumed that both models were initialized with
	the same innate opinions $\begop$ and the same seed nodes. Thus, the
	distributions are deterministic and identical. Now consider round $t>0$.
	For the \spread we note that Phase~I does not change the innate opinions
	and it has no impact on Phase~II.  Therefore, we can
	ignore Phase~I and the model coincides with the two-stage model.
	Now the claim follows from the induction hypothesis.
	
	The second claim essentially follows from the fact that the expressed
	equilibrium opinions $\z^*$ and $\zz^*$ are deterministic
	transformations of $\s$ and $\ss$. 
	However, there is a	subtlety:   
	While in the two-stage model, all opinion updates (as per Equation~\eqref{eq:opinion-round}) are performed
	based on the final vector of innate opinions $\ss$, this is not the
	case in the \spread. In the \spread, it is possible that the innate
	opinions change from round to round and, therefore, it could be possible
	that the opinion updates from early rounds might alter $\z^*$ (as
	per Equation~\eqref{eq:opinion-equilibrium}).  However, we show that this is
	not the case since the equilibrium opinions $\z^*$  are independent from how
	the initial expressed opinions $z^{(0)}$ were initialized in
	Equation~\eqref{eq:opinion-round}. Thus, we can ignore all
	opinion updates that were performed while the innate opinion
	vector did not correspond to $\s$. Now, since $\s$ and $\ss$ follow
	the same distribution and since $\z^*$ and $\zz^*$ are deterministic
	transformations of $\s$ and $\ss$, we obtain the claim.
\end{proof}

Now we proceed to the formal proof of the lemma.

\begin{proof}
    First, as mentioned above, Phase~I does not change the innate opinions
	and it has no impact on Phase~II.  
     Therefore, both models must generate the same distribution over the states of the users. 
	This implies that the distribution of $\s$ and $\ss$ are the same at the end
	of the process for both models.
    The different part is, when the information spread ceases, the two-stage model 
    and \spread may have different intermediate expressed opinions. 
    Therefore, we will prove that the equilibrium expressed opinions only depend on the 
    innate opinions, rather than the intermediate expressed opinions.

   Next, we present an equivalent formulation of the update rule of the \fj.
   The new formulation allows us to give a simple proof that when the innate
   opinions change over time, then the final expressed
   opinion $\z^*$ only depends on the last vector of innate opinions; in
   particular, $\z^*$ is independent of the intermediate expressed and innate
   opinions.

   Let $\Stoc$ be an $n \times n$ row stochastic matrix, let $\Lambda$ be a $n \times n$ diagonal matrix.
   Specifically, let
   $$\eStoc{i}{j} = \frac{w_{i,j}}{\sum_{j\in N(i)}w_{i,j}}$$
   and let
   $$\Lambda_{ii} = \frac{\sum_{j\in N(i)} w_{i,j}}{1 + \sum_{j \in N(i)}
	   w_{i,j}}.$$
   The update rule of the \fj from Equation~\eqref{eq:opinion-round} can thus be 
   equivalently formulated as:
   \begin{equation}
       \label{eq:opinion-round-2}
       \finop^{(t+1)} = \Lambda \Stoc \finop^{(t)} + (\ID - \Lambda)\begop.
   \end{equation}

   The following lemma shows how the expressed opinions evolve when the vector
   of innate opinions~$\begop$ is fixed. Note that in the lemma we do not make
   the assumption that $\finop^{(t)}$ was obtained from the same vector of
   innate opinions $\begop$.
   \begin{lemma}
      \label{lemma:opinion-model-1}
       Consider the vector of innate opinions $\begop$ and the expressed opinions
	   $\finop^{(t)}$ at time step $t$. 
	   Suppose that we perform $T$ additional time steps while $\begop$ is
	   fixed. Then based 
       on the updating rule of Equation~\eqref{eq:opinion-round}, 
       the expressed opinion becomes
       \begin{equation}
           \label{eq:opinion-round-2-intermediate}
       \finop^{(t+T)} = (\Lambda \Stoc)^T \finop^{(t)} + \sum_{i=0}^{T-1} (\Lambda \Stoc)^i (\ID - \Lambda) \begop. 
       \end{equation} 
   \end{lemma}
   \begin{proof}
       This follows from applying Equation~\eqref{eq:opinion-round-2} $T$ times.
   \end{proof}
  
   As a corollary we obtain that when the innate opinions
   change over time, then the final expressed opinions only depend on the
   \emph{last} vector of innate opinions.
   \begin{corollary}
      \label{corollary:opinion-model-2}
	  Assume we perform a sequence of modification on the innate opinions
	  during some finite time span and let the innate opinions after the last
	  modification be $\s$. Then final expressed opinion vector is given by:
      \begin{equation}
       \z^* = (\ID - \Lambda \Stoc)^{-1}(\ID - \Lambda) \s. 
      \end{equation}
  \end{corollary}
  \begin{proof}
      Let the $t$ be time step after which we perform the final modification to
	  the innate opinions. Then we apply Lemma~\ref{lemma:opinion-model-1} and
      it follows that $\z^{(t+T)} = (\Lambda \Stoc)^T \z^{(t)} + \sum_{i=0}^{T-1} (\Lambda \Stoc)^i (\ID - \Lambda) \s$
	  for all $T\geq0$. 
      Since the absolute value of the largest eigenvalue of $\Lambda \Stoc$ is strictly smaller than 1, 
      $\lim_{T \rightarrow \infty} (\Lambda \Stoc)^T = 0$ 
      and $\lim_{T \rightarrow \infty} \sum_{i=0}^T(\Lambda \Stoc)^i = (\ID - \Lambda \Stoc)^{-1}$. 
      It follows that
	  \begin{align*}
	  	\z^*
		= \lim_{T \rightarrow \infty} \z^{(T+t)}
		= (\ID - \Lambda \Stoc)^{-1}(\ID - \Lambda) \s.
		\quad \qedhere
	  \end{align*}
  \end{proof}

  To prove that $\z^*$ and $\zz^*$ have the same distribution, recall that we
  already argued that $\s$ and $\ss$ have the same distribution. Now observe
  that $\z^*$ and $\zz^*$ are deterministic transformations of $\s$ and $\ss$
  and by Corollary~\ref{corollary:opinion-model-2} they converge to the same
  vector for a given vector of innate opinions. This finishes the proof of the
  lemma.
\end{proof}

\subsection{Proof of Lemma~\ref{lemma:hard-computing}}
The lemma follows from the fact that our model generalizes the independent
cascade model. More concretely, suppose that $\begop = 0$, i.e., all initial
innate opinions are initialized to 0. Furthermore, suppose that we set
$\varepsilon=1$ and $\delta=1$. Then it can be seen that $\Delta \ebegop{u}=1$
for all $u \in V$ and, since $\delta=1$, all vertices have the states inactive,
ignore or spread.  Now if we set $\MasIdx{\mathcal{M}}$ to the $n \times n$
identity matrix then we observe that computing $\Exp[F(\cdot)]$ in our model is
the same as computing the influence spreading through the \emph{independent
cascade model}. Since the latter is known to be \NPhard~\cite{kempe2015maximizing} and \SPhard~\cite{chen2010scalable}, we obtain
the lemma.

\subsection{Proof of Lemma~\ref{lemma:fr-1}}
Let $I_g(S)$ be the set of all nodes that can be reached from any node of $S$ in
$g$ through a live path.  Let $\pathworld{S}{u} = 1$ if there is a live path
from any node in $S$ to $u$ in $g$ and $\pathworld{S}{u} = 0$ otherwise. Let
$\randomrrp_{g, u}$ be a RR-set of $u$ in $g$. Then $F(S)$ on $g$ can be
formulated as:

\begin{equation*}
  \begin{aligned}
&\Exp[F_g(S)] \\
&=\sum_{u,v \in V} \frac{1}{n} \IND[\pathworld{S}{u} = 1] \LinGain{u} + \IND[\pathworld{S}{u} \pathworld{S}{v} = 1] \PolGain{u}{v} \\
&= \sum_{u,v \in V} \frac{1}{n} \IND[\randomrrp_{g, u} \cap S \neq \emptyset] \LinGain{u} + \IND[\randomrrp_{g, u} \cap S \neq \emptyset, \randomrrp_{g, v} \cap S \neq \emptyset] \PolGain{u}{v}
  \end{aligned}
\end{equation*}

Now let $g$ be randomly generated possible world. We denote the distribution of
$g$ by $\c+G$ and write $g \sim \c+G$. Then we obtain:
\begin{equation}
\label{eq:bridge}
    \begin{split}
        &\Exp[F(S)] \\
		    &= \Exp_{g \sim \c+G} [\sum_{u,v \in V} (\IND[\randomrrp_{u} \cap S \neq \emptyset] \LinGain{u} + n\IND[\randomrrp_{u} \cap S \neq \emptyset, \randomrrp_{v} \cap S \neq \emptyset] \PolGain{u}{v})] \\
        &= \sum_{u,v \in V} \Exp_{g \sim \c+G} [ \IND[\randomrrp_{u} \cap S \neq \emptyset] \LinGain{u} + n\IND[\randomrrp_{u} \cap S \neq \emptyset, \randomrrp_{v} \cap S \neq \emptyset] \PolGain{u}{v}] \\
        &= n \Exp_{(u, v) \sim V^2, g \sim \c+G} \estimator \\
        &= n \Exp_{(u, v) \sim V^2, g \sim \c+G}[\frac{\sum_{(\randomrrp_u, \randomrrp_v) \in \c+R} \estimator}{\abs{\c+R}}] \\
        &= n \Exp_{(u, v) \sim V^2, g \sim \c+G}[\fracrrp{S}]
    \end{split}
\end{equation}

Thus $n\fracrrp{S}$ with sample $\c+R$ is an unbiased estimator of $\Exp[F(S)]$. 

\subsection{Proof of Lemma~\ref{lem:approximation}}
Before we prove the lemma, we first introduce martingales and some concentration
inequalities. We prove the lemma at the end of the subsection.

Let $x_1, x_2, \ldots, x_{\theta}$ be a sequence of random variables. 
For each $1\leq i \leq \theta$, we set $x_i = \estimator$.
Now observe that $\fracrrp{S} = \frac{\sum_{i=1}^\theta x_i}{\theta}$
and according to Lemma~\ref{lemma:fr-1}, $\set+E[F(S)] = \frac{n}{\theta}\cdot\set+E[\sum_{i=1}^\theta x_i]$. 
Similar to IMM~\cite{tang2015influence} and TDEM~\cite{aslay2018maximizing}, 
we determine the sample size by concentration laws of martingales. Let us first
introduce martingales. 

\begin{definition}[Martingale]
   A sequence of random variables $Y_1, Y_2, Y_3, \ldots$ 
   is a \emph{martingale} if and only if $\set+E[\abs{Y_i}] < +\infty$
   and $\set+E[Y_i \mid Y_1, Y_2, \ldots, Y_{i-1}] = Y_{i-1}$ 
   for any $i$.
\end{definition}

Note that as the generation of an \rr-set $\randomrrp_i$ is independent of $\randomrrp_1, \ldots, \randomrrp_{i-1}$, 
we have $\Exp[x_i \mid x_1, \ldots, x_{i-1}] = \frac{\Exp[F(S)]}{n^2}$. 
Now let $x = \frac{1}{n}\Exp[F(X)]$ and $M_j= \sum_{z=1}^{j}(x_z - x)$.
Then $\Exp[M_j] = 0$ and
\begin{displaymath}
 \begin{aligned}
\Exp[M_j \mid M_1, \ldots, M_{j-1}] &= \Exp[M_{j-1} + x_j - x \mid M_1, \ldots, M_{j-1}] \\
&= M_{j-1} - x + \Exp[x_j] \\
&= M_{j-1}.
\end{aligned}
\end{displaymath}
Therefore, the sequence $M_1, \ldots, M_{\theta}$ is a martingale.
Next, we restate a concentration inequality for martingales by 
McDiarmid~\cite{mcdiarmid1998concentration}, which we cite from Chung and Lu~\cite{chung2006concentration}.
\begin{lemma}[Theorem $6.1$ in~\cite{chung2006concentration}]
\label{lemmaMartIneqLiterature-1}
	Let $Y_1, Y_2, \ldots $ be a martingale,
	such that $Y_1 \le a$, $\Var[Y_1] \le b_1$, $\lvert Y_z - Y_{z-1} \rvert \le a$ for $z \in [2,j]$, and
	\[
	\Var[Y_z \mid Y_1, \ldots, Y_{z-1}] \le b_j, \text{ for } z \in [2,j],
	\]
	where $\Var[\cdot]$ denotes the variance. Then, for any $\gamma > 0$
	\begin{displaymath}
	\begin{aligned}
	\Pr(Y_j - \Exp[Y_j] \ge \gamma) \le \exp\left(-\dfrac{\gamma^2}{2(\sum_{z = 1}^{j} b_z + a \gamma /3)} \right)
	\end{aligned}
	\end{displaymath}
\end{lemma}

We now use Lemma~\ref{lemmaMartIneqLiterature-1} 
to get the concentration result for the martingale
$M_1, \ldots, M_{\theta}$.
As we define the random variables $x_i$ as $x_i = \estimator = \IND[(\randomrrp_u \cap S) \neq \emptyset] \LinGain{u}   + n\IND[(\randomrrp_u \cap S) \neq \emptyset, (\randomrrp_v \cap S) \neq \emptyset] \PolGain{u}{v}$, 
let $\chi = \max_{u, v} \abs{\LinGain{u} + n\PolGain{u}{v}}$, then $\max x_i \leq \chi$. 

Since $x_j \in [-\chi,\chi]$ for all $j \in \{1,\dots,\theta\}$,
we have $\lvert M_1 \rvert= \lvert x_1 - x \rvert \le 2\chi$ and
$\lvert M_j - M_{j-1} \rvert \le 2\chi$ for any $j \in [2,\theta]$.
Additionally, it holds that
 $\Var[M_1] = \Var[x_1]$, and for any $j \in \{2,\dots,\theta\}$ we have that:
\begin{displaymath}
 \begin{aligned}
\Var[M_j  \mid M_1, \ldots, M_{j-1}] 
&=  \Var[M_{j-1} + x_j - x \mid M_1, \ldots, M_{j-1}] \\
&=\Var[x_j \mid M_1, \ldots, M_{j-1}]  \\
&=\Var[x_j].
\end{aligned}  
\end{displaymath}

And for $\Var[x_j]$ we have that:
\begin{displaymath}
  \begin{aligned}
    \Var[x_j] &= \Exp[x_j^2] - \Exp[x_j]^2 
    \leq \chi x - x^2
    \leq \chi x.
  \end{aligned} 
\end{displaymath}

By using Lemma~\ref{lemmaMartIneqLiterature-1}, for
$M_{\theta}=  \sum_{j=1}^{\theta}(x_j - x)$, with
$\Exp[M_{\theta}] = 0$, $a = 2\chi$, $b_j=\chi x$, for $j  = 1, 2, \ldots ,\theta$, and
$\gamma = \delta \theta x$, we have the following corollary.

\begin{corollary}
\label{corr:ourIneq1}
For any $\delta > 0$,
\begin{displaymath}
\begin{aligned}
\Pr[\sum_{j=1}^{\theta} x_j - \theta x \ge \delta \theta x] \le \exp\left(-\dfrac{\delta^2}{2\chi(\frac{2\delta}{3} + 1)} \, \theta x \right).
\end{aligned}   
\end{displaymath}
\end{corollary}
Moreover, for the martingale $-M_1, \ldots, -M_{\theta}$,
we similarly have $a = 2\chi$ and $b_j = \chi x$ for $j  = 1, \ldots ,\theta$.
Note also that $\Exp[-M_{\theta}] = 0$.
Hence, %
for $-M_{\theta}=  \sum_{j=1}^{\theta}(x - x_j)$ and 
$\gamma = \delta \theta x$ we obtain a corollary similar to the one above.
\begin{corollary}\label{corr:ourIneq2}
For any $\delta > 0$,
\begin{displaymath}
\begin{aligned}
\Pr[\sum_{j=1}^{\theta} x_j - \theta x \le -\delta \theta x ] \le \exp\left(-\dfrac{\delta^2}{2\chi(\frac{2\delta}{3} + 1)} \, \theta  x \right).
\end{aligned}  
\end{displaymath}
\end{corollary}

\begin{proof}[Proof of Lemma~\ref{lem:approximation}]
Using Corollaries~\ref{corr:ourIneq1} and~\ref{corr:ourIneq2} and letting 
$\delta = \dfrac{\epsilon\, \optvalue}{2 n x}$, we obtain
   \begin{displaymath}
   \begin{aligned}
     &\Pr[\abs{n\fracrrp{X} - \Exp[F(X)]} \geq \frac{\epsilon}{2}\optvalue] \\
	 &= \Pr[\abs{\sum_{i=1}^{\theta}x_i - \theta x} {\geq \frac{\theta\epsilon}{2n}}\optvalue] \\ 
&\le 2 \exp\left(-\dfrac{\delta^2}{2\chi(\frac{2\delta}{3} + 1)} \, \theta x \right)  \\ 
& = 2 \exp\left(-\dfrac{3 \epsilon^2 \, \optvalue^2}{8\,n \chi( \epsilon\, \optvalue + 3 n x)}  \, \theta \right) \\
& \le 2 \exp\left(-\dfrac{3 \epsilon^2 \, \optvalue^2}{8\,n \chi ( \epsilon\, \optvalue + 3 \optvalue)}  \, \theta \right) \\
&= 2 \exp\left(-\dfrac{\epsilon^2 \, \optvalue}{8\,n \chi ( \frac{\epsilon}{3} + 1)}  \, \theta \right), 
   \end{aligned} 
  \end{displaymath}  
where the last inequality above follows from the fact that $n x \le \optvalue$.
Finally, by requiring
\begin{align*}
2 \exp\left(-\dfrac{\epsilon^2 \, \optvalue}{8\,n \chi ( \frac{\epsilon}{3} + 1)}  \, \theta \right)  \le \dfrac{1}{n^{\ell} \, \binom{n}{k}},
\end{align*}
we obtain the lower bound on $ \theta$.
\end{proof}

\subsection{Proof of Theorem~\ref{thm:greedy-marketing}}
\label{sec:app-sum-index}
First, we prove that the $\fracrrp{S}$ is submodular and monontone under the 
setting of maximizing the sum of expressed opinions. Recall that we set
$$\fracrrp{S} = \frac{\sum_{\randomrrp_u \in \rrpsample} \Delta \ebegop{u} \IND[(\randomrrp_u \cap S) \neq \emptyset] }{\abs{\rrpsample}}.$$
The submodularity and monotonicity follow from the weighted version of
maximum set coverage problem, and these two properties ensure that greedy
algorithm on $\fracrrp{S}$ achieves an approximation ration of $1-\frac{1}{e}$. 
We present the details below. 

Let $S \subset T \subseteq U$, let $v \in U \setminus T$.

First, we show that $\fracrrp{S}$ is monontone: Since for any $R_u \in \rrpsample$, $\IND[(\randomrrp_u \cap T) \neq \emptyset] - \IND[(\randomrrp_u \cap S) \neq \emptyset] \geq 0$, 
thus, $\fracrrp{T} - \fracrrp{S} \geq 0$. 

Second, we show that $\fracrrp{S}$ is submodular: Since for any $R_u \in \rrpsample$, 
$(\IND[(\randomrrp_u \cap (T\cup\{v\})) \neq \emptyset] - \IND[(\randomrrp_u \cap T) \neq \emptyset]) \leq (\IND[(\randomrrp_u \cap (S\cup\{v\})) \neq \emptyset] - \IND[(\randomrrp_u \cap S) \neq \emptyset])$, 
thus, $\fracrrp{T \cup \{v\}} - \fracrrp{T} - (\fracrrp{S \cup \{v\}} - \fracrrp{S}) \leq 0$. 

Now we apply Lemma~\ref{lem:y-size-main} and Lemma~\ref{lem:approximation} to
obtain a lower bound on $\optvalue$ and the sampling size $\theta$. This ensures that
Equation~\eqref{equation:sample-error-bound} holds.

Note that there are in total $\binom{n}{k}$ different sets of size $k$ and thus
there are $\binom{n}{k}$ ways to choose $S$. Applying a union bound to
Equation~\eqref{equation:sample-error-bound}, we obtain that
   \[\set+P[\abs{n\cdot \fracrrp{S} - \Exp[F(S)]} \geq \frac{\epsilon}{2}\cdot \optvalue \mid \mbox{for any $S$ of size k}] \leq \frac{1}{n^{\ell}}.\]

   Let $S^{G}$ denote the greedy solution for $F_{\c+R}(\cdot)$, 
   $S^{+}$ the optimal solution for $F_{\c+R}(\cdot)$, 
   and $S^*$ the optimal solution for $\Exp[F(\cdot)]$.
   Then the following inequality holds with probability at least $1 - n^{-\ell}$:
   \begin{equation*}
      \begin{aligned}
         \Exp[F(S^G)] &\geq n \fracrrp{S^G} - \frac{\epsilon}{2} \optvalue \\
         &\geq n (1 - \frac{1}{e}) \fracrrp{S^+} - \frac{\epsilon}{2} \optvalue \\
         &\geq n (1 - \frac{1}{e}) \fracrrp{S^*} - \frac{\epsilon}{2} \optvalue \\
         &\geq (1 - \frac{1}{e}) (\Exp[F(S^*)] - \frac{\epsilon}{2}\optvalue) - \frac{\epsilon}{2} \optvalue \\
         &= (1 - \frac{1}{e})\optvalue - (1 - \frac{1}{e}) \frac{\epsilon}{2} \optvalue - \frac{\epsilon}{2} \optvalue \\
         &\geq (1 - \frac{1}{e} - \epsilon) \optvalue.
      \end{aligned}
   \end{equation*}
Thus implementing the greedy algorithm on $F_{\c+R}(\cdot)$ directly gives us a
$(1 - \frac{1}{e} - \epsilon)$-approximation with high probability. 
	
\subsection{Proof of Lemma~\ref{lemma:max-sum-opinions}}
It is well-known that for the Laplacian $\laplacian$ of undirected graphs it
holds that $\ind^\intercal (\ID + \laplacian)^{-1} = \ind^{\intercal}$.
Now we obtain:
\begin{equation}
  \begin{aligned}
	\arg\max_S \sum_{u \in V} \efinop{u}^*(S)
		&= \arg \max_{S} \ind^\intercal [(\ID + \laplacian)^{-1} (\begop + \Delta \begop \odot \ind(S))]\\
		&=  \arg \max_{S} \ind^\intercal [\Delta \begop \odot \ind(S)] \\
		&=  \arg \max_{S} \sum_{u \in V} \Delta \ebegop{u} \ind_u(S) \\
		&= \arg\max_S F(S).
  \end{aligned}
\end{equation}

\subsection{Proof of Lemma~\ref{lemma:approx-max-sum-opinions}}
  Let $S \subset T \subseteq U$, let $v \in U \setminus T$. 
We first prove that in any possible world $g$, $F_g(S)$ is submodular and monontone.

First, we show that $F_g(S)$ is monontone: Since for any $u \in V$, 
  $\ind_u(T) - \ind_u(S) \geq 0$, 
thus, $F_g(T) - F_g(S) \geq 0$. 

Second, we show that $F_g(S)$ is submodular: Since for any $u \in V$, 
$\ind_u(T\cup\{v\})- \ind_u(T) \leq \ind_u(S\cup\{v\})- \ind_u(S)$, 
thus, $F_g(T \cup \{v\}) - F_g(T) - (F_g(S \cup \{v\}) - F_g(S)) \leq 0$. 

Now $\Exp[F(S)]$ is submodular and monontone as it is a linear combination of
the $F_g(S)$, and the weight for each part is non-negative. 
The approximation ratio of greedy algorithm on $\Exp[F(S)]$ is thus
$1 - \frac{1}{e}$~\cite{hochbaum1996approximating}.

\subsection{Proof of Lemma~\ref{lem:y-size-main}}
\emph{Proof of~(1).}
Let $X$ be a
set of seed nodes with 
$\abs{X} = k$ and let $x = \frac{1}{n} \Exp[F(X)]$.
We show that if $\opt < y\chi$ then
$n\fracrrp{X} \geq (1 + \epsilon_2) \,y \chi$
with probability at most $\frac{n^{-\ell}}{\log_2 n \binom{n}{k}}$.  

Note that $\optvalue < y \chi$ implies that $x < \frac{\opt}{n} < \frac{y \chi}{n}$,
and $1 < \frac{y \chi}{xn}$. Notice that by construction $y \leq n$ since in the
algorithm we set $y \leftarrow n/2^i$. 
Let $\delta = \frac{1 + \epsilon_2}{n x} y \chi - 1$ and observe that
$\delta > \frac{\epsilon_2 y \chi}{n x}$. 
Then, by using Corollary~\ref{corr:ourIneq1}, we have:
  \begin{displaymath}
    \begin{aligned}
      \Pr[n\fracrrp{X} \geq (1 + \epsilon_2)y] &= \Pr\left[\theta \fracrrp X - \theta x \geq \theta x \left(\frac{(1 + \epsilon_2)y\chi}{nx} - 1\right)\right] \\
      &= \Pr[\theta \fracrrp X - \theta x \geq \theta x \delta] \\  
      &\leq \exp \left(-\frac{\delta}{2 \chi (\frac{2}{3} + \frac{1}{\delta})}\theta x\right) \\
      &\leq \exp \left(-\frac{\delta}{2 \chi (\frac{2}{3} + \frac{1}{\epsilon_2})}\theta x\right) \\
      &\leq \exp \left(-\frac{\epsilon_2^2}{\frac{4}{3}\epsilon_2 + 2}\frac{y}{n}\theta\right) \\
      &\leq \frac{n^{-\ell}}{\log_2(n)\cdot \binom{n}{k}}.
    \end{aligned}
  \end{displaymath}
Finally by a union bound over all $\binom{n}{k}$ choices of $X$, we conclude
that if $\opt < y \chi$, then $n\fracrrp{\t+X} < (1 + \epsilon_2)y$ with
probability at least $1 - \frac{n^{-\ell}}{\log_2 n}$.

\emph{Proof of~(2).}
  Let $X$ be a set of seed nodes with $\abs{X} = k$ and let $x = \frac{1}{n} \Exp[F(X)]$.
  Assume that 
  $\opt \geq y\chi$. Note that this implies $\frac{\optvalue}{nx} \geq 1$.
  Now we will show that if $\opt \geq y \chi$ then $n \, \fracrrp{\t+X} > (1 + \epsilon_2)\optvalue$
  with probability at most $\frac{n^{-\ell}}{\log_2 n \binom{n}{k}}$. 

  Let $\delta = \frac{(1 + \epsilon_2) \optvalue}{n x} - 1$. Then we have $\delta \geq \epsilon_2$. 
  By using Corollary~\ref{corr:ourIneq2}, we obtain that:
  \begin{displaymath}
    \begin{aligned}
      &\Pr[n \, \fracrrp{X} > (1 + \epsilon_2)\opt] \\
	  &= \Pr\left[\theta \fracrrp X - \theta x > \theta x \left(\frac{(1 + \epsilon_2)\opt}{nx} - 1\right)\right] \\
      &\leq \exp \left(-\frac{\delta}{2 \chi (\frac{2}{3} + \frac{1}{\delta})}\theta x\right) \\
      &\leq \exp \left(-\frac{\delta}{2 \chi (\frac{2}{3} + \frac{1}{\epsilon_2})}\theta x\right) \\
      &\leq \exp \left(-\frac{\epsilon_2 \optvalue}{2 \chi n (\frac{2}{3} + \frac{1}{\epsilon_2})}\theta \right) \\
      &\leq \exp\left(-\frac{\epsilon_2^2}{\frac{4}{3}\epsilon_2 + 2}\frac{y}{n}\theta\right) \\
      &\leq \frac{n^{-\ell}}{\log_2(n) \cdot \binom{n}{k}}
    \end{aligned}
  \end{displaymath}
By taking a union bound over all $\binom{n}{k}$ choices of $X$, we reach the
desired result.

\subsection{Proof of Theorem~\ref{theorem:sand-lu}}
Our proof has four steps. 
First, we show that the matrices $\MasIdx{\ConIdx{}}$ and
$\MasIdx{\DisConIdx{}}$ have non-negative entries.
Second, we show that 
the objective function $\mu_0(S)$ is monotone but neither submodular nor supermodular.
Third, we show that the upper and lower bounds we consider are monotone and submodular.
Fourth, we apply the sandwich method.

\emph{Step~I: The matrices $\MasIdx{\ConIdx{}}$ and $\MasIdx{\DisConIdx{}}$ have
	non-negative entries.}
As we defined in Table~\ref{table:index}, 
$\MasIdx{\ConIdx{}} = (\laplacian + \ID)^{-2}$
and $\MasIdx{\DisConIdx{}} = (\laplacian + \ID)^{-1}$. 
We now show that all entries of $\MasIdx{\DisConIdx{}}$ 
are non-negative. Note that this implies that the entries of
$\MasIdx{\ConIdx{}}$ are non-negative as well since 
$\MasIdx{\ConIdx{}} = \MasIdx{\DisConIdx{}}^2$. 

Now observe that:
\begin{equation*}
  \begin{aligned}
    (\laplacian + \ID)^{-1} &= [(\m+D + \ID) (\ID - (\ID + \m+D)^{-1} \m+W)]^{-1} \\
    &=(\ID - (\ID + \m+D)^{-1} \m+W)^{-1}(\m+D + \ID)^{-1} \\
    &=[\sum_{i=0}^{\infty} (\ID + \m+D)^{-i}\m+W^i] (\m+D + \ID)^{-1}.
  \end{aligned}
\end{equation*}
The last equation holds since the entries of matrix $(\ID + \m+D)^{-1} \m+W$ are non-negative, and 
the sum of rows and columns are strictly smaller than 1, thus the limit $\lim_{i\rightarrow \infty} (\ID + \m+D)^{-i} \m+W^i = 0$.  
As each matrix $(\ID + \m+D)^{-i}\m+W^i$, and $(\m+D + \ID)^{-1}$ 
are non-negative, their sums and multiplication are non-negative. 
Thus $(\laplacian + \ID)^{-1}$ is non-negative. 

\emph{Step~II: The objective function $\mu_0(S)$ is monotone but neither submodular nor supermodular.}
Let $\MasIdx{\mathcal{M}} \in \{ \MasIdx{\DisConIdx{}}, \MasIdx{\ConIdx{}} \}$, 
our goal is to maximize $\Exp[\s^{\intercal} \MasIdx{\c+M} \s]$ 
by selecting a set of seed nodes $S$. 
Let
$$\mu_0(S) = \Exp[2 \begop^\intercal \MasIdx{\mathcal{M}} \Delta \s + \Delta \s^\intercal \MasIdx{\mathcal{M}} \Delta \s].$$
By substituting $\s=\begop+\Delta\s$ into $\Exp[\s^{\intercal} \MasIdx{\c+M}
\s]$ and ignoring summands which only depends on $\begop$, we have 
$$\arg \max_S \mu_0(S) = \arg \max_S \Exp[\s^{\intercal} \MasIdx{\c+M} \s].$$

Now, we show that $\mu_0(S)$ is neither submodular nor supermodular, 
but monontone in the following Lemma~\ref{lem:non-sub-super}, and Lemma~\ref{lem:mon}. 

\begin{lemma}
  \label{lem:non-sub-super}
  $\mu_0(S)$ is neither submodular nor supermodular.
\end{lemma}
\begin{proof}
$\mu_0(S)$ is not submodular:
We can set the spread probabilities $p_{uv}$ on
all the edges to be $0$, and $\delta = 0$ as well. 
We also set the non-adjusted innate opinions 
$\begop = \m+0$. In this case, $\mu_0(S)$ becomes
$\mu_0(S) = \Exp[\Delta \s^{\intercal} \MasIdx{\mathcal{M}} \Delta \s] = \sum_{u \in S, v \in S} \Delta \ebegop{u} \Delta \ebegop{v} \MasIdx{\mathcal{M}}_{u,v}$.
As in Sec.~\ref{sec:estimating}, we set $m_{u, v} = (\Delta \ebegop{u})^{\intercal} \MasIdx{\mathcal{M}}_{u,v} \Delta \ebegop{v}$. 
Let $S \subseteq T \subseteq U$, and $x \in U \setminus T$, 
$\mu_0(T \cup \{x\}) - \mu_0(T) - (\mu_0(S \cup \{x\}) - \mu_0(S)) = \sum_{u \in T \setminus S} m_{u, x} + \sum_{v \in T \setminus S} m_{x, v} \geq 0$.
The inequality is strict $>$ if any $m_{u, x}$ or $m_{x, v}$, $u, v \in T \setminus S$
is strictly larger than 0. This condition can be satisfied by selecting $v$. 

$\mu_0(S)$ is not supermodular:
We set non-adjusted innate opinions $\begop = \m+0$, 
and $\epsilon$ small enough such that $\Delta \ebegop{u} = \epsilon$ for any $u$. 
We set $w_{uv} = 0$ for any $(u, v) \in E$. 
In addition, we set $\delta = 1$, such that \spread becomes independent cascade model. 
Under these settings, 
$\MasIdx{\mathcal{M}}$ is the identity matrix $\ID$, and 
$\mu_0(S) = \epsilon^2 \sum_{v \in V} \ind_u(S)$. 
Note that $\sum_{u\in V}\ind_u(S)$ is the influence spread in the independent
cascade model with seed set $S$, 
which is submodular~\cite{kempe2015maximizing}. 
Let $\sigma(S) = \sum_{u\in V}\ind_u(S)$. 
Let $S \subseteq T \subseteq U$, and $x \in U \setminus T$, 
from the submodularity of $\sigma(S)$ it follows 
$\mu_0(T \cup \{x\}) - \mu_0(T) - (\mu_0(S \cup \{x\}) - \mu_0(S)) = 
\epsilon^2 \sum_{u \in V} (\sigma(T \cup \{x\}) - \sigma(T)) - (\sigma(S \cup \{x\}) - \sigma(S)) \leq 0$.
Now, let $S = \emptyset$, and select $T$ and $\{x\}$ such that $\sigma(T) =
\abs{V}$ and $\sigma(\{x\}) > 0$.
The inequality is strict. 
\end{proof}

\begin{lemma}
  \label{lem:mon}
  $\mu_0(S)$ is monontone.
\end{lemma}
\begin{proof}
  Let $S \subseteq T$. Then
  $$\mu_0(T) - \mu_0(S) = \sum_{u, v \in V} m_{u, v} (\ind_{u}(T) \ind_{v}(T) - \ind_{u}(S) \ind_{v}(S)) \geq 0$$
  since $\ind_{v}(T) \geq \ind_{v}(S)$ for any $v \in V$. 
\end{proof}

\emph{Step~III: The upper and lower bound are monotone and submodular.}
For 
$\mu_L(S) = \Exp[2 \begop^\intercal \MasIdx{\mathcal{M}} \Delta \s]$, 
$\mu_U(S) = \Exp[2 \begop^\intercal \MasIdx{\mathcal{M}} \Delta \s + \Delta \s^\intercal \MasIdx{\mathcal{M}}^U \Delta \s]$,
we prove that both $\mu_L(S)$ and $\mu_U(S)$ are submodular. 

\begin{lemma}
  \label{lem:sub-two}
  $\mu_L(S)$ and $\mu_U(S)$ are submodular. 
\end{lemma}
\begin{proof}
  We use the notation defined in Sec.~\ref{sec:estimating}, i.e., we set
$\LinGain{u} = (2\begop^\intercal \MasIdx{\mathcal{M}})_u \Delta \ebegop{u}$ and
$\PolGain{u}{v}=(\Delta \ebegop{u})^\intercal \MasIdx{\mathcal{M}}_{u,v} \Delta
\ebegop{v}$. 
Then, $\mu_L(S) = \sum_u \LinGain{u} \ind_u(S)$, and $\mu_U(S) = \sum_u (\LinGain{u} + \sum_{v\in V}\MasIdx{\mathcal{M}}_{u,v}\Delta \ebegop{u}^2) \ind_{u}(S)$. 
Let $S \subseteq T \subseteq U$, and $x \in U \setminus T$, since 
for any $u \in V$, 
$\ind_u(S \cup \{x\}) - \ind_u(S) \geq \ind_u(T \cup \{x\}) - \ind_u(T)$, 
and the weight $\LinGain{u}$, $(\LinGain{u} + \sum_{v\in V}\MasIdx{\mathcal{M}}_{u,v}\Delta \ebegop{u}^2)$ are positive. 
We have $\mu_L(S \cup \{x\}) - \mu_L(S) \geq \mu_L(T \cup \{x\}) - \mu_L(T)$ and 
$\mu_U(S \cup \{x\}) - \mu_U(S) \geq \mu_U(T \cup \{x\}) - \mu_U(T)$. 
It follows that $\mu_L(S)$ and $\mu_U(S)$ are submodular. 
\end{proof}

\emph{Step~IV: Application of the sandwich method.}
Similarly with Lemma~\ref{lem:mon}, we can also see that $\mu_L(S)$ 
and $\mu_U(S)$ are monotone. 
Since $\mu_0(S)$ is monotone, and both $\mu_L(S)$ and $\mu_U(S)$ are submodular and monontone, 
we can apply the sandwich approximation scheme proposed by Lu et al.~\cite[Theorem 9]{lu2015competition},
and it gives us the approximation result.

%% file: tikz/basline/con-1.tex
\pgfplotsset{compat=1.3}
\begin{tikzpicture}

\begin{groupplot}[group style={group size=1 by 1}]
\nextgroupplot[
tick pos=both,
axis lines=left, xtick=\empty, 
x grid style={white!69.0196078431373!black},
yticklabel style = {xshift= 0.8 ex},
xticklabel style = {yshift=-1.6 ex, rotate=45.0,anchor=east},
xmin=-0.2, xmax=4.2,
xtick style={color=black},
xtick={0,1,2,3,4},
xticklabels={\Netscience,\NipsEgo,\PagesGovernment,\Anybeat,\Gplus},
y grid style={white!69.0196078431373!black},
ymin=0.0086326326214931, ymax=0.371832140442839,
ytick style={color=black},
ytick={0,0.05,0.1,0.15,0.2,0.25,0.3,0.35,0.4},
yticklabels={
  \(\displaystyle {0.00}\),
  \(\displaystyle {0.05}\),
  \(\displaystyle {0.10}\),
  \(\displaystyle {0.15}\),
  \(\displaystyle {0.20}\),
  \(\displaystyle {0.25}\),
  \(\displaystyle {0.30}\),
  \(\displaystyle {0.35}\),
  \(\displaystyle {0.40}\)
},
ylabel={rel. increase \disconidx},
legend style={font = \Huge, draw=none},
legend style={at={(1,1), anchor=north east}},
line width = 0.50 mm, 
tick style={line width=0.50mm},
ytick scale label code/.code={$\times 10^{-2}$},
]
\path [draw=black]
(axis cs:-0.1,0.027083924880246)
--(axis cs:-0.1,0.02911814058553);

\path [draw=black]
(axis cs:0.9,0.025141701158827)
--(axis cs:0.9,0.025710236216547);

\path [draw=black]
(axis cs:1.9,0.031158369462285)
--(axis cs:1.9,0.031560900049333);

\path [draw=black]
(axis cs:2.9,0.030308172053967)
--(axis cs:2.9,0.030552990529017);

\path [draw=black]
(axis cs:3.9,0.0283791693649424)
--(axis cs:3.9,0.0285657268246916);

\path [draw=red]
(axis cs:0,0.121864417286312)
--(axis cs:0,0.127747632131382);

\path [draw=red]
(axis cs:1,0.344059407961135)
--(axis cs:1,0.351718594347667);

\path [draw=red]
(axis cs:2,0.121620025118767)
--(axis cs:2,0.124210184756131);

\path [draw=red]
(axis cs:3,0.267089061094514)
--(axis cs:3,0.270784255916584);

\path [draw=red]
(axis cs:4,0.339521823990702)
--(axis cs:4,0.342283532540786);

\path [draw=blue]
(axis cs:0.1,0.14018786926178)
--(axis cs:0.1,0.150962996047816);

\path [draw=blue]
(axis cs:1.1,0.343687929460303)
--(axis cs:1.1,0.355323071905505);
 
\path [draw=blue]
(axis cs:2.1,0.14690116651478)
--(axis cs:2.1,0.14916931869781);

\path [draw=blue]
(axis cs:3.1,0.274113704050368)
--(axis cs:3.1,0.279272859840664);

\path [draw=blue]
(axis cs:4.1,0.341213418804761)
--(axis cs:4.1,0.344730788860937);

\addplot [black, mark=asterisk, mark size=3, mark options={solid}, only marks]
table {%
-0.1 0.0281010327328888
0.9 0.025425968687687
1.9 0.031359634755809
2.9 0.030430581291492
3.9 0.028472448094817
};
\addplot [red, mark=asterisk, mark size=3, mark options={solid}, only marks]
table {%
0 0.124806024708847
1 0.347889001154401
2 0.122915104937449
3 0.268936658505549
4 0.340902678265744
};
\addplot [blue, mark=asterisk, mark size=3, mark options={solid}, only marks]
table {%
0.1 0.145575432654798
1.1 0.349505500682904
2.1 0.148035242606295
3.1 0.276693281945516
4.1 0.342972103832849
};

\end{groupplot}

\end{tikzpicture}

%% file: tikz/basline/con-2.tex
\pgfplotsset{compat=1.3}
\begin{tikzpicture}

\begin{groupplot}[group style={group size=2 by 1}]
\nextgroupplot[
tick pos=both,
axis lines=left, xtick=\empty, 
x grid style={white!69.0196078431373!black},
yticklabel style = {xshift=-1 ex},
xticklabel style = {yshift=-1ex},
xmin=-0.2, xmax=4.2,
xtick style={color=black},
xtick={0,1,2,3,4},
xticklabel style={rotate=45.0,anchor=east},
xticklabels={\Netscience,\NipsEgo,\PagesGovernment,\Anybeat,\Gplus},
y grid style={white!69.0196078431373!black},
ymin=-0.0028395685807644, ymax=0.0955551554273084,
ytick style={color=black},
ytick={-0.02,0,0.02,0.04,0.06,0.08,0.1},
ylabel={rel. increase \disconidx},
legend style={draw=none},
legend style={at={(1.1,1), anchor=north east}},
line width = 0.50 mm, 
tick style={line width=0.50mm},
ytick scale label code/.code={$\times 10^{-2}$},
]
\path [draw=black]
(axis cs:-0.1,0.003415170293928)
--(axis cs:-0.1,0.00584530673014);

\path [draw=black]
(axis cs:0.9,0.00561081322043)
--(axis cs:0.9,0.006078493127502);

\path [draw=black]
(axis cs:1.9,0.001632918874148)
--(axis cs:1.9,0.00221685470642);

\path [draw=black]
(axis cs:2.9,0.004353135360906)
--(axis cs:2.9,0.004804054736506);

\path [draw=black]
(axis cs:3.9,0.006176785879706)
--(axis cs:3.9,0.006598900859466);

\path [draw=red]
(axis cs:0,0.00984346285947)
--(axis cs:0,0.026512664297808);

\path [draw=red]
(axis cs:1,0.081030893540672)
--(axis cs:1,0.090792688664296);

\path [draw=red]
(axis cs:2,0.004600997726438)
--(axis cs:2,0.00785352843588);

\path [draw=red]
(axis cs:3,0.046208115197612)
--(axis cs:3,0.050060758041);

\path [draw=red]
(axis cs:4,0.07092250153131)
--(axis cs:4,0.074345115004608);

\path [draw=blue]
(axis cs:0.1,0.009643635773477)
--(axis cs:0.1,0.021558796494785);

\path [draw=blue]
(axis cs:1.1,0.079558321935674)
--(axis cs:1.1,0.091082667972396);

\path [draw=blue]
(axis cs:2,0.00708863970367)
--(axis cs:2,0.012119522501372);

\path [draw=blue]
(axis cs:3.1,0.048453555145443)
--(axis cs:3.1,0.051269708881937);

\path [draw=blue]
(axis cs:4.1,0.071041626835282)
--(axis cs:4.1,0.074510745036488);

\addplot [black, mark=asterisk, mark size=3, mark options={solid}, only marks]
table {%
-0.1 0.0046302385120344
0.9 0.005844653173966
1.9 0.001924886790284
2.9 0.004578595048706
3.9 0.006387843369586
};
\addplot [red, mark=asterisk, mark size=3, mark options={solid}, only marks]
table {%
0 0.018178063578639
1 0.085911791102484
2 0.006227263081159
3 0.048134436619306
4 0.072633808267959
};
\addplot [blue, mark=asterisk, mark size=3, mark options={solid}, only marks]
table {%
0.1 0.015601216134131
1.1 0.085320494954035
2.1 0.009604081102521
3.1 0.04986163201369
4.1 0.072776185935885
};
\addlegendentry{\Random}
\addlegendentry{\HighDegree}
\addlegendentry{\MaxInfluence}
\end{groupplot}

\end{tikzpicture}

%% file: tikz/basline/dis-1.tex
\pgfplotsset{compat=1.3}
\begin{tikzpicture}

\begin{groupplot}[group style={group size=2 by 1}]
\nextgroupplot[
tick pos=both,
axis lines=left, xtick=\empty, 
x grid style={white!69.0196078431373!black},
yticklabel style = {xshift=-1 ex},
xticklabel style = {yshift=-1ex},
xmin=-0.2, xmax=4.2,
xtick style={color=black},
xtick={0,1,2,3,4},
xticklabel style={rotate=45.0,anchor=east},
xticklabels={\Netscience,\NipsEgo,\PagesGovernment,\Anybeat,\Gplus},
y grid style={white!69.0196078431373!black},
ymin=-0.060095260946855, ymax=0.0185368843957079,
ytick style={color=black},
ytick={-0.07,-0.06,-0.05,-0.04,-0.03,-0.02,-0.01,0,0.01,0.02},
ylabel={rel. increase \disidx},
legend style={font = \Huge, draw=none},
legend style={at={(1,1), anchor=north east}},
line width = 0.50 mm, 
tick style={line width=0.50mm},
ytick scale label code/.code={$\times 10^{-2}$},
]
\path [draw=black]
(axis cs:-0.1,-0.000629016685996)
--(axis cs:-0.1,0.000874192890588);

\path [draw=black]
(axis cs:0.9,-0.00066713752601)
--(axis cs:0.9,-3.4461455246e-05);

\path [draw=black]
(axis cs:1.9,-0.000546548523255)
--(axis cs:1.9,-3.8093497635e-05);

\path [draw=black]
(axis cs:2.9,-0.000600460884754)
--(axis cs:2.9,-0.000142616954252);

\path [draw=black]
(axis cs:3.9,-0.000468144573322369)
--(axis cs:3.9,-0.000287840197355631);

\path [draw=red]
(axis cs:0,-0.013201040705656)
--(axis cs:0,0.014962695971046);

\path [draw=red]
(axis cs:1,-0.053161738761712)
--(axis cs:1,-0.04951643591064);

\path [draw=red]
(axis cs:2,0.001772071773992)
--(axis cs:2,0.009185941851456);

\path [draw=red]
(axis cs:3,-0.029806988035078)
--(axis cs:3,-0.025602771342954);

\path [draw=red]
(axis cs:4,-0.048818782618228)
--(axis cs:4,-0.047575990206984);

\path [draw=blue]
(axis cs:0.1,-0.01259445844042)
--(axis cs:0.1,0.006526653379248);

\path [draw=blue]
(axis cs:1.1,-0.056521072522193)
--(axis cs:1.1,-0.048359265956923);

\path [draw=blue]
(axis cs:2.1,-0.005747946475147)
--(axis cs:2.1,-0.000269971909429);

\path [draw=blue]
(axis cs:3.1,-0.032254707621049)
--(axis cs:3.1,-0.028175865272457);

\path [draw=blue]
(axis cs:4.1,-0.049328980002168)
--(axis cs:4.1,-0.046700998247052);

\addplot [black, mark=asterisk, mark size=3, mark options={solid}, only marks]
table {%
-0.1 0.000122588102296
0.9 -0.000350799490628
1.9 -0.000292321010445
2.9 -0.000371538919503
3.9 -0.000377992385339
};
\addplot [red, mark=asterisk, mark size=3, mark options={solid}, only marks]
table {%
0 0.000880827632695
1 -0.051339087336176
2 0.005479006812724
3 -0.027704879689016
4 -0.048197386412606
};
\addplot [blue, mark=asterisk, mark size=3, mark options={solid}, only marks]
table {%
0.1 -0.003033902530586
1.1 -0.052440169239558
2.1 -0.003008959192288
3.1 -0.030215286446753
4.1 -0.04801498912461
};
\end{groupplot}

\end{tikzpicture}

%% file: tikz/basline/dis-2.tex
\pgfplotsset{compat=1.3}
\begin{tikzpicture}

\begin{groupplot}[group style={group size=2 by 1}]
\nextgroupplot[
tick pos=both,
axis lines=left, xtick=\empty, 
x grid style={white!69.0196078431373!black},
yticklabel style = {xshift=-1 ex},
xticklabel style = {yshift=-1ex},
xmin=-0.2, xmax=4.2,
xtick style={color=black},
xtick={0,1,2,3,4},
xticklabel style={rotate=45.0,anchor=east},
xticklabels={\Netscience,\NipsEgo,\PagesGovernment,\Anybeat,\Gplus},
y grid style={white!69.0196078431373!black},
ymin=0.0126591508247966, ymax=0.628825108358854,
ytick style={color=black},
ytick={0,0.1,0.2,0.3,0.4,0.5,0.6,0.7},
yticklabels={
  \(\displaystyle {0.0}\),
  \(\displaystyle {0.1}\),
  \(\displaystyle {0.2}\),
  \(\displaystyle {0.3}\),
  \(\displaystyle {0.4}\),
  \(\displaystyle {0.5}\),
  \(\displaystyle {0.6}\),
  \(\displaystyle {0.7}\)
},
ylabel={rel. increase \disidx},
legend style={draw=none},
legend style={at={(1.1,1), anchor=north east}},
line width = 0.50 mm, 
tick style={line width=0.50mm},
ytick scale label code/.code={$\times 10^{-2}$},
]

\path [draw=black]
(axis cs:-0.1,0.047622555008085)
--(axis cs:-0.1,0.052406552216763);

\path [draw=black]
(axis cs:0.9,0.040666694349072)
--(axis cs:0.9,0.041332203261736);

\path [draw=black]
(axis cs:1.9,0.046770741290215)
--(axis cs:1.9,0.047782168701931);

\path [draw=black]
(axis cs:2.9,0.0437442274915)
--(axis cs:2.9,0.044146197165868);

\path [draw=black]
(axis cs:3.9,0.050406886150427)
--(axis cs:3.9,0.050630570108613);

\path [draw=red]
(axis cs:0,0.183323978402376)
--(axis cs:0,0.202294724490924);

\path [draw=red]
(axis cs:1,0.58035204456031)
--(axis cs:1,0.599587247145578);

\path [draw=red]
(axis cs:2,0.154477056495996)
--(axis cs:2,0.158617057564446);

\path [draw=red]
(axis cs:3,0.452513071831155)
--(axis cs:3,0.460981497127829);

\path [draw=red]
(axis cs:4,0.569375941262302)
--(axis cs:4,0.57436989666069);

\path [draw=blue]
(axis cs:0.1,0.222331212763594)
--(axis cs:0.1,0.238569458515456);

\path [draw=blue]
(axis cs:1.1,0.590663422958611)
--(axis cs:1.1,0.600817564834579);

\path [draw=blue]
(axis cs:2.1,0.21482208571373)
--(axis cs:2.1,0.21882258172096);

\path [draw=blue]
(axis cs:3.1,0.47044561692322)
--(axis cs:3.1,0.47604180175384);

\path [draw=blue]
(axis cs:4.1,0.566887331046118)
--(axis cs:4.1,0.575600923924598);

\addplot [black, mark=asterisk, mark size=3, mark options={solid}, only marks]
table {%
-0.1 0.0500145536124244
0.9 0.040999448805404
1.9 0.047276454996073
2.9 0.043945212328684
3.9 0.05051872812952
};
\addplot [red, mark=asterisk, mark size=3, mark options={solid}, only marks]
table {%
0 0.19280935144665
1 0.589969645852944
2 0.156547057030221
3 0.456747284479492
4 0.571872918961496
};
\addplot [blue, mark=asterisk, mark size=3, mark options={solid}, only marks]
table {%
0.1 0.230450335639525
1.1 0.595740493896595
2.1 0.216822333717345
3.1 0.47324370933853
4.1 0.571244127485358
};
\addlegendentry{\Random}
\addlegendentry{\HighDegree}
\addlegendentry{\MaxInfluence}
\end{groupplot}

\end{tikzpicture}

%% file: tikz/basline/int-1.tex
\pgfplotsset{compat=1.3}
\begin{tikzpicture}

\begin{groupplot}[group style={group size=2 by 1}]
\nextgroupplot[
tick pos=both,
axis lines=left, xtick=\empty, 
x grid style={white!69.0196078431373!black},
yticklabel style = {xshift=-1 ex},
xticklabel style = {yshift=-1ex},
xmin=-0.2, xmax=4.2,
xtick style={color=black},
xtick={0,1,2,3,4},
xticklabel style={rotate=45.0,anchor=east},
xticklabels={\Netscience,\NipsEgo,\PagesGovernment,\Anybeat,\Gplus},
y grid style={white!69.0196078431373!black},
ymin=-0.0597880123921739, ymax=0.0071491992231419,
ytick style={color=black},
ytick={-0.06,-0.05,-0.04,-0.03,-0.02,-0.01,0,0.01},
ylabel={rel. increase \intidx},
legend style={font = \Huge, draw=none},
legend style={at={(1,1), anchor=north east}},
line width = 0.50 mm, 
tick style={line width=0.50mm},
ytick scale label code/.code={$\times 10^{-2}$},
]

\path [draw=black]
(axis cs:-0.1,-0.001090484360664)
--(axis cs:-0.1,0.000680881649102);

\path [draw=black]
(axis cs:0.9,-0.000655213622717)
--(axis cs:0.9,-9.7492561923e-05);

\path [draw=black]
(axis cs:1.9,-0.000581512252346)
--(axis cs:1.9,-0.000188492936932);

\path [draw=black]
(axis cs:2.9,-0.000510104608403)
--(axis cs:2.9,-0.000230447748653);

\path [draw=black]
(axis cs:3.9,-0.00049651022985)
--(axis cs:3.9,-0.000281396594136);

\path [draw=red]
(axis cs:0,-0.017205497706377)
--(axis cs:0,0.004106598695173);

\path [draw=red]
(axis cs:1,-0.052804989304264)
--(axis cs:1,-0.048845418821184);

\path [draw=red]
(axis cs:2,-0.007344950292105)
--(axis cs:2,-0.004442105231719);

\path [draw=red]
(axis cs:3,-0.027023598453499)
--(axis cs:3,-0.024974636625455);

\path [draw=red]
(axis cs:4,-0.045799022100929)
--(axis cs:4,-0.044107901219903);

\path [draw=blue]
(axis cs:0.1,-0.013210234808446)
--(axis cs:0.1,0.002501466302444);

\path [draw=blue]
(axis cs:1.1,-0.056745411864205)
--(axis cs:1.1,-0.048493913653099);

\path [draw=blue]
(axis cs:2.1,-0.007619073257552)
--(axis cs:2.1,-0.004783954813946);

\path [draw=blue]
(axis cs:3.1,-0.02656739096934)
--(axis cs:3.1,-0.024335262751822);

\path [draw=blue]
(axis cs:4.1,-0.046254425772522)
--(axis cs:4.1,-0.043265841983194);

\addplot [black, mark=asterisk, mark size=3, mark options={solid}, only marks]
table {%
-0.1 -0.0002048013557811
0.9 -0.00037635309232
1.9 -0.000385002594639
2.9 -0.000370276178528
3.9 -0.000388953411993
};
\addplot [red, mark=asterisk, mark size=3, mark options={solid}, only marks]
table {%
0 -0.006549449505602
1 -0.050825204062724
2 -0.005893527761912
3 -0.025999117539477
4 -0.044953461660416
};
\addplot [blue, mark=asterisk, mark size=3, mark options={solid}, only marks]
table {%
0.1 -0.005354384253001
1.1 -0.052619662758652
2.1 -0.006201514035749
3.1 -0.025451326860581
4.1 -0.044760133877858
};
\end{groupplot}

\end{tikzpicture}

%% file: tikz/basline/int-2.tex
\pgfplotsset{compat=1.3}
\begin{tikzpicture}

\begin{groupplot}[group style={group size=2 by 1}]
\nextgroupplot[
tick pos=both,
axis lines=left, xtick=\empty, 
x grid style={white!69.0196078431373!black},
yticklabel style = {xshift=-1 ex},
xticklabel style = {yshift=-1ex},
xmin=-0.2, xmax=4.2,
xtick style={color=black},
xtick={0,1,2,3,4},
xticklabel style={rotate=45.0,anchor=east},
xticklabels={\Netscience,\NipsEgo,\PagesGovernment,\Anybeat,\Gplus},
y grid style={white!69.0196078431373!black},
ymin=0.0127386235742769, ymax=0.629363755388465,
ytick style={color=black},
ytick={0,0.1,0.2,0.3,0.4,0.5,0.6,0.7},
yticklabels={
  \(\displaystyle {0.0}\),
  \(\displaystyle {0.1}\),
  \(\displaystyle {0.2}\),
  \(\displaystyle {0.3}\),
  \(\displaystyle {0.4}\),
  \(\displaystyle {0.5}\),
  \(\displaystyle {0.6}\),
  \(\displaystyle {0.7}\)
},
ylabel={rel. increase \intidx},
legend style={draw=none},
legend style={at={(axis cs:2.4,0.7)}, anchor=north west},
line width = 0.50 mm, 
tick style={line width=0.50mm},
ytick scale label code/.code={$\times 10^{-2}$},
]
\path [draw=black]
(axis cs:-0.1,0.046074030509379)
--(axis cs:-0.1,0.051232188321367);

\path [draw=black]
(axis cs:0.9,0.04076703865674)
--(axis cs:0.9,0.041253556213468);

\path [draw=black]
(axis cs:1.9,0.047578169853077)
--(axis cs:1.9,0.048346335224339);

\path [draw=black]
(axis cs:2.9,0.045309241044207)
--(axis cs:2.9,0.045663520716761);

\path [draw=black]
(axis cs:3.9,0.049467024251114)
--(axis cs:3.9,0.04971866185877);

\path [draw=red]
(axis cs:0,0.191839754682586)
--(axis cs:0,0.220375975900772);

\path [draw=red]
(axis cs:1,0.574838799358796)
--(axis cs:1,0.596274720426214);

\path [draw=red]
(axis cs:2,0.198741616145655)
--(axis cs:2,0.201783074784027);

\path [draw=red]
(axis cs:3,0.415101483309098)
--(axis cs:3,0.4222217844019);

\path [draw=red]
(axis cs:4,0.549666617968446)
--(axis cs:4,0.556299935074844);

\path [draw=blue]
(axis cs:0.1,0.234790523368836)
--(axis cs:0.1,0.253356711951856);

\path [draw=blue]
(axis cs:1.1,0.591384058407926)
--(axis cs:1.1,0.601335340306002);

\path [draw=blue]
(axis cs:2.1,0.228673238709404)
--(axis cs:2.1,0.232769440490988);

\path [draw=blue]
(axis cs:3.1,0.418382219477667)
--(axis cs:3.1,0.425067884750365);

\path [draw=blue]
(axis cs:4.1,0.546742832579124)
--(axis cs:4.1,0.554878069026602);

\addplot [black, mark=asterisk, mark size=3, mark options={solid}, only marks]
table {%
-0.1 0.0486531094153733
0.9 0.041010297435104
1.9 0.047962252538708
2.9 0.045486380880484
3.9 0.049592843054942
};
\addplot [red, mark=asterisk, mark size=3, mark options={solid}, only marks]
table {%
0 0.206107865291679
1 0.585556759892505
2 0.200262345464841
3 0.418661633855499
4 0.552983276521645
};
\addplot [blue, mark=asterisk, mark size=3, mark options={solid}, only marks]
table {%
0.1 0.244073617660346
1.1 0.596359699356964
2.1 0.230721339600196
3.1 0.421725052114016
4.1 0.550810450802863
};
\addlegendentry{\Random}
\addlegendentry{\HighDegree}
\addlegendentry{\MaxInfluence}
\end{groupplot}

\end{tikzpicture}

%% file: tikz/greedy/time.tex
\begin{tikzpicture}

\definecolor{darkgray176}{RGB}{176,176,176}
\definecolor{green01270}{RGB}{0,127,0}

\begin{axis}[
log basis y={10},
tick pos=both,
axis lines=left, xtick=\empty, ytick=\empty,
x grid style={white!69.0196078431373!black},
xlabel={$n$},
xmin=185.5, xmax=922.5,
xtick style={color=black},
xtick={100,300,500,700,900,1100},
xticklabels={
  \(\displaystyle {100}\),
  \(\displaystyle {300}\),
  \(\displaystyle {500}\),
  \(\displaystyle {700}\),
  \(\displaystyle {900}\),
  \(\displaystyle {1100}\)
},
y grid style={darkgray176},
ylabel={seconds},
ymin=0.0389282771166603, ymax=10066.7456896512,
ymode=log,
ytick style={color=black},
ytick={0.001,0.01,0.1,1,10,100,1000,10000,100000,1000000},
legend pos= outer north east,
legend columns=1,
legend style={draw=none},
line width = 0.50 mm, 
tick style={line width=0.50mm},
ylabel={seconds}
]
\addplot [green]
table {%
219 0.351797819137573
379 0.498874902725219
404 0.690531015396118
889 1.92104387283325
};
\addplot [red]
table {%
219 0.117319107055664
379 0.227393150329589
404 0.331827878952026
889 1.62721514701843
};
\addplot [blue]
table {%
219 0.164546012878417
379 0.336296081542968
404 0.448657035827636
889 1.43192887306213
};
\addplot [black]
table {%
219 0.141379117965698
379 0.287919044494628
404 0.174530029296875
889 1.02280020713806
};
\addplot [red, dashed]
table {%
219 23.1959509849548
379 96.7579200267792
404 97.9669659137726
889 1780.42035198212
};
\addplot [blue, dotted]
table {%
219 8.38209581375122
379 31.5538189411163
404 34.2430381774902
889 1825.1451280117
};
\addplot [green, dashed]
table {%
219 11.2410268783569
379 54.1678268909454
404 84.2020509243011
889 4585.47620415688
};
\addplot [black, dashed]
table {%
219 24.2821140289307
379 97.0339431762695
404 259.328497886658
889 8610.16955184937
};

\addlegendentry{\MaxLinInt}
\addlegendentry{\MaxLinDis}
\addlegendentry{\MaxLinPol}
\addlegendentry{\MaxLinDisCon}
\addlegendentry{\MaxInt}
\addlegendentry{\MaxDis}
\addlegendentry{\MaxPol}
\addlegendentry{\MaxDisCon}

\end{axis}

\end{tikzpicture}

%% file: tikz/greedy/d.tex
\pgfplotsset{compat=1.3}
\begin{tikzpicture}

\begin{axis}[
    legend cell align={right},
    legend style={fill opacity=0.8, draw opacity=1, text opacity=1, draw=none},
    tick pos=both,
    axis lines=left, xtick=\empty, 
    x grid style={white!69.0196078431373!black}, 
    xmin=-0.5495, xmax=3.2895,
    xtick style={color=black},
    xtick={0,1,2,3},
    xticklabel style={rotate=20.0,anchor=east},
    xticklabels={Convote,WikiTalkHT,Netscience,WikiVote},
    y grid style={darkgray176},
    ymin=0, ymax=0.240018069198431,
    y grid style={white!69.0196078431373!black},
    ylabel={rel. increase dis--con},
ytick style={color=black},
legend pos= north west,
legend style={draw=none},
line width = 0.50 mm, 
tick style={line width=0.50mm},
    ]
\draw[draw=none,fill=black,fill opacity=0.7] (axis cs:-0.375,0) rectangle (axis cs:-0.225,0.0813089467594634);
\draw[draw=none,fill=black,fill opacity=0.7] (axis cs:0.625,0) rectangle (axis cs:0.775,0.0673478485883574);
\draw[draw=none,fill=black,fill opacity=0.7] (axis cs:1.625,0) rectangle (axis cs:1.775,0.228588637331248);
\draw[draw=none,fill=black,fill opacity=0.7] (axis cs:2.625,0) rectangle (axis cs:2.775,0.055083670469406);
\draw[draw=none,fill=red,fill opacity=0.7] (axis cs:-0.205,0) rectangle (axis cs:-0.055,0.0813089467594616);
\draw[draw=none,fill=red,fill opacity=0.7] (axis cs:0.795,0) rectangle (axis cs:0.945,0.0673478485883579);
\draw[draw=none,fill=red,fill opacity=0.7] (axis cs:1.795,0) rectangle (axis cs:1.945,0.228588637331839);
\draw[draw=none,fill=red,fill opacity=0.7] (axis cs:2.795,0) rectangle (axis cs:2.945,0.0550836704694053);
\draw[draw=none,fill=blue,fill opacity=0.7] (axis cs:-0.035,0) rectangle (axis cs:0.115,0.0788504647853582);
\draw[draw=none,fill=blue,fill opacity=0.7] (axis cs:0.965,0) rectangle (axis cs:1.115,0.0661599570396653);
\draw[draw=none,fill=blue,fill opacity=0.7] (axis cs:1.965,0) rectangle (axis cs:2.115,0.211866551919509);
\draw[draw=none,fill=blue,fill opacity=0.7] (axis cs:2.965,0) rectangle (axis cs:3.115,0.0542642510595786);
\addlegendimage{line width=0.3mm,color=black};
\addlegendentry{\MaxLinDisCon};
\addlegendimage{line width=0.3mm,color=red};
\addlegendentry{\MaxDisCon};
\addlegendimage{line width=0.3mm,color=blue}
\addlegendentry{\LowDisCon};
\end{axis}
\end{tikzpicture}

%% file: tikz/ttt.tex
\begin{tikzpicture}

\definecolor{color0}{rgb}{0.0470588235294118,0.364705882352941,0.647058823529412}
\definecolor{color1}{rgb}{0,0.725490196078431,0.270588235294118}
\definecolor{color2}{rgb}{1,0.584313725490196,0}
\definecolor{color3}{rgb}{1,0.172549019607843,0}
\definecolor{color4}{rgb}{0.517647058823529,0.356862745098039,0.592156862745098}

\begin{axis}[
log basis y={10},
tick pos=both,
axis lines=left, xtick=\empty, ytick=\empty,
x grid style={white!69.0196078431373!black},
xmode=log,
xmin=287.974598322573, xmax=121234.106075193,
xtick style={color=black},
xtick={10,100,1000,10000,100000,1000000,10000000},
xlabel={$n$},
y grid style={white!69.0196078431373!black},
ymin=0.00270347050270336, ymax=26146.2301867163,
ymode=log,
ytick style={color=black},
ymin=0.0491058247144428, ymax=4717.94954243572,
legend style={font=\Large, draw=none},
legend pos= outer north east, 
legend columns=1,
line width = 0.50 mm, 
tick style={line width=0.50mm},
ylabel={seconds},
ytick={0.001,0.01,0.1,1,10,100,1000,10000,100000},
]
\addplot [black]
table {%
379 0.764966964721679
889 1.16822504997253
962 2.09355688095093
1133 1.3598461151123
2000 3.24952888488769
2672 15.2518429756165
2888 1.04771494865417
4158 12.8684680461884
4991 6.07284808158875
5054 12.2550020217896
7057 45.4769439697266
11204 147.640490055084
12645 38.9246850013733
14113 105.535808086395
17903 345.76726102829
21363 276.831474065781
23613 88.4117879867554
56739 476.900470018387
92117 1352.07726287842
};
\addplot [red, dashed]
table {%
379 0.300781965255737
889 4.1201479434967
962 17.492301940918
1133 6.64484000205994
2000 32.1900320053101
2672 115.749627828598
2888 21.6390209197998
4158 95.3211090564728
4991 117.369243860245
5054 199.107218027115
7057 489.44070315361
11204 1028.44246792793
12645 750.429357051849
14113 915.232490062714
17903 2702.09905195236
21363 2182.64607310295
23613 2445.63373088837
56739 652.830957174301
92117 1789.24687790871
};
\addplot [blue, dashed]
table {%
379 0.08272099494934
889 1.06456112861633
962 3.11181807518005
1133 1.66212701797485
2000 13.1593551635742
2672 67.3092830181122
2888 15.0609781742096
4158 73.0469350814819
4991 97.3521199226379
5054 199.926979064941
7057 565.310620069504
11204 991.69727897644
12645 760.318145036697
14113 880.545186042786
17903 2506.07513093948
21363 2190.7809381485
23613 2009.31821107864
56739 689.032766103745
92117 1932.56712293625
};
\addplot [green, dashed]
table {%
379 0.150960206985473
889 1.82266402244568
962 4.09466600418091
1133 4.20449185371399
2000 27.2559521198273
2672 102.760493040085
2888 18.1482269763947
4158 89.5487418174744
4991 87.510880947113
5054 184.709882974625
7057 591.682937860489
11204 980.06405210495
12645 744.870545864105
14113 876.566281080246
17903 2488.39201998711
21363 2176.75665092468
23613 2800.72553992271
56739 680.772832870483
92117 1969.11894917488
};
\addplot [red]
table {%
379 0.177518844604492
889 0.576751947402954
962 2.07290506362915
1133 1.50677299499512
2000 3.97834992408752
2672 18.1553838253021
2888 0.86129093170166
4158 19.3920450210571
4991 9.81180000305176
5054 15.1049220561981
7057 74.8052289485931
11204 215.945448160172
12645 66.7483649253845
14113 171.059938907623
17903 524.689456224442
21363 443.971834897995
23613 139.407995939255
56739 575.343125104904
92117 1477.30988001823
};
\addplot [blue]
table {%
379 0.190727949142456
889 0.567249059677124
962 2.13277697563171
1133 1.13939118385315
2000 3.90835094451904
2672 18.2101020812988
2888 0.772948026657104
4158 17.3098080158234
4991 9.41839289665222
5054 14.9316539764404
7057 69.4169249534607
11204 215.08848285675
12645 63.4836950302124
14113 167.158182144165
17903 497.33438205719
21363 422.177358865738
23613 135.526823997498
56739 573.797589063644
92117 1474.43640303612
};
\addlegendentry{\MaxSum}
\addlegendentry{\MaxLinInt}
\addlegendentry{\MaxLinDis}
\addlegendentry{\MaxLinPol}
\addlegendentry{\MaxLinDisCon}
\addlegendentry{\MaxInfluence}
\end{axis}

\end{tikzpicture}

%% file: tikz/greedy/a.tex
\pgfplotsset{compat=1.3}
\begin{tikzpicture}

\begin{axis}[
legend cell align={left},
legend style={fill opacity=0.8, draw opacity=1, text opacity=1, draw=none},
tick pos=both,
axis lines=left, xtick=\empty, 
x grid style={white!69.0196078431373!black},
xmin=-0.5495, xmax=3.2895,
xtick style={color=black},
xtick={0,1,2,3},
xticklabel style={rotate=20.0,anchor=east},
xticklabels={Convote,WikiTalkHT,Netscience,WikiVote},
y grid style={white!69.0196078431373!black},
ymin=0, ymax=0.0239027820464249,
ytick style={color=black},
ylabel={rel. increase internal conflict},
ytick style={color=black},
legend pos= outer north east,
legend style={draw=none},
line width = 0.50 mm, 
tick style={line width=0.50mm},
ytick scale label code/.code={$\times 10^{-2}$},
]
\draw[draw=none,fill=black,fill opacity=0.7] (axis cs:-0.375,0) rectangle (axis cs:-0.225,0.0213633247241984);
\draw[draw=none,fill=black,fill opacity=0.7] (axis cs:0.625,0) rectangle (axis cs:0.775,0.0132705057006914);
\draw[draw=none,fill=black,fill opacity=0.7] (axis cs:1.625,0) rectangle (axis cs:1.775,0.0114420192393931);
\draw[draw=none,fill=black,fill opacity=0.7] (axis cs:2.625,0) rectangle (axis cs:2.775,0.00572828070638435);
\draw[draw=none,fill=red,fill opacity=0.7] (axis cs:-0.205,0) rectangle (axis cs:-0.055,0.0227645543299285);
\draw[draw=none,fill=red,fill opacity=0.7] (axis cs:0.795,0) rectangle (axis cs:0.945,0.0132705057006914);
\draw[draw=none,fill=red,fill opacity=0.7] (axis cs:1.795,0) rectangle (axis cs:1.945,0.0121377799474679);
\draw[draw=none,fill=red,fill opacity=0.7] (axis cs:2.795,0) rectangle (axis cs:2.945,0.00649534895404421);
\draw[draw=none,fill=blue,fill opacity=0.7] (axis cs:-0.035,0) rectangle (axis cs:0.115,0.0194819477392247);
\draw[draw=none,fill=blue,fill opacity=0.7] (axis cs:0.965,0) rectangle (axis cs:1.115,0.0119573926231118);
\draw[draw=none,fill=blue,fill opacity=0.7] (axis cs:1.965,0) rectangle (axis cs:2.115,0.0115944448087757);
\draw[draw=none,fill=blue,fill opacity=0.7] (axis cs:2.965,0) rectangle (axis cs:3.115,0.00594627857391189);
\addlegendimage{line width=0.3mm,color=black};
\addlegendentry{\MaxLinInt};
\addlegendimage{line width=0.3mm,color=red};
\addlegendentry{\MaxInt};
\addlegendimage{line width=0.3mm,color=blue}
\addlegendentry{\LowInt};
\end{axis}

\end{tikzpicture}

%% file: tikz/greedy/b.tex
\pgfplotsset{compat=1.3}
\begin{tikzpicture}

\begin{axis}[
legend cell align={left},
legend style={fill opacity=0.8, draw opacity=1, text opacity=1, draw=none},
tick pos=both,
axis lines=left, xtick=\empty, 
x grid style={white!69.0196078431373!black},
xmin=-0.5495, xmax=3.2895,
xtick style={color=black},
xtick={0,1,2,3},
xticklabel style={rotate=20.0,anchor=east},
xticklabels={Convote,WikiTalkHT,Netscience,WikiVote},
y grid style={white!69.0196078431373!black},
ymin=0, ymax=0.0436403065496831,
ytick style={color=black},
ylabel={rel. increase disagreement},
ytick style={color=black},
legend pos= north east,
legend style={draw=none},
line width = 0.50 mm, 
tick style={line width=0.50mm},
ytick scale label code/.code={$\times 10^{-2}$},
]
\draw[draw=none,fill=black,fill opacity=0.7] (axis cs:-0.375,0) rectangle (axis cs:-0.225,0.039543712755692);
\draw[draw=none,fill=black,fill opacity=0.7] (axis cs:0.625,0) rectangle (axis cs:0.775,0.0215579563220537);
\draw[draw=none,fill=black,fill opacity=0.7] (axis cs:1.625,0) rectangle (axis cs:1.775,0.0179380903354767);
\draw[draw=none,fill=black,fill opacity=0.7] (axis cs:2.625,0) rectangle (axis cs:2.775,0.0112781900240184);
\draw[draw=none,fill=red,fill opacity=0.7] (axis cs:-0.205,0) rectangle (axis cs:-0.055,0.0415621967139839);
\draw[draw=none,fill=red,fill opacity=0.7] (axis cs:0.795,0) rectangle (axis cs:0.945,0.0215579563220537);
\draw[draw=none,fill=red,fill opacity=0.7] (axis cs:1.795,0) rectangle (axis cs:1.945,0.0189136145487223);
\draw[draw=none,fill=red,fill opacity=0.7] (axis cs:2.795,0) rectangle (axis cs:2.945,0.0114047202727589);
\draw[draw=none,fill=blue,fill opacity=0.7] (axis cs:-0.035,0) rectangle (axis cs:0.115,0.0253975011585583);
\draw[draw=none,fill=blue,fill opacity=0.7] (axis cs:0.965,0) rectangle (axis cs:1.115,0.0187993913663749);
\draw[draw=none,fill=blue,fill opacity=0.7] (axis cs:1.965,0) rectangle (axis cs:2.115,0.0164674043978995);
\draw[draw=none,fill=blue,fill opacity=0.7] (axis cs:2.965,0) rectangle (axis cs:3.115,0.0105221842300732);
\addlegendimage{line width=0.3mm,color=black};
\addlegendentry{\MaxLinDis};
\addlegendimage{line width=0.3mm,color=red};
\addlegendentry{\MaxDis};
\addlegendimage{line width=0.3mm,color=blue}
\addlegendentry{\LowDis};
\end{axis}

\end{tikzpicture}

%% file: tikz/greedy/c.tex
\pgfplotsset{compat=1.3}
\begin{tikzpicture}

\begin{axis}[
legend cell align={left},
legend style={fill opacity=0.8, draw opacity=1, text opacity=1, draw=none},
tick pos=both,
axis lines=left, xtick=\empty, 
x grid style={white!69.0196078431373!black},
xmin=-0.5495, xmax=3.2895,
xtick style={color=black},
xtick={0,1,2,3},
xticklabel style={rotate=20.0,anchor=east},
xticklabels={Convote,WikiTalkHT,Netscience,WikiVote},
y grid style={white!69.0196078431373!black},
ymin=0, ymax=0.116277122770421,
ytick style={color=black},
ylabel={rel. increase polarization},
ytick style={color=black},
legend style={draw=none},
line width = 0.50 mm, 
tick style={line width=0.50mm},
]
\draw[draw=none,fill=black,fill opacity=0.7] (axis cs:-0.375,0) rectangle (axis cs:-0.225,0.0910891753009288);
\draw[draw=none,fill=black,fill opacity=0.7] (axis cs:0.625,0) rectangle (axis cs:0.775,0.101090703760452);
\draw[draw=none,fill=black,fill opacity=0.7] (axis cs:1.625,0) rectangle (axis cs:1.775,0.0525743173524555);
\draw[draw=none,fill=black,fill opacity=0.7] (axis cs:2.625,0) rectangle (axis cs:2.775,0.0274842017235349);
\draw[draw=none,fill=red,fill opacity=0.7] (axis cs:-0.205,0) rectangle (axis cs:-0.055,0.0943371109464605);
\draw[draw=none,fill=red,fill opacity=0.7] (axis cs:0.795,0) rectangle (axis cs:0.945,0.11074011692421);
\draw[draw=none,fill=red,fill opacity=0.7] (axis cs:1.795,0) rectangle (axis cs:1.945,0.0562017282423905);
\draw[draw=none,fill=red,fill opacity=0.7] (axis cs:2.795,0) rectangle (axis cs:2.945,0.0305632286026342);
\draw[draw=none,fill=blue,fill opacity=0.7] (axis cs:-0.035,0) rectangle (axis cs:0.115,0.064040587880875);
\draw[draw=none,fill=blue,fill opacity=0.7] (axis cs:0.965,0) rectangle (axis cs:1.115,0.0895149503460742);
\draw[draw=none,fill=blue,fill opacity=0.7] (axis cs:1.965,0) rectangle (axis cs:2.115,0.0421250338391266);
\draw[draw=none,fill=blue,fill opacity=0.7] (axis cs:2.965,0) rectangle (axis cs:3.115,0.0215947159442649);
\addlegendimage{line width=0.3mm,color=black};
\addlegendentry{\MaxLinPol};
\addlegendimage{line width=0.3mm,color=red};
\addlegendentry{\MaxPol};
\addlegendimage{line width=0.3mm,color=blue}
\addlegendentry{\LowPol};
\end{axis}

\end{tikzpicture}

%% file: tables/marketing-2.tex
\begin{table*}[!htb]
\caption{Results for marketing campaigns with
	$k = \lceil 0.5\%\cdot n\rceil$ seeds,
	where we initialized the innate opinions using the uniform distribution.
	We report the relative increase of each index in percent.}
\label{tab:marketing-2}
\centering
  \vspace{-4mm}
\begin{adjustbox}{max width=\textwidth}
\begin{tabular}{c ccccccccc ccccccccc}
\toprule
\textbf{Dataset}  & \multicolumn{9}{c}{\textbf{\disidx}} & \multicolumn{9}{c}{\textbf{\intidx}} \\
\cmidrule(lr){2-10} 
\cmidrule(lr){11-19} 
 & \MaxSum & \MaxLinDisCon & \MaxLinPol & \MaxLinDis & \MaxLinInt &  \MaxInfluence & \Random & \FJGreedy & \FJUpp & \MaxSum & \MaxLinDisCon & \MaxLinPol & \MaxLinDis & \MaxLinInt &  \MaxInfluence & \Random & \FJGreedy &\FJUpp\\
\midrule
\Netscience  & 0.81 & 0.79 & 1.21 & \textbf{1.28} & 0.44 & 0.79 & -0.12 & 1.04 &  15.84   & -0.53 & -0.54 & 0.21 & -0.13 & 0.86 & -0.53 & -0.1 & \textbf{1.1} &  16.98   \\
\WikiVote  & -1.33 & -1.32 & 1.57 & \textbf{1.77} & 0.77 & -1.28 & -0.07 & 1.12 &  15.74   & -1.47 & -1.49 & 0.35 & 0.67 & 1.0 & -1.42 & -0.08 & \textbf{1.17} &  17.84   \\
\Reed  & -0.05 & -0.09 & 2.71 & 2.59 & 0.45 & -0.08 & -0.08 & \textbf{2.74} &  44.57   & -0.39 & -0.42 & -0.0 & 0.09 & 0.69 & -0.4 & -0.08 & \textbf{0.77} &  11.49   \\
\EmailUniv  & -0.19 & -0.13 & 1.73 & \textbf{1.8} & 0.78 & -0.11 & -0.09 & 1.32 &  19.14   & -0.75 & -0.63 & 0.26 & 0.44 & \textbf{1.11} & -0.59 & -0.08 & 1.0 &  14.67   \\
\Hamster  & -0.97 & -0.93 & 1.67 & \textbf{1.87} & 0.51 & -0.81 & -0.04 & 1.45 &  22.82   & -1.14 & -1.11 & 0.12 & 0.24 & 0.67 & -1.05 & -0.07 & \textbf{0.84} &  12.85   \\
\USFCA  & -0.33 & -0.32 & 2.8 & 2.85 & 0.5 & -0.35 & -0.07 & \textbf{3.2} &  50.80   & -0.62 & -0.57 & 0.1 & 0.21 & 0.66 & -0.62 & -0.09 & \textbf{0.8} &  11.81   \\
\NipsEgo  & -5.37 & -5.35 & 0.4 & 0.55 & 0.43 & -5.35 & -0.06 & \textbf{0.71} &  10.22   & -5.33 & -5.31 & 0.38 & 0.53 & 0.83 & -5.3 & -0.06 & \textbf{1.54} &  39.06   \\
\PagesGovernment  & -0.44 & -0.44 & 2.06 & \textbf{2.26} & 0.71 & -0.56 & -0.08 & 1.83 &  29.07   & -0.81 & -0.8 & 0.18 & 0.38 & 0.75 & -0.82 & -0.08 & \textbf{0.83} &  12.67   \\
\HepPh  & -0.49 & -0.32 & 1.37 & \textbf{1.89} & 0.54 & -0.4 & -0.07 & 1.24 &  18.85   & -0.64 & -0.59 & 0.01 & 0.31 & 0.82 & -0.67 & -0.08 & \textbf{0.98} &  14.81   \\
\Anybeat  & -2.34 & -2.32 & 1.24 & \textbf{1.41} & 0.59 & -2.35 & -0.08 & 0.85 &  12.63   & -2.43 & -2.41 & 0.3 & 0.57 & 1.16 & -2.43 & -0.07 & \textbf{1.37} &  22.25   \\
\CondMat  & -0.28 & -0.21 & 1.26 & \textbf{1.77} & 0.66 & -0.27 & -0.07 & 1.15 &  17.84   & -0.68 & -0.68 & 0.03 & 0.4 & 0.81 & -0.69 & -0.07 & \textbf{0.98} &  15.40   \\
\Gplus  & -4.94 & -4.94 & \textbf{0.9} & 0.72 & 0.36 & -4.94 & -0.07 & 0.72 &  10.73   & -4.71 & -4.72 & -0.04 & 0.46 & 1.05 & -4.71 & -0.07 & \textbf{1.53} &  29.31   \\
\Brightkite  & -0.88 & -0.84 & 1.14 & \textbf{1.49} & 0.54 & -0.91 & -0.07 & - &  -   & -1.52 & -1.52 & 0.07 & 0.47 & \textbf{0.9} & -1.53 & -0.07 & - &  -   \\
\WikiTalk  & -1.43 & -1.34 & 1.66 & \textbf{1.68} & 0.61 & -1.41 & -0.09 & - &  -   & -1.94 & -1.9 & 0.38 & 0.64 & \textbf{1.06} & -1.94 & -0.09 & - &  -   \\
\bottomrule
\end{tabular}
\end{adjustbox}
\end{table*}

%% file: tables/backfire-2.tex
\begin{table*}[!htb]
\caption{Results for polarizing campaigns with
	$k = \lceil 0.5\%\cdot n\rceil$ seeds,
	where we initialized the innate opinions using the uniform distribution.
	We report the relative increase of each index in percent.}
\label{tab:backfire-2}
\centering
  \vspace{-4mm}
\begin{adjustbox}{max width=\textwidth}
\begin{tabular}{c ccccccccc ccccccccc}
\toprule
\textbf{Dataset}  & \multicolumn{9}{c}{\textbf{\disidx}} & \multicolumn{9}{c}{\textbf{\intidx}} \\
\cmidrule(lr){2-10} 
\cmidrule(lr){11-19} 
 & \MaxSum & \MaxLinDisCon & \MaxLinPol & \MaxLinDis & \MaxLinInt &  \MaxInfluence & \Random & \FJGreedy & \FJUpp & \MaxSum & \MaxLinDisCon & \MaxLinPol & \MaxLinDis & \MaxLinInt &  \MaxInfluence & \Random & \FJGreedy &\FJUpp\\
\midrule
\Netscience  & 3.64 & 5.36 & 2.31 & \textbf{7.33} & \textbf{7.33} & 7.32 & 0.68 & 1.04 &  15.84   & 4.85 & 4.88 & 0.97 & 7.41 & 7.95 & \textbf{7.96} & 0.68 & 1.1 &  16.98   \\
\WikiVote  & 2.6 & 5.79 & 11.05 & \textbf{11.09} & 10.93 & \textbf{11.09} & 0.68 & 1.12 &  15.74   & 2.26 & 6.06 & 12.36 & 12.41 & 12.6 & \textbf{12.68} & 0.67 & 1.17 &  17.84   \\
\Reed  & 5.62 & 5.55 & 6.01 & \textbf{8.89} & 8.4 & 8.39 & 0.73 & 2.74 &  44.57   & 5.98 & 5.31 & 1.99 & 7.96 & \textbf{8.49} & 8.48 & 0.71 & 0.77 &  11.49   \\
\EmailUniv  & 2.69 & 4.63 & 7.7 & \textbf{8.2} & 8.15 & 8.11 & 0.79 & 1.32 &  19.14   & 2.29 & 4.43 & 8.9 & \textbf{9.94} & 9.88 & 9.84 & 0.78 & 1.0 &  14.67   \\
\Hamster  & 3.25 & 4.96 & 7.61 & \textbf{10.08} & 10.02 & 10.06 & 0.68 & 1.45 &  22.82   & 3.1 & 5.01 & 7.49 & 11.15 & 11.42 & \textbf{11.48} & 0.65 & 0.84 &  12.85   \\
\USFCA  & 2.74 & 4.02 & 5.58 & \textbf{7.17} & 6.83 & 6.87 & 0.79 & 3.2 &  50.80   & 1.83 & 2.51 & 3.92 & 8.16 & 8.67 & \textbf{8.7} & 0.77 & 0.8 &  11.81   \\
\NipsEgo  & 37.94 & 59.44 & 59.44 & 59.44 & \textbf{59.45} & 59.44 & 0.46 & 0.71 &  10.22   & 37.32 & 58.95 & 58.93 & 58.95 & \textbf{58.96} & 58.94 & 0.46 & 1.54 &  39.06   \\
\PagesGovernment  & 3.41 & 4.61 & 7.52 & \textbf{9.28} & 7.83 & 8.48 & 0.69 & 1.83 &  29.07   & 2.48 & 3.86 & 6.05 & 9.05 & \textbf{10.09} & 9.99 & 0.7 & 0.83 &  12.67   \\
\HepPh  & 2.36 & 3.49 & 4.82 & \textbf{6.46} & 5.53 & 5.97 & 0.68 & 1.24 &  18.85   & 1.58 & 2.58 & 4.14 & 6.79 & \textbf{7.51} & 7.36 & 0.7 & 0.98 &  14.81   \\
\Anybeat  & 27.84 & 36.77 & 38.55 & \textbf{38.58} & 38.19 & 38.5 & 0.5 & 0.85 &  12.63   & 22.18 & 30.28 & 31.73 & 32.13 & \textbf{32.56} & 32.35 & 0.55 & 1.37 &  22.25   \\
\CondMat  & 3.04 & 4.7 & 5.39 & \textbf{7.31} & 6.76 & 7.05 & 0.65 & 1.15 &  17.84   & 2.71 & 4.61 & 4.77 & 7.92 & \textbf{8.46} & 8.35 & 0.65 & 0.98 &  15.40   \\
\Gplus  & 28.44 & 56.09 & 56.51 & \textbf{56.52} & 56.51 & 56.51 & 0.66 & 0.72 &  10.73   & 25.87 & 53.11 & 53.61 & 53.62 & \textbf{53.64} & 53.63 & 0.66 & 1.53 &  29.31   \\
\Brightkite  & 5.4 & 14.28 & 16.85 & \textbf{17.13} & 16.81 & 17.12 & 0.7 & - &  -   & 5.28 & 15.41 & 17.73 & 18.33 & 18.55 & \textbf{18.56} & 0.69 & - &  -   \\
\WikiTalk  & 13.15 & 25.45 & 28.31 & \textbf{28.32} & 27.97 & 28.26 & 0.74 & - &  -   & 12.09 & 23.02 & 25.35 & 25.51 & \textbf{25.69} & \textbf{25.69} & 0.74 & - &  -   \\
\bottomrule
\end{tabular}
\end{adjustbox}
\end{table*}

%% file: tikz/epsilon/alpha1.tex
\pgfplotsset{compat=1.3}
\begin{tikzpicture}

\definecolor{darkgray176}{RGB}{176,176,176}
\definecolor{green01270}{RGB}{0,127,0}

\begin{axis}[
tick pos=both,
axis lines=left,
x grid style={darkgray176},
xmin=0.04, xmax=0.26,
xtick style={color=black},
y grid style={darkgray176},
ymin=-0.0132134415500106, ymax=0.167680041387049,
ytick style={color=black},
yticklabel style = {xshift=-1ex},
legend pos= outer north east,
legend columns=1,
legend style={draw=none},
line width = 0.50 mm, 
tick style={line width=0.50mm},
ylabel={rel. increase polarization},
xlabel={$\epsilon$},
ytick scale label code/.code={$\times 10^{-2}$}
]
\addplot [black]
table {%
0.05 -0.00150354674069208
0.1 -0.00499101050741696
0.15 -0.00193440551660201
0.2 -0.00216421979118444
0.25 -4.51681235308303e-05
};
\addplot [red, dashed]
table {%
0.05 0.00183483433532726
0.1 0.00402366164020674
0.15 0.00712873072578184
0.2 0.0043399611627054
0.25 0.003464125948958
};
\addplot [blue, dotted]
table {%
0.05 0.0170610510956553
0.1 0.034780020333204
0.15 0.0508856653639139
0.2 0.0633712041020991
0.25 0.0792006936741834
};
\addplot [green01270, dash pattern=on 1pt off 3pt on 3pt off 3pt]
table {%
0.05 0.0325130227140441
0.1 0.063892375575521
0.15 0.0973675296817644
0.2 0.124912384461718
0.25 0.159457610344456
};
\addplot [black]
table {%
0.05 -0.000152811096278735
0.1 -0.00038046441164185
0.15 -0.000957380455787917
0.2 0.00112711389692672
0.25 0.000758960116832218
};

\addlegendentry{\MaxSum}
\addlegendentry{\MaxLinInt}
\addlegendentry{\MaxLinDis}
\addlegendentry{\MaxLinPol}
\addlegendentry{\MaxLinDisCon}

\end{axis}

\end{tikzpicture}

%% file: tikz/epsilon/alpha2.tex
\pgfplotsset{compat=1.3}
\begin{tikzpicture}

\definecolor{darkgray176}{RGB}{176,176,176}
\definecolor{green01270}{RGB}{0,127,0}

\begin{axis}[
tick pos=both,
axis lines=left,
x grid style={darkgray176},
xmin=0.04, xmax=0.26,
xtick style={color=black},
xtick = {0.05, 0.1, 0.15, 0.2, 0.25},
y grid style={darkgray176},
ymin=0.0130612900214665, ymax=0.157841545040412,
ytick style={color=black},
ytick = {0, 0.05, 0.1, 0.15},
yticklabel style = {xshift=-1ex},
legend pos= outer north east,
legend columns=1,
legend style={draw=none},
line width = 0.50 mm, 
tick style={line width=0.50mm},
ylabel={rel. increase polarization},
xlabel={$\epsilon$},
ytick scale label code/.code={$\times 10^{-2}$}
]
\addplot [black]
table {%
0.05 0.0196422107041459
0.1 0.037298196135675
0.15 0.0542923682664487
0.2 0.06261314584278
0.25 0.0786095685897592
};
\addplot [red, dashed]
table {%
0.05 0.025468871427417
0.1 0.0455394163248411
0.15 0.0645467184605504
0.2 0.078986442058507
0.25 0.092072625193864
};
\addplot [blue, dotted]
table {%
0.05 0.0317149201996606
0.1 0.0590459785490715
0.15 0.082347689311136
0.2 0.102048666437045
0.25 0.121172537738778
};
\addplot [green01270, dash pattern=on 1pt off 3pt on 3pt off 3pt]
table {%
0.05 0.0370309595182637
0.1 0.0712260242358546
0.15 0.102590709363027
0.2 0.130449496415652
0.25 0.151260624357733
};
\addplot [black]
table {%
0.05 0.0257338499496445
0.1 0.0471401689373272
0.15 0.0692092276787279
0.2 0.0884235948084136
0.25 0.0994291053086571
};

\addlegendentry{\MaxSum}
\addlegendentry{\MaxLinInt}
\addlegendentry{\MaxLinDis}
\addlegendentry{\MaxLinPol}
\addlegendentry{\MaxLinDisCon}

\end{axis}

\end{tikzpicture}

%% file: tables/exp-marketing-1.tex
\begin{table*}[ht!]
    \caption{Results for marketing campaigns with
        $k = \lceil 0.5\%\cdot n\rceil$ seeds,
		where we initialized the innate opinions using the exponential distribution.
        We report the relative increase of each index in percent.
        }
    \label{tab:exp-marketing-1}
    \centering
      \vspace{-4mm}
    \begin{adjustbox}{max width=\textwidth}
    \begin{tabular}{c cccccccc ccccccccc}
    \toprule
    \textbf{Dataset}  & \multicolumn{8}{c}{\textbf{\sumidx}} & \multicolumn{9}{c}{\textbf{\polidx}} \\
    \cmidrule(lr){2-9} 
    \cmidrule(lr){10-18} 
      &  \MaxSum  &  \MaxLinDisCon  &  \MaxLinPol  &  \MaxLinDis  &  \MaxLinInt  &   \MaxInfluence  &  \Random  &  \FJGreedy  &  \MaxSum  &  \MaxLinDisCon  &  \MaxLinPol  &  \MaxLinDis  &  \MaxLinInt  &  \MaxInfluence  &  \Random  & \FJGreedy  & \FJUpp \\
    \midrule
\Netscience  & 6.49 & 6.59 & 2.06 & 0.25 & 0.26 & \textbf{6.65} & 0.48 & 0.25 & 8.91 & 9.17 & \textbf{38.83} & 5.38 & 2.52 & 8.98 & 0.13 & 7.29 & 86.98  \\
\WikiVote  & \textbf{7.76} & 7.75 & 0.88 & 0.8 & 0.99 & 7.74 & 0.49 & 0.2 & 2.94 & 2.81 & \textbf{14.44} & 12.13 & 0.62 & 3.05 & -0.18 & 8.58 & 55.66  \\
\Reed  & 7.05 & \textbf{7.06} & 0.71 & 0.22 & 0.83 & 7.04 & 0.6 & 0.22 & 1.4 & 1.45 & 35.55 & 31.35 & 0.24 & 1.69 & -0.04 & \textbf{37.93} & 523.59  \\
\EmailUniv  & \textbf{6.65} & 6.56 & 0.57 & 0.59 & 2.5 & 6.64 & 0.59 & 0.21 & -0.33 & -0.37 & \textbf{12.31} & 9.35 & 0.53 & 0.35 & 0.01 & 11.44 & 68.31  \\
\Hamster  & 10.27 & 10.27 & 1.69 & 1.53 & 2.61 & \textbf{10.33} & 0.7 & 0.24 & 9.87 & 9.33 & \textbf{30.42} & 22.16 & 5.26 & 9.82 & 0.06 & 20.54 & 189.96  \\
\USFCA  & 6.04 & 5.96 & 0.5 & 0.63 & 0.87 & \textbf{6.05} & 0.53 & 0.2 & 0.3 & 0.66 & \textbf{31.83} & 26.1 & 1.6 & 0.44 & -0.07 & 29.95 & 396.04  \\
\NipsEgo  & \textbf{54.35} & \textbf{54.35} & 23.57 & 2.73 & 30.95 & \textbf{54.35} & 0.45 & 0.28 & -0.77 & -0.8 & \textbf{106.59} & 20.44 & 62.12 & -0.85 & 0.14 & 6.02 & 50.36  \\
\PagesGovernment  & \textbf{8.4} & 8.34 & 2.3 & 1.39 & 1.2 & \textbf{8.4} & 0.63 & 0.23 & 11.86 & 13.21 & \textbf{29.3} & 17.34 & 4.07 & 12.18 & 0.01 & 20.98 & 265.14  \\
\HepPh  & 7.94 & 7.87 & 1.72 & 2.21 & 1.43 & \textbf{7.96} & 0.81 & 0.31 & 11.03 & 11.11 & \textbf{39.53} & 22.75 & 5.05 & 11.84 & 0.11 & 20.14 & 202.7  \\
\Anybeat  & 37.35 & 37.34 & 21.51 & 14.98 & 3.91 & \textbf{37.39} & 0.58 & 0.3 & 28.39 & 28.32 & \textbf{53.95} & 22.19 & 5.5 & 28.37 & -0.0 & 11.02 & 89.58  \\
\CondMat  & 7.92 & 7.86 & 1.77 & 1.7 & 1.06 & \textbf{7.95} & 0.74 & 0.27 & 11.38 & 12.83 & \textbf{35.5} & 17.06 & 3.91 & 12.23 & -0.03 & 15.87 & 149.43  \\
\Gplus  & 59.58 & \textbf{59.59} & 31.24 & 14.03 & 16.9 & 59.58 & 0.94 & 0.32 & 1.08 & 1.07 & \textbf{88.71} & 59.38 & 59.24 & 1.09 & -0.01 & 8.97 & 79.93  \\
\Brightkite  & 19.83 & 19.79 & 4.97 & 5.06 & 2.86 & \textbf{19.84} & 0.93 &  - & 17.2 & 17.13 & \textbf{26.8} & 16.88 & 5.54 & 17.1 & 0.06 &  - &  -  \\
\WikiTalk  & \textbf{31.54} & 31.53 & 24.23 & 0.49 & 0.49 & \textbf{31.54} & 1.2 &  - & 31.41 & 32.19 & \textbf{47.26} & 9.05 & 7.45 & 31.54 & 0.27 &  - &  -  \\
    \bottomrule
    \end{tabular}
    \end{adjustbox}
    \end{table*}

%% file: tables/exp-marketing-2.tex
\begin{table*}[!htb]
    \caption{Results for marketing campaigns with
        $k = \lceil 0.5\%\cdot n\rceil$ seeds,
		where we initialized the innate opinions using the exponential distribution.
        We report the relative increase of each index in percent.
	}
    \label{tab:exp-marketing-2}
    \centering
      \vspace{-4mm}
    \begin{adjustbox}{max width=\textwidth}
    \begin{tabular}{c ccccccccc ccccccccc}
    \toprule
    \textbf{Dataset}  & \multicolumn{9}{c}{\textbf{\disidx}} & \multicolumn{9}{c}{\textbf{\intidx}} \\
    \cmidrule(lr){2-10} 
    \cmidrule(lr){11-19} 
     & \MaxSum & \MaxLinDisCon & \MaxLinPol & \MaxLinDis & \MaxLinInt &  \MaxInfluence & \Random & \FJGreedy & \FJUpp & \MaxSum & \MaxLinDisCon & \MaxLinPol & \MaxLinDis & \MaxLinInt &  \MaxInfluence & \Random & \FJGreedy &\FJUpp\\
    \midrule
\Netscience  & 0.58 & 0.61 & \textbf{5.38} & 5.24 & 3.1 & 0.51 & -0.03 & 5.24 & 53.49 & -1.19 & -1.21 & -1.97 & 1.92 & \textbf{4.27} & -1.2 & -0.06 & 7.69 & 43.62  \\
\WikiVote  & 1.23 & 1.18 & 7.49 & \textbf{7.64} & 1.72 & 1.29 & -0.12 & 5.79 & 40.77 & 0.42 & 0.44 & 2.16 & 2.48 & \textbf{4.47} & 0.34 & -0.12 & 5.85 & 45.9  \\
\Reed  & 1.55 & 1.38 & 9.3 & 10.3 & 1.46 & 1.52 & 0.06 & \textbf{11.82} & 124.5 & 1.14 & 1.0 & 0.73 & 1.17 & \textbf{3.44} & 1.02 & 0.09 & 4.7 & 30.48  \\
\EmailUniv  & -0.37 & -0.49 & 4.41 & 5.89 & 2.58 & 0.1 & 0.02 & \textbf{6.15} & 39.9 & 0.07 & -0.15 & 0.78 & 2.08 & \textbf{4.66} & 0.34 & 0.05 & 4.82 & 37.14  \\
\Hamster  & 3.9 & 4.27 & 8.49 & \textbf{9.7} & 4.57 & 3.9 & -0.0 & 8.7 & 69.09 & 1.54 & 1.78 & 1.75 & 2.69 & \textbf{4.43} & 1.54 & -0.06 & 5.89 & 42.2  \\
\USFCA  & 0.01 & 0.2 & 10.18 & \textbf{11.36} & 2.01 & 0.03 & 0.0 & 11.18 & 123.72 & -0.03 & 0.01 & 1.03 & 1.63 & \textbf{3.14} & 0.05 & -0.02 & 4.26 & 29.82  \\
\NipsEgo  & -0.4 & -0.47 & 4.25 & \textbf{7.45} & 4.0 & -0.54 & 0.03 & 5.95 & 34.55 & -0.11 & -0.24 & 2.81 & \textbf{5.62} & 3.07 & -0.39 & 0.03 & 6.7 & 132.74  \\
\PagesGovernment  & 3.58 & 4.26 & \textbf{9.43} & 9.22 & 3.41 & 3.84 & 0.01 & 8.88 & 83.73 & 1.05 & 1.16 & 1.58 & 2.9 & \textbf{4.22} & 1.12 & 0.01 & 5.46 & 38.0  \\
\HepPh  & 5.54 & 6.13 & 10.95 & \textbf{12.25} & 5.33 & 5.93 & 0.02 & 10.26 & 77.54 & 2.91 & 3.09 & 1.9 & 3.84 & \textbf{7.58} & 2.88 & -0.01 & 8.59 & 59.45  \\
\Anybeat  & 15.58 & 15.5 & \textbf{25.8} & 14.71 & 6.36 & 15.5 & 0.0 & 8.4 & 48.07 & 7.15 & 7.1 & 8.45 & 7.1 & \textbf{9.79} & 7.22 & 0.01 & 10.23 & 83.03  \\
\CondMat  & 5.03 & 5.04 & 8.38 & \textbf{9.25} & 4.64 & 5.02 & 0.0 & 7.91 & 61.22 & 2.52 & 2.42 & 1.29 & 3.68 & \textbf{5.96} & 2.48 & 0.01 & 7.15 & 54.05  \\
\Gplus  & 1.29 & 1.29 & \textbf{21.78} & 20.7 & 20.36 & 1.31 & 0.01 & 7.77 & 42.41 & 1.97 & 1.96 & 10.48 & 11.71 & \textbf{15.03} & 2.0 & 0.01 & 10.24 & 116.28  \\
\Brightkite  & 10.49 & 10.6 & 10.44 & \textbf{11.34} & 6.31 & 10.44 & 0.03 &  \_ &  \_ & 5.93 & 5.97 & 2.82 & 5.01 & \textbf{9.91} & 5.91 & -0.0 &  - &  -  \\
\WikiTalk  & 4.82 & 4.92 & \textbf{12.07} & 7.87 & 7.71 & 4.85 & -0.04 &  \_ &  \_ & 2.62 & 2.72 & 4.61 & 5.99 & \textbf{7.01} & 2.7 & -0.13 &  - &  -  \\
    \bottomrule
    \end{tabular}
    \end{adjustbox}
    \end{table*}

%% file: tables/exp-backfire-1.tex
\begin{table*}[ht!]
    \caption{Results for polarizing campaigns with
        $k = \lceil 0.5\%\cdot n\rceil$ seeds,
		where we initialized the innate opinions using the exponential distribution.
        We report the relative increase of each index in percent.}
    \label{tab:exp-backfire-1}
    \centering
      \vspace{-4mm}
    \begin{adjustbox}{max width=\textwidth}
    \begin{tabular}{c ccccccccc ccccccccc}
    \toprule
    \textbf{Dataset}  & \multicolumn{8}{c}{\textbf{\sumidx}} & \multicolumn{9}{c}{\textbf{\polidx}} \\
    \cmidrule(lr){2-9} 
    \cmidrule(lr){10-18} 
      &  \MaxSum  &  \MaxLinDisCon  &  \MaxLinPol  &  \MaxLinDis  &  \MaxLinInt  &   \MaxInfluence  &  \Random  &  \FJGreedy  &  \MaxSum  &  \MaxLinDisCon  &  \MaxLinPol  &  \MaxLinDis  &  \MaxLinInt  &  \MaxInfluence  &  \Random  & \FJGreedy  & \FJUpp \\
    \midrule
\Netscience  & \textbf{0.25} & \textbf{0.25} & -4.41 & -0.87 & -0.78 & -6.38 & -0.69 & \textbf{0.25} & 6.15 & 5.23 & \textbf{55.56} & 2.36 & 0.22 & 37.64 & 1.09 & 7.29 & 86.98  \\
\WikiVote  & \textbf{0.22} & 0.2 & -4.29 & -5.89 & -5.6 & -6.61 & -0.47 & 0.2 & 4.14 & 4.54 & \textbf{20.63} & 18.64 & 12.48 & 16.93 & 0.97 & 8.58 & 55.66  \\
\Reed  & 0.16 & \textbf{0.22} & -1.45 & -4.22 & -3.8 & -6.54 & -0.63 & \textbf{0.22} & 13.78 & 31.35 & \textbf{48.24} & 40.63 & 4.75 & 15.32 & 0.98 & 37.93 & 523.59  \\
\EmailUniv  & \textbf{0.21} & 0.2 & -3.39 & -5.22 & -5.19 & -5.98 & -0.51 & \textbf{0.21} & 2.2 & 6.87 & \textbf{13.28} & 11.62 & 8.31 & 9.79 & 0.8 & 11.44 & 68.31  \\
\Hamster  & 0.16 & 0.22 & -3.99 & -4.93 & -5.5 & -9.94 & -0.67 & \textbf{0.24} & 6.33 & 9.34 & \textbf{26.65} & 17.36 & 4.95 & 8.54 & 0.51 & 20.54 & 189.96  \\
\USFCA  & 0.14 & 0.19 & -0.91 & -2.8 & -4.3 & -5.39 & -0.48 & \textbf{0.2} & 6.65 & 13.99 & 28.37 & 23.41 & 9.95 & 8.48 & 0.71 & \textbf{29.95} & 396.04  \\
\NipsEgo  & \textbf{0.28} & \textbf{0.28} & -36.7 & -50.13 & -52.0 & -52.77 & -0.87 & \textbf{0.28} & 4.92 & 5.46 & \textbf{96.9} & 42.55 & 31.68 & 26.4 & 0.4 & 6.02 & 50.36  \\
\PagesGovernment  & 0.21 & 0.21 & -1.75 & -4.31 & -4.71 & -7.91 & -0.57 & \textbf{0.23} & 4.83 & 7.74 & \textbf{39.47} & 19.14 & 5.08 & 9.72 & 0.49 & 20.98 & 265.14  \\
\HepPh  & 0.17 & 0.18 & -1.8 & -2.87 & -2.53 & -7.72 & -0.81 & \textbf{0.31} & 5.94 & 6.19 & \textbf{32.43} & 18.04 & 2.7 & 10.21 & 0.18 & 20.14 & 202.7  \\
\Anybeat  & 0.28 & 0.29 & -30.31 & -30.03 & -29.7 & -36.32 & -1.01 & \textbf{0.3} & 6.46 & 7.93 & \textbf{55.25} & 50.77 & 48.3 & 39.56 & 0.27 & 11.02 & 89.58  \\
\CondMat  & 0.24 & 0.24 & -2.16 & -2.53 & -3.23 & -7.68 & -0.69 & \textbf{0.27} & 6.59 & 7.72 & \textbf{28.5} & 12.83 & 4.5 & 10.99 & 0.36 & 15.87 & 149.43  \\
\Gplus  & 0.3 & 0.31 & -26.97 & -41.32 & -38.25 & -57.37 & -0.79 & \textbf{0.32} & 7.53 & 8.02 & \textbf{120.0} & 74.17 & 70.28 & 16.68 & 0.2 & 8.97 & 79.93  \\
\Brightkite  & \textbf{0.24} & 0.23 & -9.12 & -11.47 & -10.31 & -19.24 & -0.84 &  \_ & 7.8 & 7.82 & \textbf{25.63} & 20.59 & 12.29 & 19.4 & 0.17 &  - &  -  \\
\WikiTalk  & \textbf{0.38} & \textbf{0.38} & -22.0 & -21.52 & -21.37 & -29.75 & -0.66 &  \_ & 7.23 & 8.11 & \textbf{60.33} & 51.67 & 51.64 & 47.0 & 0.76 &  - &  -  \\
    \bottomrule
    \end{tabular}
    \end{adjustbox}
    \end{table*}

%% file: tables/exp-backfire-2.tex
\begin{table*}[!htb]
    \caption{Results for polarizing campaigns with
        $k = \lceil 0.5\%\cdot n\rceil$ seeds,
		where we initialized the innate opinions using the exponential distribution.
        We report the relative increase of each index in percent.
        }
    \label{tab:exp-backfire-2}
    \centering
      \vspace{-4mm}
    \begin{adjustbox}{max width=\textwidth}
    \begin{tabular}{c ccccccccc ccccccccc}
    \toprule
    \textbf{Dataset}  & \multicolumn{9}{c}{\textbf{\disidx}} & \multicolumn{9}{c}{\textbf{\intidx}} \\
    \cmidrule(lr){2-10} 
    \cmidrule(lr){11-19} 
     & \MaxSum & \MaxLinDisCon & \MaxLinPol & \MaxLinDis & \MaxLinInt &  \MaxInfluence & \Random & \FJGreedy & \FJUpp & \MaxSum & \MaxLinDisCon & \MaxLinPol & \MaxLinDis & \MaxLinInt &  \MaxInfluence & \Random & \FJGreedy &\FJUpp\\
    \midrule
\Netscience  & 0.58 & 0.61 & \textbf{5.38} & 5.24 & 3.1 & 0.51 & -0.03 & 5.24 & 53.49 & -1.19 & -1.21 & -1.97 & 1.92 & \textbf{4.27} & -1.2 & -0.06 & 7.69 & 43.62  \\
\WikiVote  & 1.23 & 1.18 & 7.49 & \textbf{7.64} & 1.72 & 1.29 & -0.12 & 5.79 & 40.77 & 0.42 & 0.44 & 2.16 & 2.48 & \textbf{4.47} & 0.34 & -0.12 & 5.85 & 45.9  \\
\Reed  & 1.55 & 1.38 & 9.3 & 10.3 & 1.46 & 1.52 & 0.06 & \textbf{11.82} & 124.5 & 1.14 & 1.0 & 0.73 & 1.17 & \textbf{3.44} & 1.02 & 0.09 & 4.7 & 30.48  \\
\EmailUniv  & -0.37 & -0.49 & 4.41 & 5.89 & 2.58 & 0.1 & 0.02 & \textbf{6.15} & 39.9 & 0.07 & -0.15 & 0.78 & 2.08 & \textbf{4.66} & 0.34 & 0.05 & 4.82 & 37.14  \\
\Hamster  & 3.9 & 4.27 & 8.49 & \textbf{9.7} & 4.57 & 3.9 & -0.0 & 8.7 & 69.09 & 1.54 & 1.78 & 1.75 & 2.69 & \textbf{4.43} & 1.54 & -0.06 & 5.89 & 42.2  \\
\USFCA  & 0.01 & 0.2 & 10.18 & \textbf{11.36} & 2.01 & 0.03 & 0.0 & 11.18 & 123.72 & -0.03 & 0.01 & 1.03 & 1.63 & \textbf{3.14} & 0.05 & -0.02 & 4.26 & 29.82  \\
\NipsEgo  & -0.4 & -0.47 & 4.25 & \textbf{7.45} & 4.0 & -0.54 & 0.03 & 5.95 & 34.55 & -0.11 & -0.24 & 2.81 & \textbf{5.62} & 3.07 & -0.39 & 0.03 & 6.7 & 132.74  \\
\PagesGovernment  & 3.58 & 4.26 & \textbf{9.43} & 9.22 & 3.41 & 3.84 & 0.01 & 8.88 & 83.73 & 1.05 & 1.16 & 1.58 & 2.9 & \textbf{4.22} & 1.12 & 0.01 & 5.46 & 38.0  \\
\HepPh  & 5.54 & 6.13 & 10.95 & \textbf{12.25} & 5.33 & 5.93 & 0.02 & 10.26 & 77.54 & 2.91 & 3.09 & 1.9 & 3.84 & \textbf{7.58} & 2.88 & -0.01 & 8.59 & 59.45  \\
\Anybeat  & 15.58 & 15.5 & \textbf{25.8} & 14.71 & 6.36 & 15.5 & 0.0 & 8.4 & 48.07 & 7.15 & 7.1 & 8.45 & 7.1 & \textbf{9.79} & 7.22 & 0.01 & 10.23 & 83.03  \\
\CondMat  & 5.03 & 5.04 & 8.38 & \textbf{9.25} & 4.64 & 5.02 & 0.0 & 7.91 & 61.22 & 2.52 & 2.42 & 1.29 & 3.68 & \textbf{5.96} & 2.48 & 0.01 & 7.15 & 54.05  \\
\Gplus  & 1.29 & 1.29 & \textbf{21.78} & 20.7 & 20.36 & 1.31 & 0.01 & 7.77 & 42.41 & 1.97 & 1.96 & 10.48 & 11.71 & \textbf{15.03} & 2.0 & 0.01 & 10.24 & 116.28  \\
\Brightkite  & 10.49 & 10.6 & 10.44 & \textbf{11.34} & 6.31 & 10.44 & 0.03 &  \_ &  \_ & 5.93 & 5.97 & 2.82 & 5.01 & \textbf{9.91} & 5.91 & -0.0 &  - &  -  \\
\WikiTalk  & 4.82 & 4.92 & \textbf{12.07} & 7.87 & 7.71 & 4.85 & -0.04 &  \_ &  \_ & 2.62 & 2.72 & 4.61 & 5.99 & \textbf{7.01} & 2.7 & -0.13 &  - &  -  \\
    \bottomrule
    \end{tabular}
    \end{adjustbox}
    \end{table*}

%% file: tikz/greedy/f.tex
\pgfplotsset{compat=1.3}
\begin{tikzpicture}

\begin{axis}[
legend cell align={left},
legend style={
  fill opacity=0.8,
  draw opacity=1,
  text opacity=1,
  at={(0.03,0.97)},
  anchor=north west,
  draw=none
},
tick pos=both,
axis lines=left, xtick=\empty,
x grid style={white!69.0196078431373!black},
xmin=0.8, xmax=5.2,
xtick style={color=black},
y grid style={white!69.0196078431373!black},
ymin=0.0251678234807635, ymax=0.127350614006021,
ytick style={color=black},
ylabel={rel. increase dis--con},
xlabel={$k$},
xtick={1, 3, 5},
xticklabels={
  \(\displaystyle {1}\),
  \(\displaystyle {3}\),
  \(\displaystyle {5}\)
},
legend pos= south east,
legend style={draw=none},
line width = 0.80 mm, 
tick style={line width=0.70mm},
]
\addplot [black]
table {%
1 0.0308574186059966
3 0.0813089467594634
5 0.118576472208525
};
\addplot [red, dashed]
table {%
1 0.0332943292510187
3 0.0861334627317464
5 0.122705941709419
};
\addplot [blue, dotted]
table {%
1 0.0298124957773662
3 0.0788504647853582
5 0.112227670833381
};
\end{axis}

\end{tikzpicture}

%% file: tikz/greedy/g.tex
\pgfplotsset{compat=1.3}
\begin{tikzpicture}

\begin{axis}[
legend cell align={left},
legend style={
  fill opacity=0.8,
  draw opacity=1,
  text opacity=1,
  at={(0.03,0.97)},
  anchor=north west,
  draw=none
},
tick pos=both,
axis lines=left, xtick=\empty,
x grid style={white!69.0196078431373!black},
xmin=0.8, xmax=5.2,
xtick style={color=black},
y grid style={white!69.0196078431373!black},
ymin=0.0245612837862375, ymax=0.0927483005538264,
ytick style={color=black},
ylabel={rel. increase dis--con},
xlabel={$k$},
xtick={1, 3, 5},
xticklabels={
  \(\displaystyle {1}\),
  \(\displaystyle {3}\),
  \(\displaystyle {5}\)
},
legend pos= south east,
legend style={draw=none},
line width = 0.80 mm, 
tick style={line width=0.70mm},
ytick scale label code/.code={$\times 10^{-2}$},
]
\addplot [black]
table {%
1 0.0276606936393097
3 0.0661845429014457
5 0.0845205978361526
};
\addlegendentry{\SandDisCon}
\addplot [red, dashed]
table {%
1 0.0296071278555435
3 0.0700298981155137
5 0.0896488907007542
};
\addlegendentry{\UppDisCon}
\addplot [blue, dotted]
table {%
1 0.0277975018056531
3 0.0635403354086625
5 0.0831464016516068
};
\addlegendentry{\LowDisCon}
\end{axis}

\end{tikzpicture}